%% file: paper.tex
\documentclass[acmsmall, nonacm, screen]{acmart} %

\pdfoutput=1 %

\usepackage{tikz}
\usepackage{tikz-cd}
\usetikzlibrary{positioning,calc,math,arrows,automata,patterns}
\usepackage[customcolors]{hf-tikz}

\input{paper/settings}

\input{paper/graphs/measurementResults}

\input{paper/graphs/codeResults}

\title{\lang{}: A Tierless Language for Enforcing Contract-Client Protocols in Decentralized Applications}

\makeatletter\if@ACM@journal\makeatother
  \acmJournal{TOPLAS}
  \acmVolume{1}
  \acmNumber{ECOOP} %
  \acmArticle{1}
  \acmYear{2022}
  \acmMonth{1}
  \acmDOI{} %
  \startPage{1}
\else\makeatother
  \acmConference[PL'00]{ACM SIGPLAN Conference on Programming Languages}{January 01--03, 2000}{New York, NY, USA}
  \acmYear{2000}
  \acmISBN{} %
  \acmDOI{} %
  \startPage{1}
\fi

\setcopyright{none}             %

\author{David Richter}
\orcid{https://orcid.org/0000-0002-8672-0265}
\email{david.richter@tu-darmstadt.de}
\author{David Kretzler}
\orcid{https://orcid.org/0000-0002-6556-6457}
\email{david.kretzler@tu-darmstadt.de}
\affiliation{\institution{Technical University of Darmstadt}\country{Germany}}

\author{Pascal Weisenburger}
\orcid{https://orcid.org/0000-0003-1288-1485}
\email{pascal.weisenburger@unisg.ch}
\author{Guido Salvaneschi}
\orcid{https://orcid.org/0000-0002-9324-8894}
\email{guido.salvaneschi@unisg.ch}
\affiliation{\institution{University of St. Gallen}\country{Switzerland}}

\author{Sebastian Faust}
\orcid{https://orcid.org/0000-0002-8625-4639}
\email{sebastian.faust@tu-darmstadt.de}
\author{Mira Mezini}
\orcid{https://orcid.org/0000-0001-6563-7537}
\email{mezini@informatik.tu-darmstadt.de}
\affiliation{\institution{Technical University of Darmstadt}\country{Germany}}



\begin{document}

\begin{abstract}
\input{paper/abstract.tex}

\end{abstract}

\maketitle

\input{paper/introduction.tex}

\input{paper/contributions.tex}

\input{paper/design.tex}

\input{paper/core.tex}

\input{paper/formalisation.tex}

\input{paper/evaluation.tex}

\input{paper/discussion.tex}

\input{paper/relatedwork.tex}

\input{paper/conclusion.tex}

\begin{acks}
\input{paper/acknowledgements.tex}

\end{acks}

\bibliography{paper.bib}

\clearpage
\appendix
\label{part:Appendix}
\input{paper/case_studies.tex}
\input{paper/loc_evaluation.tex}

\clearpage
\input{paper/appendix.tex}

\end{document}

%% file: paper/settings.tex
\usepackage[scaled]{beramono}
\usepackage{enumitem}
\usepackage{listings}
\usepackage{upquote}
\usepackage{setspace}
\usepackage{xspace}
\usepackage[skip=10pt]{caption}
\usepackage{subcaption} %
\usepackage{wrapfig}
\usepackage{bcprules}
\usepackage{xspace}
\usepackage{multirow}
\usepackage[normalem]{ulem} %
\usepackage{mathpartir} %
\usepackage{booktabs}
\usepackage{amsmath}
\usepackage{amsfonts}
\usepackage{mathtools} %
\usepackage{stmaryrd}

\usepackage{array,longtable,colortbl} %
\newcolumntype{C}{>{$}c<{$}}  %
\newcolumntype{L}{>{$}l<{$}}  %

\usepackage{hyperref}
\usepackage{pgfplots}

\pgfplotsset{compat=1.3}
\usepgfplotslibrary{statistics}

\definecolor{ao(english)}{rgb}{0.0, 0.5, 0.0}

\newcommand{\lang}{\textsf{Prisma}\xspace}
\newcommand{\minilang}{\textsf{MiniPrisma}\xspace}

\makeatletter
\newcommand{\closeoverset}[3][0ex]{
  \ensuremath{\mathrel{\mathop{#3}\limits^{
    \vbox to#1{\kern-2\ex@
    \hbox{$\scriptstyle#2$}\vss}}}}}
\makeatother

\newcommand{\raiseoperator}[2]{
  \ensuremath{\mathrel{\raisebox{#1}{#2}}}}

\DeclareRobustCommand\longrightarrow{
  \relbar\joinrel\relbar\joinrel\relbar\joinrel%
  \relbar\joinrel\relbar\joinrel\relbar\joinrel%
  \rightarrow}

\newtheorem{globaltheorem}{Theorem}
\newtheorem{globallemma}{Lemma}

\renewcommand{\Finv}{a}

\newcommand{\compilestep}[2][]{\ensuremath{\operatorname{\mathit{#2}}\def\optarg{#1}\ifx\optarg\empty\else(#1)\fi}}
\newcommand{\assoc}[1][]{\compilestep[#1]{assoc}}
\newcommand{\mnfe}[1][]{\compilestep[#1]{mnf}}
\newcommand{\mnfp}[1][]{\compilestep[#1]{mnf^\prime}}
\newcommand{\dtle}[1][]{\compilestep[#1]{cps}}
\newcommand{\dtlp}[1][]{\compilestep[#1]{cps^\prime}}
\newcommand{\defun}[1][]{\compilestep[#1]{defun}}
\newcommand{\defunOuter}[1][]{\compilestep[#1]{defun^\prime}}
\newcommand{\guarde}[1][]{\compilestep[#1]{guard}}
\newcommand{\guardp}[1][]{\compilestep[#1]{guard^\prime}}
\newcommand{\coclfn}{\compilestep{coclfn}}
\newcommand{\fv}[1][]{\compilestep{fv}}

\newcommand{\msgevent}{\ensuremath{\mathsf{msg}}}
\newcommand{\wrevent}{\ensuremath{\mathsf{wr}}}

\renewcommand{\downarrow}{\ensuremath{\mathsf{awaitCl}}}

\newcommand{\try}{\ensuremath{\operatorname{\mathsf{try}}}}
\newcommand{\whovar}{\ensuremath{\operatorname{\mathsf{who}}}}
\newcommand{\statevar}{\ensuremath{\operatorname{\mathsf{state}}}}
\newcommand{\clfnvar}{\ensuremath{\operatorname{\mathsf{clfn}}}}
\newcommand{\cofnvar}{\ensuremath{\operatorname{\mathsf{cofn}}}}
\newcommand{\sendervar}{\ensuremath{\operatorname{\mathsf{sender}}}}
\newcommand{\freshvar}{\ensuremath{\operatorname{\mathsf{fresh}}}}
\newcommand{\trampoline}{\ensuremath{\operatorname{\mathsf{trmp}}}}

\newcommand{\clVal}{\ensuremath{\mathsf{@cl\,this.}}}
\newcommand{\coVal}{\ensuremath{\mathsf{@co\,this.}}}
\newcommand{\this}{\ensuremath{\mathsf{this}}}
\newcommand{\letexp}{}
\newcommand{\ifletexp}{\ensuremath{\operatorname{\mathsf{if\,let}}\,}}
\newcommand{\thenexp}{\ensuremath{\operatorname{\,\mathsf{then}}\,}}
\newcommand{\elseexp}{\ensuremath{\operatorname{\,\mathsf{else}}\,}}
\newcommand{\assert}{\ensuremath{\operatorname{\mathsf{assert}}}}
\newcommand{\bind}{\ensuremath{\,\operatorname{\mathsf{\gg\hspace{-2pt}=}}\,}}
\newcommand{\false}{\ensuremath{\operatorname{\mathsf{false}}}}
\newcommand{\true}{\ensuremath{\operatorname{\mathsf{true}}}}
\newcommand{\seq}{\ensuremath{;\;}}

\newcommand\iden{{i}}
\newcommand\decl{{d}}
\newcommand\trace{{p}}
\newcommand\syniden{{j}}
\newcommand\syndecl{{b}}
\newcommand\syntrace{{q}}
\newcommand\mainexp{{m}}
\newcommand\subexp{{e}}

\lstdefinelanguage{Scala}{
	morekeywords={
		abstract,case,catch,class,def,
		do,else,extends,false,final,finally,
		for,if,implicit,import,lazy,match,
		new,null,object,override,package,
		private,protected,return,sealed,
		super,this,throw,trait,true,try,
		type,val,var,while,with,yield,contract,awaitCl,balance, transfer, sender, value},
	morekeywords=[2]{
		@contract,@dApp,@co,@cl,@cross,@prisma},
	sensitive=true,
	morecomment=[l]{//},
	morecomment=[n]{/*}{*/},
	morecomment=[s][identifierstyle]{`}{`},
	morestring=[b]",
	morestring=[b]',
	morestring=[b]""",
	morestring=*[s]{s"""}{"""},
	morestring=[s][]{\$\{}{\}}
}

\lstdefinelanguage{Solidity}{
	morekeywords=[1]{anonymous, assembly, assert, balance, break, call, callcode, case, catch, class, constant, continue, constructor, contract, debugger, default, delegatecall, delete, do, else, emit, event, experimental, export, external, false, finally, for, function, gas, if, implements, import, in, indexed, instanceof, interface, internal, is, length, library, log0, log1, log2, log3, log4, memory, modifier, new, payable, pragma, private, protected, public, pure, push, require, return, returns, revert, sender, selfdestruct, send, solidity, storage, struct, suicide, super, switch, then, this, throw, transfer, true, try, typeof, using, value, view, while, with, addmod, ecrecover, keccak256, mulmod, ripemd160, sha256, sha3},
	morekeywords=[2]{address, bool, byte, bytes, bytes1, bytes2, bytes3, bytes4, bytes5, bytes6, bytes7, bytes8, bytes9, bytes10, bytes11, bytes12, bytes13, bytes14, bytes15, bytes16, bytes17, bytes18, bytes19, bytes20, bytes21, bytes22, bytes23, bytes24, bytes25, bytes26, bytes27, bytes28, bytes29, bytes30, bytes31, bytes32, enum, int, int8, int16, int24, int32, int40, int48, int56, int64, int72, int80, int88, int96, int104, int112, int120, int128, int136, int144, int152, int160, int168, int176, int184, int192, int200, int208, int216, int224, int232, int240, int248, int256, mapping, string, uint, uint8, uint16, uint24, uint32, uint40, uint48, uint56, uint64, uint72, uint80, uint88, uint96, uint104, uint112, uint120, uint128, uint136, uint144, uint152, uint160, uint168, uint176, uint184, uint192, uint200, uint208, uint216, uint224, uint232, uint240, uint248, uint256, var, void, ether, finney, szabo, wei, days, hours, minutes, seconds, weeks, years},	%
	morekeywords=[1]{asset, state, transaction}, %
	sensitive=true,
	morecomment=[l]{//},
	morecomment=[n]{/*}{*/},
	morecomment=[s][identifierstyle]{`}{`},
	morestring=[b]",
	morestring=[b]',
	morestring=[b]"""
}

\lstdefinelanguage{nomos}{
	morekeywords=[1]{proc,contract,type,transaction,if,then,else},
	morecomment=[n]{/*}{*/},
	morecomment=[l]{::},
}

\lstdefinelanguage{Haskell}{
	morekeywords=[1]{type,do,return,where,if,then,else},
	morecomment=[l]{::},
}

\lstdefinelanguage{JavaScript}{
	morekeywords=[1]{ break, case, catch, do, else, false, function, if, in, new, null, return, switch, true, typeof, var, while}
	morekeywords=[2]{class, export, implements, import, this, throw},
	sensitive=true,
	morecomment=[l]{//},
	morecomment=[n]{/*}{*/},
	morestring=[b]",
	morestring=[b]',
	morestring=[b]""",
	morestring=*[b]`,
	morestring=[s][]{\$\{}{\}}
}

\definecolor{dkgreen}{rgb}{0,0.6,0}
\definecolor{gray}{rgb}{0.5,0.5,0.5}
\definecolor{mauve}{rgb}{0,0.5,0.82}
\definecolor{light-gray}{gray}{0.25}
\lstdefinelanguage{theory}{
  numbers=none,
  basicstyle=\footnotesize\ttfamily,
  stringstyle=\color{mauve},
  moredelim=**[is][\color{black}]{\$\{}{\}},
  morestring=**[d]{'},
  morestring=[d]{\\'},
  morestring=**[d]{"},
}

\newcommand\figurehead{}

\makeatletter

\let\c@figure\c@table
\let\ftype@figure\ftype@table
\makeatother

\makeatletter %
\def\arcr{\@arraycr}
\makeatother

\renewcommand{\subparagraph}[1]{\paragraph{#1}}

%% file: paper/graphs/measurementResults.tex
\newcommand\varTTTChannelManualDeployAvg{1553.26}

\newcommand\varTTTChannelCompiledDeployAvg{1639.62}

\newcommand\varTTTChannelManualExecutionAvg{147.10}

\newcommand\varTTTChannelCompiledExecutionAvg{150.66}

\newcommand\varTTTChannelDiffExecRelAvg{3.50}
\newcommand\varTTTChannelDiffExecRelBoxLW{-14.54}

\newcommand\varTTTChannelDiffExecRelBoxUW{26.80}

\newcommand\varTTTChannelDiffDeployRelAvg{5.56}

\newcommand\varTTTManualDeployAvg{762.85}

\newcommand\varTTTCompiledDeployAvg{821.28}

\newcommand\varTTTManualExecutionAvg{52.82}

\newcommand\varTTTCompiledExecutionAvg{54.93}

\newcommand\varTTTDiffExecRelAvg{3.82}
\newcommand\varTTTDiffExecRelBoxLW{-16.65}

\newcommand\varTTTDiffExecRelBoxUW{7.94}

\newcommand\varTTTDiffDeployRelAvg{7.66}

\newcommand\varRPSManualDeployAvg{517.68}

\newcommand\varRPSCompiledDeployAvg{518.57}

\newcommand\varRPSManualExecutionAvg{48.49}

\newcommand\varRPSCompiledExecutionAvg{51.28}

\newcommand\varRPSDiffExecRelAvg{8.28}
\newcommand\varRPSDiffExecRelBoxLW{1.11}

\newcommand\varRPSDiffExecRelBoxUW{26.33}

\newcommand\varRPSDiffDeployRelAvg{0.17}

\newcommand\varTokenManualDeployAvg{464.82}

\newcommand\varTokenCompiledDeployAvg{466.34}

\newcommand\varTokenManualExecutionAvg{37.45}

\newcommand\varTokenCompiledExecutionAvg{37.51}

\newcommand\varTokenDiffExecRelAvg{0.17}
\newcommand\varTokenDiffExecRelBoxLW{0.05}

\newcommand\varTokenDiffExecRelBoxUW{0.25}

\newcommand\varTokenDiffDeployRelAvg{0.33}

\newcommand\varMultiSigManualDeployAvg{784.79}

\newcommand\varMultiSigCompiledDeployAvg{801.93}

\newcommand\varMultiSigManualExecutionAvg{80.44}

\newcommand\varMultiSigCompiledExecutionAvg{81.39}

\newcommand\varMultiSigDiffExecRelAvg{1.17}
\newcommand\varMultiSigDiffExecRelBoxLW{0.92}

\newcommand\varMultiSigDiffExecRelBoxUW{1.31}

\newcommand\varMultiSigDiffDeployRelAvg{2.19}

\newcommand\varEscrowManualDeployAvg{381.08}

\newcommand\varEscrowCompiledDeployAvg{389.77}

\newcommand\varEscrowManualExecutionAvg{25.07}

\newcommand\varEscrowCompiledExecutionAvg{26.49}

\newcommand\varEscrowDiffExecRelAvg{5.93}
\newcommand\varEscrowDiffExecRelBoxLW{4.02}

\newcommand\varEscrowDiffExecRelBoxUW{11.21}

\newcommand\varEscrowDiffDeployRelAvg{2.28}

\newcommand\varCrowdfundingManualDeployAvg{303.73}

\newcommand\varCrowdfundingCompiledDeployAvg{316.94}

\newcommand\varCrowdfundingManualExecutionAvg{35.91}

\newcommand\varCrowdfundingCompiledExecutionAvg{39.00}

\newcommand\varCrowdfundingDiffExecRelAvg{9.35}
\newcommand\varCrowdfundingDiffExecRelBoxLW{-23.21}

\newcommand\varCrowdfundingDiffExecRelBoxUW{16.66}

\newcommand\varCrowdfundingDiffDeployRelAvg{4.35}

\newcommand\varHangmanManualDeployAvg{1086.58}

\newcommand\varHangmanCompiledDeployAvg{1089.69}

\newcommand\varHangmanManualExecutionAvg{52.99}

\newcommand\varHangmanCompiledExecutionAvg{53.22}

\newcommand\varHangmanDiffExecRelAvg{0.57}
\newcommand\varHangmanDiffExecRelBoxLW{-6.27}

\newcommand\varHangmanDiffExecRelBoxUW{22.10}

\newcommand\varHangmanDiffDeployRelAvg{0.29}

\newcommand\varNotaryManualDeployAvg{314.38}

\newcommand\varNotaryCompiledDeployAvg{314.38}

\newcommand\varNotaryManualExecutionAvg{120.08}

\newcommand\varNotaryCompiledExecutionAvg{120.08}

\newcommand\varNotaryDiffExecRelAvg{0.00}
\newcommand\varNotaryDiffExecRelBoxLW{-0.01}

\newcommand\varNotaryDiffExecRelBoxUW{0.03}

\newcommand\varNotaryDiffDeployRelAvg{0.00}

\newcommand\varTTTLibraryManualDeployAvg{494.44}

\newcommand\varTTTLibraryCompiledDeployAvg{495.07}

\newcommand\varTTTLibraryDiffDeployRelAvg{0.13}

\newcommand\varTTTViaLibManualDeployAvg{1183.73}

\newcommand\varTTTViaLibCompiledDeployAvg{1258.36}

\newcommand\varTTTViaLibManualExecutionAvg{68.66}

\newcommand\varTTTViaLibCompiledExecutionAvg{70.89}

\newcommand\varTTTViaLibDiffExecRelAvg{3.10}
\newcommand\varTTTViaLibDiffExecRelBoxLW{-12.03}

\newcommand\varTTTViaLibDiffExecRelBoxUW{5.87}

\newcommand\varTTTViaLibDiffDeployRelAvg{6.30}

\newcommand\varChineseCheckersManualDeployAvg{1348.61}

\newcommand\varChineseCheckersCompiledDeployAvg{1388.48}

\newcommand\varChineseCheckersManualExecutionAvg{59.61}

\newcommand\varChineseCheckersCompiledExecutionAvg{64.90}

\newcommand\varChineseCheckersDiffExecRelAvg{9.03}
\newcommand\varChineseCheckersDiffExecRelBoxLW{-16.50}

\newcommand\varChineseCheckersDiffExecRelBoxUW{10.12}

\newcommand\varChineseCheckersDiffDeployRelAvg{2.96}

%% file: paper/graphs/codeResults.tex
\newcommand\varTTTChannelLoCPrisma{177}
\newcommand\varTTTChannelLoCSol{56}
\newcommand\varTTTChannelLoCJS{177}

\newcommand\varTTTChannelRC{9}
\newcommand\varTokenLoCPrisma{79}
\newcommand\varTokenLoCSol{48}
\newcommand\varTokenLoCJS{50}

\newcommand\varTokenRC{4}
\newcommand\varCrowdfundingLoCPrisma{59}
\newcommand\varCrowdfundingLoCSol{27}
\newcommand\varCrowdfundingLoCJS{63}

\newcommand\varCrowdfundingRC{11}
\newcommand\varEscrowLoCPrisma{63}
\newcommand\varEscrowLoCSol{33}
\newcommand\varEscrowLoCJS{56}

\newcommand\varEscrowRC{9}
\newcommand\varNotaryLoCPrisma{32}
\newcommand\varNotaryLoCSol{16}
\newcommand\varNotaryLoCJS{36}

\newcommand\varNotaryRC{3}
\newcommand\varHangmanLoCPrisma{119}
\newcommand\varHangmanLoCSol{86}
\newcommand\varHangmanLoCJS{83}

\newcommand\varHangmanRC{15}
\newcommand\varTTTLoCPrisma{61}
\newcommand\varTTTLoCSol{31}
\newcommand\varTTTLoCJS{52}

\newcommand\varTTTRC{5}
\newcommand\varRPSLoCPrisma{79}
\newcommand\varRPSLoCSol{41}
\newcommand\varRPSLoCJS{77}

\newcommand\varRPSRC{12}
\newcommand\varMultiSigLoCPrisma{76}
\newcommand\varMultiSigLoCSol{41}
\newcommand\varMultiSigLoCJS{52}

\newcommand\varMultiSigRC{3}
\newcommand\varTTTLibraryLoCPrisma{16}
\newcommand\varTTTLibraryLoCSol{15}
\newcommand\varTTTLibraryLoCJS{0}

\newcommand\varTTTViaLibLoCPrisma{53}
\newcommand\varTTTViaLibLoCSol{29}
\newcommand\varTTTViaLibLoCJS{52}

\newcommand\varTTTViaLibRC{5}
\newcommand\varChineseCheckersLoCPrisma{167}
\newcommand\varChineseCheckersLoCSol{141}
\newcommand\varChineseCheckersLoCJS{47}

\newcommand\varChineseCheckersRC{4}

%% file: paper/abstract.tex
Decentralized applications (dApps) consist of smart contracts that run on blockchains
and clients that model collaborating parties.
dApps are used to model financial and legal business functionality.
Today, contracts and clients are written
as separate programs -- in different programming languages -- communicating via 
send and receive operations.
This makes distributed program flow awkward to express and reason about,
increasing the potential for mismatches in the client-contract interface,
which can be exploited by malicious clients, potentially leading to huge financial losses.

In this paper, we present \lang{}, a language for tierless decentralized applications,
where the contract and its clients are defined in one unit and 
pairs of send and receive actions that “belong together” 
are encapsulated into a single direct-style operation, 
which is executed differently by sending and receiving parties.
This enables expressing distributed program flow via standard control flow
and renders mismatching communication impossible.
We prove formally that our compiler preserves program behavior
in presence of an attacker controlling the client code.
We systematically compare \lang{} with mainstream and advanced programming models
for dApps and provide empirical evidence for its expressiveness and performance.

%% file: paper/introduction.tex
\section{Introduction}
\label{sec:introduction}

dApps enable multiple parties sharing state to jointly execute functionality according to a predefined agreement.
This predefined agreement is called a \emph{smart contract} and regulates the interaction between the dApp's clients.
Such client--contract interactions can be logically described by state machines~\cite{SergeyNJ0TH19,MavridouL18,MavridouLSD19,SuvorovU19} specifying which party is allowed to do what and when.

dApps can operate without centralized trusted intermediaries by relying on a blockchain and its consensus protocol.
To this end, a contract is deployed to and executed on the blockchain, which guarantees its correct execution; clients that run outside of the blockchain can interact with the contract via transactions.
A key feature of dApps is that they can directly link application logic with transfer of monetary assets.
This enables a wide range of correctness/security-sensitive
business applications, e.g., for cryptocurrencies, crowdfunding, and public offerings,\!\footnote{
  700\,K to 2.7\,M contracts have been deployed
  per month between July 2020 and June 2021~\cite{google-cloud-big-query-contracts-per-month}
  on the Ethereum blockchain -- the most popular dApps platform~\cite{dapp-q2-2020-report}.
  Some dApps manage tremendous amounts of assets, e.g.,
  Uniswap \cite{uniswap-info} -- the largest Ethereum trading platform
  had a daily trading volume of 0.5\,B\,--\,1.5\,B~USD in June 2021.}
and the same feature makes them an attractive target for attackers.
The attack surface is wide since contracts can be called by any client in the network, including malicious ones that try to force the contract to deviate from the intended behavior~\cite{daoExploit}.
In recent years, there have been several large attacks exploiting flawed program flow control in smart contracts.
Most famously, attackers managed to steal around 50\,M~USD~\cite{daoExploitAnalysis,daoExploit} from a decentralized autonomous organization, the DAO.
In two attacks on the Parity multi-signature wallet, attackers stole  cryptocurrencies worth 30\,M~USD~\cite{parityBug}
and froze 150\,M~USD~\cite{parity2}.

\subparagraph{Programming dApps.} 
In this paper, we explore a programming model
that ensures the correctness and security of 
the client--contract interaction of dApps by-design.
Deviations from the intended
interaction protocols due to implementation errors
and/or malicious attacks are a critical threat 
(besides other issues such as arithmetic or buffer overflows, etc.) as demonstrated e.g., by the DAO attack~\cite{daoExploitAnalysis,daoExploit} mentioned above.

dApps are multi-party applications. For such applications,
there are two options for the programming model: a \emph{local} and a \emph{global model}.
In a \emph{local model}, parties are defined each in a
separate \emph{local} program and their interactions are encoded via 
effectful send and receive instructions. Approaches that follow this model stem from process calculi~\cite{Hoare78}
 and include actor systems~\cite{Agha86} and approaches using
  session types~\cite{sessiontypes09}, and
  linear logic~\cite{Wadler12}.
In contrast, in a \emph{global model}, there is 
a single program shared by both parties and interactions are
encoded via combined send-and-receive operations with no effects visible to the outside world.
 This model is represented by tierless~\cite{Queinnec00,
  Cooper:2006:LWP:1777707.1777724, Choi19, Fowler:2019, Serrano:2006,
  Serrano:2016, Radanne:2016, Weisenburger:2018}
and choreographic~\cite{
  HondaMBCY11, MontesiGZ14, giallorenzo2020choreographies} languages.
The local model requires an explicitly specified protocol
to ensure that every send effect has a corresponding receive operation 
in an interacting -- separately defined -- process.
With a global model, there is no need to 
separately specify such a protocol.
All parties run the same program
in lock-step, where
a single send-and-receive operation
performs a send when executed by one party and
a receive by the other party.
Due to encapsulating communication effects,
there is no non-local information to track --  
the program's control flow defines the correct interaction
and a simple %
type system is sufficient. 

Current approaches to dApp programming --
industrial or research ones -- follow a local model,
Contract and client are implemented in separate programs,
thus safety relies on explicitly specifying the client--contract interaction protocol.
Moreover, the contract and clients are implemented in different languages,
hence, developers have to master two technology stacks.

The dominating approach in industry uses
Solidity~\cite{ethereumPopularity} for the contract and JavaScript for clients.
Solidity relies on developers following best practices recommending to express the protocol 
as runtime assertions integrated into the contract code~\cite{solidity-best-practises}. 
Failing to correctly introduce assertions may give parties illegal access 
to monetary values to the detriment of others~\cite{10.1145/3274694.3274743,10.1145/2976749.2978309}.

The currently dominant encoding
style of the protocol as \emph{finite state machine (FSM)}
uses one contract-side function per FSM transition~\cite{ethereumPopularity, Coblenz17, SchransED18, Schrans19, Coblenz19, Coblenz20}.
While FSMs model a useful class of programs that can be efficiently verified,
writing programs in such style directly has several shortcomings.
First, an FSM corresponds to a control-flow graph of basic blocks, which
is low-level and more suited as an internal compiler representation than as a front-end
language for humans.
Second, with the FSM style, the contract is a passive entity whose execution is driven by clients.
This design puts the burden of enforcing the protocol on the programmers of the contract,
as they have to explicitly consider in what state which messages are valid
and reject all invalid messages from the clients.
Otherwise, malicious clients would be able to force the contract to deviate from its intended behavior
by sending messages that are invalid in the current state.
Third, ensuring protocol compliance statically to guarantee safety requires advanced types,
as the type of the next action depends on the current state.

In research, some smart contract languages~\cite{DasBHP19,
Coblenz17, SchransED18, Schrans19, Coblenz19, Coblenz20, psamathe,
BlackshearCDGMNPQRSZZ19}
have been proposed to overcome the FSM-style shortcomings.
They rely on advanced type systems such as
session types, type states, and linear types.
There, processes are typed by the protocol
(of side-effects such as sending and receiving)
that they follow
and non-compliant processes are rejected by the type-checker.

The global model has not been explored for dApp programming -- which is unfortunate given
the potential to get by with a standard typing discipline and 
to avoid intricacies and potential mismatches of a two-language stack.
Our work fills this gap by proposing
\lang{} -- the first language that
features a \emph{global programming model} for Ethereum dApps.
While we focus on the Ethereum blockchain, we believe our techniques to be applicable to other smart contract platforms as well.

\subparagraph{\lang.}
\lang enables interleaving contract and client logic within the same program 
and adopts a \emph{direct style (DS)} notation for encoding
send-and-receive operations akin to languages with baked-in support for asynchronous interactions,
e.g., via async/await \cite{DBLP:conf/ecoop/BiermanRMMT12,scala-async}.
\lang leaves it to the compiler to map down high-level declarative 
DS to low-level FSM style. It avoids the need for advanced typing discipline 
and allows the contract to actively ask clients for input, promoting
an execution model
where a dominant acting role controls
the execution and diverts control to other parties
when their input is needed, which matches well the 
dApp setting. 

Overall, \lang relieves the developer from the responsibility of
correctly managing distributed, asynchronous program flows and 
the heterogeneous technology stack.
Instead, the burden is on the compiler, which distributes the program flow 
by means of selective continuation-passing-style (CPS) translation
and defunctionalisation, as well as inserts guards against malicious client interactions.

For this, we needed to develop a CPS translation for the code that runs on
the Ethereum Virtual Machine (EVM), since the EVM has no built-in support
for concurrency primitives to suspend execution and resume later --
which could be used, otherwise, to implement asynchronous communication.
Given that CPS translations reify control flow,
without proper guarding,
malicious clients could force the contract to deviate from the intended flow
by passing a spoofed value to the contract.
Thus, it is imperative to prove that our \emph{distributed CPS translation}
ensures control-flow integrity of the contract, which we do on top of a formal definition of the compilation steps.
The formally proven secure \lang compiler eliminates the risk of programmers implementing
unsafe interactions that can potentially be exploited.

%% file: paper/contributions.tex
\subparagraph{Contributions.}
We make the following contributions:
\begin{enumerate}[noitemsep]
\item
  We introduce \lang{},\!\footnote{
    \lang{} implementation and case studies are publicly available:
    \url{https://github.com/stg-tud/prisma}
  } a global language for tierless dApps
  with direct-style client--contract interactions and explicit access control,
  implemented as an embedded DSL in Scala.
  Crucially, \lang{} automatically enforces the correct program flow
  (Section~\ref{sec:the-language}).

\item
  A core calculus, Mini\lang{}, which formalizes both \lang{} and its compiler, 
  as well as a proof that our compiler guarantees the preservation of control flow
  in presence of an attacker that controls the client code (Section~\ref{sec:formalization}).

\item
  Case studies which show that \lang{} can be used to implement common applications
  without prohibitive performance overhead (Section~\ref{sec:discussion}).

\item
A comparison of \lang{} with a session type and a type state smart contract programming language %
and the mainstream Solidity/JavaScript programming model
(Section~\ref{sec:background}).
\end{enumerate}

%% file: paper/design.tex
\section{\lang{} in a Nutshell}
\label{sec:the-language}

We present \lang{} by the example of a TicTacToe game, demonstrating
that client and contract are written in a single language,
where protocols are expressed by control flow (instead of relying on advanced typing disciplines)
and enforced by the compiler.

\subparagraph{Example.}
TicTacToe is a two-player game over a $3\times3$ board. 
Players take turns in writing their sign into one of the free fields until all fields
are occupied, or one player wins by owning three fields in any row, column, or diagonal.
The main transaction of a TicTacToe dApp is \lstinline!Move(x,y)!
used by a player to occupy field \lstinline!(x,y)!.
A \lstinline!Move(x,y)! is valid if it is the sender's turn and \lstinline!(x,y)! is empty.
Before the game, players deposit their stakes, and after the game, the stakes are paid to the winner.

\input{paper/graphs/prismaFigures.tex}

Fig.~\ref{fig:funding-diagram} depicts possible control flows
with transitions labeled by client actions that trigger them.
Black arrows depict intended control flows.
The dApp starts in the funding state
where both parties deposit  stakes via \textit{Fund$(c)$}.
Next,
  parties execute \textit{Move$(x,y)$}
  until one party wins or the game ends in a draw.
Finally, any party can invoke a payout of the stakes via \textit{Payout$()$}.\footnote{
  We omit handling timeouts on funding and execution for brevity.
}
Red dashed arrows illustrate the effects of a mismanaged control flow:
a malicious player could trigger a premature payout 
preventing the counterpart to get financial gains.

%% file: paper/graphs/prismaFigures.tex
\begin{figure}
	\begin{minipage}{.5\textwidth}
		\centering
		\input{paper/graphs/GraphTTTFSM}
		\captionof{figure}{TicTacToe control flow. \label{fig:funding-diagram}}
	\end{minipage}%
	\begin{minipage}{.5\textwidth}
		\centering
		\footnotesize
		\captionof{table}{Location annotations. \label{tab:prisma-annotations}}
		\input{paper/graphs/TablePrismaAnnotations}
	\end{minipage}%
	\vspace{1em}%
\end{figure}
\begin{figure}
	\begin{minipage}{0.95\textwidth}
		\begin{minipage}[t]{0.45\textwidth}
			\input{paper/graphs/CodePrismaModules}
		\end{minipage}%
		\begin{minipage}[t]{0.55\textwidth}
			\input{paper/graphs/CodeClientExpressions}
		\end{minipage}
		\captionof{figure}{TicTacToe dApp.\hspace{4em} \label{fig:portable-class}}
	\end{minipage}
\end{figure}

%% file: paper/graphs/GraphTTTFSM.tex
\begin{tikzpicture}[->,shorten >=2pt,node distance=1.75cm,inner sep=0,semithick,font=\tiny, scale=.5, /tikz/initial text = ]
\node[state, initial]    (A) []                     {T0};
\node[state]             (B) [right of=A]           {T1};
\node[state]             (C) [below of=B]           {T2};
\node[state, accepting]  (D) [left of=C]            {T3};
\path
(A) edge[loop above]  node[above=3pt] {Fund$(c)$}   ()
(A) edge[above, bend left]       node[above=3pt] {Fund$(c)$}   (B)
(B) edge[loop above]  node[above=3pt] {Move$(x,y)$} ()
(B) edge[right, bend left]       node[draw, draw = white, left=-10pt, fill=white] {Move$(x,y)$} (C)
(C) edge[below, bend left]       node[below=3pt] {Payout$()$}  (D)
(A) edge[red, dashed, bend right] node[draw, draw = white, left=-15pt, fill=white]  {Payout$()$}  (D)
(B) edge[red, dashed] node[draw, draw = white, align=center, fill=white]  {Payout$()$}  (D);
\end{tikzpicture}

%% file: paper/graphs/TablePrismaAnnotations.tex
\begin{tabular}{rll}
		\toprule
		Annotations & Description \\
		\midrule
		\lstinline!@co!        & on contract \\
		\lstinline!@cl!        & on clients \\
		\lstinline!@co @cl!    & independent copies \\
		& on clients and contract \\
		\lstinline!@co @cross! & on contract, but also \\
		& accessible by client \\
		\lstinline!@cl @cross! & (illegal combination) \\
		\bottomrule
	\end{tabular}

%% file: paper/graphs/CodePrismaModules.tex
\begin{lstlisting}
@prisma object TicTacToeModule {(*@ \label{code:prisma-module} @*)

  @co @cl case class UU(x: U8, y: U8)(*@ \label{code:portable-class} @*)

  class TicTacToe((*@ \label{code:ttt-snippet-contract-expression} @*)
    val players: Arr[Address],
    val fundingGoal: Uint) {

    // u8 is an unsigned 8-bit integer
    @co @cross var moves: U8 = "0".u8(*@ \label{code:snippet-moves-variable} @*)
    @co @cross var winner: U8 = "0".u8(*@ \label{code:snippet-winner-variable} @*)
    @co @cross val board: Arr[Arr[U8]] =
      Arr.ofDim("3".u, "3".u)(*@ \label{code:snippet-board-variable} @*)

    @co def performMove(x: U8, y: U8): Unit =(*@ \label{code:snippet-move-function} @*)
      { /* ... */ }
    @cl def updateBoard(): Unit =(*@ \label{code:snippet-move-executed-function} @*)
      { /* ... */ }
    @cl def fund(): (U256, Unit) =(*@ \label{code:snippet-continue-funding} @*)
      (readLine("How much?").u, ())
    @cl def move(): (U256, UU) =(*@ \label{code:snippet-continue-executing} @*)
      ("0".u, UU(readLine("x-pos?"),
                 readLine("y-pos?"))
\end{lstlisting}

%% file: paper/graphs/CodeClientExpressions.tex
\begin{lstlisting}[firstnumber=24]
    @cl def payout(): (U256, Unit) = {
      readLine("Press (enter) for payout")
      ("0".u, ())
    }
    @co val init: Unit = { (*@\label{code:snippet-await-init}@*)
      while (balance() < FUNDING_GOAL) { (*@\label{code:snippet-loop-condition}@*)
        awaitCl(_ => true) { fund() } (*@\label{code:snippet-await-funding}@*)
      }
      while (moves < "9".u && winner == "0".u) {
        val pair: UU = awaitCl(a => (*@\label{code:snippet-await-move}@*)
          a == players(moves %
        performMove(pair.x, pair.y)
      }
      awaitCl(a => true) { payout() } (*@\label{code:snippet-await-payout}@*)
      if (winner != "0".u) {
        players(winner - "1".u).transfer(balance())
      } else {
        players("0".u).transfer(balance() / "2".u)
        players("1".u).transfer(balance()) // remainder
      }
    }
  }
}
\end{lstlisting}

%% file: paper/core.tex
\subparagraph{Tierless dApps.}
\lang is implemented as a DSL embedded into Scala,
and \lang programs are also valid Scala programs.\footnote{
  In Scala \lstinline!val!/\lstinline!var! definitions are used for mutable/immutable fields and variables,
  \lstinline!def! for methods, \lstinline!class! for classes, and \lstinline!object! for singletons.
  A \lstinline!case class! is a class whose instances are compared by structure and not by reference.}
\lang interleaves contract and client logic within the same program. 
Annotations \lstinline!@co! and \lstinline!@cl! explicitly place 
declarations on the contract and on the client, respectively (cf. Tab.~\ref{tab:prisma-annotations}).
A declaration marked as both \lstinline!@co! and \lstinline!@cl!
has two copies.
For security, code placed in one location cannot access 
definitions from the other --- an attempt to do so
yields a
compile-time error.
Developers can overrule this constraint to enable clients to 
read contract variables or call contract functions
by combining \lstinline!@co! with \lstinline!@cross!.
Combining \lstinline!@cl! 
with \lstinline!@cross! is not allowed
-- information can only flow from client to contract
as part of a client--contract interaction protocol.

There are three kinds of classes.
\emph{Located classes} are placed in one location (annotated with either \lstinline!@co! or \lstinline!@cl!);
they cannot contain located members (annotated with either \lstinline!@co! or \lstinline!@cl!)
and their instances cannot cross the client--contract boundary, e.g., be passed to or returned from \lstinline!@cross! functions.
\emph{Portable classes} are annotated with both \lstinline!@co! and \lstinline!@cl!.
Their instances can be passed to and returned from \lstinline!@cross! functions;
they must not contain mutable fields.
\emph{Split classes} have no location annotation; 
their instances live partly in both locations; 
they cannot be passed to or returned from \lstinline!@cross! functions
and their members must be located.

\lang{} code is grouped into modules.
While client declarations can use and be used from standard (non-\lang{}) Scala code,
contract declarations are not accessible from Scala,
and can only reference contract code from other \lang{} modules
(because contract/client code lives in different VMs).

For illustration, consider the TicTacToe dApp (Fig.~\ref{fig:portable-class}).
The \lstinline!TicTacToeModule!
(Line~\ref{code:prisma-module})
-- modules are called \lstinline!object! in Scala --
contains a portable class \lstinline!UU! (Line~\ref{code:portable-class}) and a split class \lstinline!TicTacToe! (Line~\ref{code:ttt-snippet-contract-expression}).
Variables \lstinline!moves!, \lstinline!winner!,  \lstinline!board!
(Lines~\ref{code:snippet-moves-variable}, \ref{code:snippet-winner-variable}, \ref{code:snippet-board-variable})
are placed on the contract and can be read by clients (\lstinline!@co @cross!).
The \lstinline!updateBoard! function (Line~\ref{code:snippet-move-executed-function})
is placed on the client and updates client state (e.g., client's UI).
The \lstinline!move! function (Line~\ref{code:snippet-move-function})
is placed on the contract and changes the game state (\lstinline!move!).
\lstinline!move! is not annotated with \lstinline!@cross!,
because \lstinline!@cross! is intended for functions that do not change contract state
and can be executed out-of-order
without tampering with the client--contract interaction protocol.
While Scala only has signed integers and signed longs literals,
these are uncommon in Ethereum.
Therefore, \lang provides portable unsigned and signed integers
for power-of-two bitsizes between $2^3$ to $2^8$,
with common arithmetic operations,
e.g., \lstinline!"0".u8! is an unsigned 8-bit integer of value 0
(Line \ref{code:snippet-moves-variable}).

\subparagraph{Encoding client--contract protocols.}
In \lang, a client-contract protocol is encoded as a split class 
containing dedicated \lstinline!awaitCl! expressions
for actively requesting and awaiting messages from specific clients
and standard control-flow constructs.  
Hence, creating a new contract instance corresponds to creating a new instance of a protocol;
once created, the contract instance actively triggers interactions with clients. 
The \lstinline!awaitCl! expressions have the following syntax: 

\lstinline!def awaitCl[T](who: Addr => Bool)(body: => (Ether, T)): T!

They take two arguments.
The first (\lstinline!who!) is a predicate used by clients to decide whether it is their turn
and by the contract to decide whether to accept a message from a client.
This is unlike Solidity, 
where a function may be called by any party by default.
By forcing developers to explicitly define access control, 
\lang{} reduces the risk of human failure.
The second argument (\lstinline!body!) is the expression to be executed by the client. 
The client returns a pair of values to the contract:
the amount of Ether and the message.
The former can be accessed by the contract via the built-in expression \lstinline!value!,
the latter is returned by \lstinline!awaitCl!.
Besides receiving funds via \lstinline!awaitCl!,
a contract can also check its current balance (\lstinline!balance()!),
and transfer funds to an account (\lstinline!account.transfer(amount)!)).

\lang{}'s programming model is specifically designed to accommodate blockchain use cases.
In contrast to other tierless models like client-server,
we emphasise inversion of control such that the code is written as if the 
contract was the active driver of the protocol,
while clients are passive and only react to requests by the contract.
This enables to enforce the protocol on the contract side.
For this reason, for example, we support the \lstinline!awaitCl! construct on the active contract side
whereas there is no corresponding construct on the passive client side.

For illustration,
consider the definition of \lstinline!init!
on the right-hand side of Fig.~\ref{fig:portable-class},
Line~\ref{code:snippet-await-init}. 
It defines the protocol of \lstinline!TicTacToe! as follows.
From the beginning of \lstinline!init!, the flow reaches 
\lstinline!awaitCl! in Line~\ref{code:snippet-await-funding}
where the contract waits for 
clients to provide funding (by calling \lstinline!fund!).
Next, the contract continues until \lstinline!awaitCl!
in Line~\ref{code:snippet-await-move}
and clients execute \lstinline!move!
(Line~\ref{code:snippet-continue-executing}) 
until the game ends with a winner (\lstinline|winner != 0|)
or a draw (\lstinline!moves >= 9!).
At this point -- \lstinline!awaitCl!
in Line~\ref{code:snippet-await-payout} -- any party can request a payout and
the contract executes to the end.
The example illustrates how direct-style \lstinline!awaitCl! expressions and the
tierless model enable encoding multiparty protocols as standard control flow,
with protocol phases corresponding to basic blocks 
between \lstinline!awaitCl! expressions.

\begin{figure}\figurehead
  \begin{subfigure}{\textwidth}\centering
\begin{tikzcd}
      & Prisma \arrow[d]                   &          \\
      & {[Compiler]} \arrow[rd] \arrow[ld] &          \\
Scala &                                    & Solidity
\end{tikzcd}
  \subcaption{Compilation scheme}\label{fig:compilation-scheme}
  \end{subfigure}
  \input{paper/listings/solidity_ttt.tex}
  \caption{TicTacToe dApp after compilation, simplified}
  \vspace{-1.5em}
\end{figure}

\subparagraph{Compiling \lang to Solidity.}
Abstractly, \lang's compiler takes a \lang{} dApp program and splits it into two separate programs: A Scala client program and a Solidity contract program (Fig.~\ref{fig:compilation-scheme}).
In more detail, the compiler
(1)~places all definitions according to their annotations and
(2)~splits contract methods that contain \lstinline!awaitCl! expressions into a method that contains the code up to the \lstinline!awaitCl!
and a method that contains the continuation after the \lstinline!awaitCl!
(taking the result of the \lstinline!awaitCl! as an argument).
Once deployed, a contract is public and can be messaged by arbitrary clients --
not exclusively the ones generated by \lang{} -- hence, we cannot assume that clients will actually execute the body
passed to them by an \lstinline!awaitCl! expression.
To cope with malicious clients trying to tamper with the
control flow of the contract, the compiler hardens contract code
by generating code to enforce \emph{control flow integrity}:
storing the current phase before giving control to the client and
rejecting methods invoked by wrong clients or in the wrong phase.

For illustration, the code generated from Fig.~\ref{fig:portable-class}
is schematically shown in Fig.~\ref{fig:scala-client} and \ref{fig:solidity}.
The methods \lstinline!updateBoard!, \lstinline!fund!, \lstinline!move!, and \lstinline!payout!
are annotated \lstinline!@cl! and thus compiled into the client program (Fig.~\ref{fig:scala-client}).
The variables \lstinline!moves!, \lstinline!winner! and \lstinline!board!,
and the method \lstinline!performMove! are annotated \lstinline!@co!
and thus compiled into the contract program (Fig.~\ref{fig:solidity}).
Further, three new methods are generated on both the client and the contract
-- one for each \lstinline!awaitCl! expression in \lstinline!init! --
corresponding to phases in the logical protocol (Fig.~\ref{fig:funding-diagram}).
The \lstinline!Funding! method of the client (Line~\ref{code:client-funding-phase}) is generated from the body of the first \lstinline!awaitCl!.
Similarly,
the \lstinline!Move! method (Line~\ref{code:client-move-phase}) is generated from the second \lstinline!awaitCl!
and the \lstinline!Payout! method (Line~\ref{code:client-payout-phase}) from the third \lstinline!awaitCl!.
In the example,
  the generated methods are given meaningful names by capitalizing
  the single method called in the body of the \lstinline!awaitCl! expressions form which they were generated.
  In the actual implementation, generated methods are simply enumerated.
The code up to the first \lstinline!awaitCl! (Line~\ref{code:snippet-await-funding})
is placed in the constructor of the generated contract,
which ends by setting the active phase to \lstinline!Funding!.
The code between the first and the second \lstinline!awaitCl! either loops back to 
the first \lstinline!awaitCl! or continues to the second one (Line~\ref{code:snippet-await-move}).
The code is placed in the \lstinline!Fund! method that requires the phase to be \lstinline!Funding!,
and may change it to \lstinline!Exec! if the loop condition
fails.
Similarly, the method \lstinline!Move! is generated to contain the loop
between the second and the third \lstinline!awaitCl! (Line~\ref{code:snippet-await-payout});
and the method \lstinline!Payout! contains the code from the third \lstinline!awaitCl!
to the end of \lstinline!init!.
Only the second \lstinline!awaitCl! contains a (non-trivial) access control predicate,
which results in an additional assertion in the body of \lstinline!Move!
(Line \ref{code:solidity-access-control}).
Observe that the return types of the generated client methods
are the argument types of the corresponding contract methods.

\subparagraph{Compilation Techniques.}

While CPS is a key step in our translation pipeline,
the example shows the final defunctionalised, trampolined code.
The final output does not contain explicit continuations (i.e., a function that takes another function as an argument and calls that as its continuation).
Instead, after defunctionalizing and trampolinizing the CPS translation, only one top-level function (\lstinline!Fund!, \lstinline!Move!, \lstinline!Payout!) is callable at each \lstinline!phase!, which is ensured by the \lstinline!require! statement at the beginning of each function, and each function sets the next \lstinline!phase! at the end.
These functions play the role of the continuations.

Let us look at the correspondence between the original \lang{} code (Fig.~\ref{fig:portable-class})
and the generated Solidity code (Fig.~\ref{fig:solidity}) from a higher-level perspective:
To verify that the \lang{} code matches the generated Solidity code,
we proceed as follows.

First, we verify that the control flow of
Fig.~\ref{fig:portable-class} is accurately described by the automaton diagram in Fig.~\ref{fig:funding-diagram}.
In particular, we observe that there are two loops in the automaton and there are
also two \lstinline!while! loops in the \lang{} code. Further, there are three \lstinline!awaitCl!
expressions in the code, and there are three states in the automaton (plus a final state).

Second, we verify that the automaton in Fig.~\ref{fig:portable-class} corresponds to the program flow of the
Solidity code in Fig.~\ref{fig:solidity}.
In particular, we observe that there are four
states in the automaton and there are four states in the Solidity code.
Three of those have an associated function (\lstinline!T0! is \lstinline!Fund!, \lstinline!T1! is \lstinline!Move!, \lstinline!T2! is \lstinline!Payout!),
which are the only public functions that can be invoked in that state,
thanks to the \lstinline!require! statements.
In the final state \lstinline!T3!, no public function can be invoked.
Furthermore, we can see that the automaton has two loops.
It is possible to go from \lstinline!T0! either to \lstinline!T1! or stay in \lstinline!T0!.
This is represented in the Solidity code,
by checking for the loop condition at the end of the function associated to \lstinline!T0!,
and then either changing the phase to \lstinline!T1!, or doing nothing, which means staying in \lstinline!T0!.
Similarly, the loop in state \lstinline!T1! is encoded with an \lstinline!if! at the end of the function
to conditionally move to the next phase.

These two steps should illustrate how the control flow of the \lang{} program
-- which is abstractly visualized by the automaton -- is implemented and enforced by the generated Solidity program.

%% file: paper/listings/solidity_ttt.tex
\begin{subfigure}{.48\textwidth}
  \begin{lstlisting}
  class TTT {

    // @cl annotated definitions
    def updateBoard(): Unit =
      { /* ... */ }
    def fund(): (U256, Unit) =
      (readLine("How much?").u, ())
    def move(): (U256, UU) =
      ("0".u, UU(readLine("x-pos?"),
      readLine("y-pos?"))
    def payout(): (Ether, Unit) = {(*@\label{code:client-payout-phase}@*)
      readLine("Press (enter) for payout")
      ("0".u, ()) }

    // body of awaitCl expressions
    def Fund(): (Ether, Unit) =(*@\label{code:client-funding-phase}@*)
      fund()
    def Move(): (Ether, UU)   =(*@\label{code:client-move-phase}@*)
      move()
    def Payout(): (Ether, Unit) =(*@\label{code:client-payout-phase}@*)
      payout()

    /* ... */




  }
  \end{lstlisting}
  \subcaption{Scala client}\label{fig:scala-client}
\end{subfigure}%
\begin{subfigure}{.48\textwidth}
\begin{lstlisting}[language=Solidity, xleftmargin=0pt, xrightmargin=0pt, firstnumber=29]
contract TTT {
  State phase = T0;(*@\label{code:solidity-init-phase}@*) enum State {T0, T1, T2, T3}

  // @co annotated definitions
  int moves  = 0;
  int winner = 0;
  int[][] board;
  function peformMove(int x, int y) private { /*...*/ }

  // continuation of awaitCl expressions
  function Fund() public {
    require(phase == T0);(*@\label{code:solidity-require-funding}@*)
    /*...*/;
    if (!(balance < FUNDING_GOAL)) phase = T1;
    /* else phase remains T0; this models the first while loop */
  }
  function Move(int x, int y) public {
    require(phase == T1 && sender == players(moves %
    /*...*/;
    if (!(moves < 9 && winner == 0)) phase = T2;
    /* else phase remains T2; this models the second while loop */
  }
  function Payout() public {(*@\label{code:solidity-to-final-phase}@*)
    require(phase == T2);
    /*...*/;
    phase = T3;
  }
}
\end{lstlisting}
  \subcaption{Solidity contract}\label{fig:solidity}
\end{subfigure}%

%% file: paper/formalisation.tex
\section{Compilation and its Correctness}
\label{sec:formalization}
We informally introduce \lang's compilation process and our notion of correctness before formally specifying and proving the compiler correct. 

\subsection{High-level Overview of \lang's Secure Compilation}

To implement the contract-client interaction,
we CPS-translate \lang code
and execute continuations alternately between contract and client.
A standard CPS translation is, however, not sufficient
because the control flow is distributed and
we need to send function calls (i.e., the current continuation)
over the network -- or, more specifically, send
the name and the arguments of the next function to execute.
For this, we defunctionalise~\cite{Reynolds1972DefinitionalIF}
the code to turn functions calls (which represent continuations) into data.
This compilation process performs an \emph{inversion of control} between the contract and the client.
With \lang's contract--client communication in direct style, 
 we can write dApps as if \emph{the contract} was in control of the execution; 
\lang allows the contract to request messages from clients and to process only responses that 
it requested.

After the compilation process, clients are in control of the execution
because, in 
blockchains, contracts purely
respond to messages from clients.
As a result, dApps may become the target of malicious attacks.
In our security model, we trust the contract to execute the code that we generate for it,
whereas we consider the client code untrusted, i.e., the client side
can run arbitrary code.
Crucially, it could pass unintended continuations to the contract
to force the execution to continue in an arbitrary state.
For example, in the source code of the TicTacToe game (Section~\ref{sec:the-language}),
one needs to go through the game loop after funding and before payout.
Yet, the compiled code is separated and distributed into small chunks.
Parties execute a chunk and then wait for other parties
to decide on a move that influences 
how to proceed with the execution.
For this reason, the client could send a message at any time telling the contract to go into the payout phase.
We need to guard the contract against such attempts to make it deviate from the protocol.
Conceptually,
if the client was able to force the execution to continue in an arbitrary state,
the control flow in the Prisma source would be violated.
Execution would 'jump' from one client expression to another one skipping the code in between,
which is not possible with the semantics of the source language.

\lang's compiler avoids such attacks and 
preserves control flow by inserting guards on the contract side. Guards are in places 
where the basic blocks of the program have been separated 
and distributed onto different hosts by CPS translation -- to reject any 
improper continuations from clients. 
Guarding ensures
the control flow integrity~\cite{AbadiBEL09} of the contract
in the presence of malicious clients by
excluding any behavior of the compilation target 
that cannot be observed from the source.
Informally, this is our notion of \emph{secure compilation},
which we rigorously define and prove for \lang's compiler in this section.
The compilation process is key in hiding the complexity of enforcing
distributed control flow from the developer
-- hence, a formal proof of it correctness is critical.

\begin{figure}
  \figurehead
  \centering
\begin{tikzcd}
\mathsf{source}                                                                                                                                       &  & \mathsf{target}                                                                                              \\
\mathsf{Term_s} \arrow[rr, "compile"'] \arrow[d, "eval_s" description, dotted, bend right=49] \arrow[dd, "eval_{s,b}" description, dotted, bend right=60] &  & \mathsf{Term_t} \arrow[d, "eval_t" description, bend left=49] \arrow[dd, "eval_{t,b}" description, bend left=60] \\
\mathsf{Trace_s}                                                                                                                                      &  & \mathsf{Trace_t} \arrow[ll, "correctness", leftrightarrow]                                                                   \\
\mathsf{Trace^*_s}                                                                                                                                    &  & \mathsf{Trace^*_t} \arrow[ll, "security", leftrightarrow]                                                                   
\end{tikzcd}
  \caption{Secure Compilation}
  \label{fig:secure-compilation}
\end{figure}
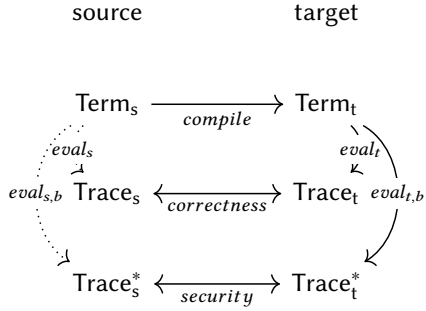

To formalize the compiler, we specify a source and a target language.
Fig.~\ref{fig:secure-compilation} shows a schema of our compilation and the proof.
The compiler (Fig.~\ref{fig:secure-compilation}, top) is a function that maps terms in the source language ($Term_s$)
into terms in the target language ($Term_t$).
A correct compiler preserves some properties of the code -- depending on the notion of correctness.
For example, typeability-preserving and semantics-preserving compilers
have been extensively studied~\cite{PatrignaniAC19}.
Because types are not the focus of this paper, we omitted them from the figure.
In the middle part of Fig.~\ref{fig:secure-compilation}, we show the evaluation of source and target 
to traces ($eval_s$ and $eval_t$, respectively) -- and traditional compiler correctness as
the equivalence between traces generated from the sources ($Trace_s$) and from the target ($Trace_t$).
But compiler correctness\footnote{
  Type and semantics preservation is not the focus of this paper;
  we presume them for our compiler without a formal proof.}
 in this traditional sense 
is not sufficient in the presence of malicious
attackers that can tamper with parts of the code.
Instead, we need to prove that \lang is a {\it secure abstraction},
i.e., if security problems can arise on the target, they 
must be visible in the \lang source code, too, so that
developers do not need to look at target code to reason about 
potentially misbehaving clients.
To this end, we define 
a hypothetical \emph{attacker model on the source code} as 
the ability to only replace the body of a \lang client expression 
and show that, with the contract part hardened with guards, the target attacker
does not gain additional power over the 
hypothetical source attacker. Specifically, we define 
malicious semantics $eval_{t,b}$ and $eval_{s,b}$ for the target
and the source language, respectively, and show 
that $eval_{t,b} (compile~e) = compile (eval_{s,b}~e)$ (security property in Fig.~\ref{fig:secure-compilation}).

In the reminder of this section:

\begin{itemize}
\item We present the core calculus (Section~\ref{subsec:syntax})
$\minilang_*$  -- 
a hybrid language that includes elements of both the source ($\minilang_s$)
and the compilation target ($\minilang_t$), while abstracting over details of both Scala and Solidity. 
We define a hybrid language because
the source and the target share many constructs -- the hybrid language
allows us to focus on how the differences are compiled.

\item We model the compiler (Section~\ref{subsec:compilation})
as a sequence of steps that transform $\minilang_*$ programs 
via several intermediate representations.

\item We define $\minilang_*$ semantics as a reduction relation over 
configurations consisting of traces of evaluation events and expressions being evaluated (Section~\ref{subsec:semantics}).
We distinguish between a good semantics, which evaluates the program in the usual way, 
and a bad semantics, which models attackers by ignoring client instructions and producing 
arbitrary values that are sent to the contract.

\item We prove secure compilation
by showing that the observable behavior of the programs 
before and after compilation is equivalent (Section~\ref{subsec:secure-compilation}).
We capture the observable program 
behavior by the trace of events generated during program evaluation (as guided by the semantic definition) 
and show trace equivalence of programs before and after compilation.
\end{itemize}

\begin{figure}
	\figurehead
	\centering
	\input{notes/syntax.tex}

	\caption{$\minilang_*$ syntax.}
	\label{fig:syntax}
\end{figure}

\subsection{Syntax}
\label{subsec:syntax}

The syntax of $\minilang_*$ (Fig.~\ref{fig:syntax}),
has three kinds of identifiers $id$, $\iden$, $\syniden$,
from unspecified sets of distinct names.
Pure identifiers $id$ are for function arguments and let bindings;
mutable variables $\iden$ are for heap variable assignment and access.
In the target program, mutable variables $\syniden$ ($\whovar, \statevar, \clfnvar, \cofnvar$)
generated by the compiler can also appear.
We call compiler-generated identifiers \emph{synthetic}.
Normal identifiers are separated from synthetic ones to distinguish
compiler generated and developer code.
Definitions $\decl$ and definitions for synthetic identifiers $\syndecl$
are semicolon-separated lists of declarations that
assign values to variables and annotate either the contract or the client location.
Each program $P$ consists of definitions $\decl$ and synthetic definitions $\syndecl$
followed by the main contract expression $\mainexp$.
Program $P$ corresponds to a single \lang{} split class,
$\decl$ and $\syndecl$ to methods and generated methods,
and $\mainexp$ to a constructor containing the initialisation of its class members
(such as the body of \lstinline!init!, Fig.~\ref{fig:portable-class}).

Constants $c$ are unsigned 256\,bit integer literals and built-in operators.
$\minilang_*$ supports tuples introduced by nesting pairs ($::$)
and eliminated by pattern matching.
Tuples allow multiple values to cross tiers in a single message.
Values $v$ are constants, value pairs, and lambdas.
Patterns $x$ are constants, pattern pairs, and variables.
Expressions $\subexp$ are constants, expression pairs, lambdas, variables,
  variable accesses/assignments, bindings and function applications.

Main expressions $\mainexp$ may further contain remote client expressions, 
embedding client code into contract code and waiting for its result.
The source client expression ${\downarrow}_s(e,~()\to e)$
can be answered by any client whose address fulfills the predicate specified as first argument.
$\downarrow_s$ corresponds to direct-style remote access via \lstinline!awaitCl! in \lang{}.
We use the syntax form ${\downarrow}_t(c,~()\to e)$
to model the execution of code $e$ on the specified client $c$.
$\downarrow_t$ has no correspondence in the source syntax.
Our compilation first splits the predicate from the source client expressions
into a separate access control guard.
Then, it eliminates client expressions, 
turning the contract into a passive entity that 
stops and waits for client input.

We now map the hybrid language $\minilang_*$ to 
the source and target languages, $\minilang_s$ and $\minilang_t$.
$\minilang_s$ has all expressions of $\minilang_*$,
except those that contain $\bind$ (bind), $\trampoline$ (trampoline), $\mathsf{Done}$, $\mathsf{More}$, $\downarrow_t$,
or synthetic identifiers $j$.
$\minilang_t$ has all expressions of $\minilang_*$
except those that contain $\downarrow_s, \downarrow_t, \bind$.

$\bind$ and $\downarrow_t$ may not appear neither in source nor target programs;
the former is used only as an intermediate construction for the compiler,
the latter only during evaluation to track the current location.

\begin{figure}
  \figurehead
  \centering
  \input{notes/def_fv_sugar_match.tex}
  \caption{Syntactic sugar.}
  \label{fig:sugar}
\end{figure}

\subparagraph{Syntactic sugar.}
In Fig.~\ref{fig:sugar}, we define some syntactic sugar to improve readability.
We use infix binary operators and tuple syntax for nested pairs ending in the unit value $()$;
we elide the let expression head for let bindings matching $()$,
$\assert(x)$ is a let binding matching true;
we use monadic syntax for let bindings of effectful expressions; 
 $\ifletexp x = \mainexp \thenexp \subexp \elseexp \subexp$
is the application of the built-in $\try$ function.

\subparagraph{Events and configurations.}
In Fig.~\ref{fig:syntax2}, we define left-to-right
evaluation contexts $E$~\cite{FelleisenH92};
and compilation frames $F$~\cite{Pitts00},
such that every expression decomposes into a frame-redex pair $F~\subexp$ or is an atom $\Finv$.
Events $\trace$ and $\syntrace$ are lists that capture
the observable side-effects of evaluating expressions.
They are either
(a)~state changes
  $\wrevent(c,\iden,v)$ and $\wrevent(c,\syniden,v)$,
  from the initial definitions or variable assignment,
  where $\iden$ and $\syniden$ are the variable being assigned, $c$ the location,
  and $v$ the assigned value, 
  or
(b)~client-to-contract communication $\msgevent(c,v)$,
  where $c$ is the address of the client and $v$ the sent value. 
Configurations $C = \trace;\syntrace;c \mainexp$, represent a particular execution state,
where $\trace$ (and $\syntrace$) are
traces of normal (and synthetic) events produced by the evaluation,
$c$ is the evaluating location, and $\mainexp$ is the expression under evaluation.

\begin{figure}
  \figurehead
  \centering
  \input{notes/syntax2.tex}

  \caption{Frames, Events and configurations.}
  \label{fig:syntax2}
\end{figure}

\subparagraph{Initialization.}
Initialization in Fig.~\ref{fig:initialisation} generates the initial program configuration,
which
models the decentralized application with a single contract and multiple clients.
We model a fixed set of clients $A$ interacting with a contract.
The initialization of a program 
$\decl;\syndecl;\mainexp$
to a configuration 
$\trace;\syntrace;0;\mainexp$
  leaves the expression $\mainexp$ untouched
  and generates a
  list of events
  -- one write event for each normal and synthetic 
  definition.
Location 0 represents the contract.

\begin{figure}
	\figurehead
	\centering
  $\begin{array}{lll}
  init_A(\decl;\syndecl;\mainexp)
    &=& init_A(\decl;\syndecl);~0;~\mainexp \\
  init_A(\decl;\syndecl)
    &=& ( wr(0,\makebox[0pt][l]{$\iden$}\phantom{\syniden},v)    ~|~ \forall ~ (\coVal \iden=v)    \in \decl ) \\
    & & ( wr(0,\makebox[0pt][l]{$\syniden$}\phantom{\syniden},v) ~|~ \forall ~ (\coVal \syniden=v) \in \syndecl ) \\
    & & ( wr(c,\makebox[0pt][l]{$\iden$}\phantom{\syniden},v)    ~|~ \forall ~ (\clVal \iden=v)    \in \decl,~~ c \in A ) \\
    & & ( wr(c,\makebox[0pt][l]{$\syniden$}\phantom{\syniden},v) ~|~ \forall ~ (\clVal \syniden=v) \in \syndecl,~~ c \in A )
  \end{array}$
  \caption{Initialization.}
  \label{fig:initialisation}
\end{figure}

\subsection{Compilation}
\label{subsec:compilation}

The compiler eliminates language features
not supported by the compilation target
one by one, lowering the abstraction level from
  (1)~\emph{direct style communication (DS)} -- which needs language support for !-notation~\cite{idris-bang-notation} --
  through the intermediate representations of
  (2)~\emph{monadic normal form (MNF)} -- which needs support for do-notation~\cite{haskell-do-notation} -- and
  (3)~\emph{continuation-passing style (CPS)} -- which needs higher-order functions -- to
  (4)~explicitly encoding \emph{finite state machines (FSM)} -- for which first-order functions suffice.
In the following, we provide an intuition for the compiler steps
and subsequently their formal definitions.

First, the compilation steps $\mnfe$ and $\assoc$ transform DS remote communication ${\downarrow_s}(e, () \to \subexp)$ 
to variable bindings ($id := \subexp$)
and nested let bindings are flattened such that a program is prefixed by a sequence of let expressions.
Second, step $\guarde$ generates access control guards around client expressions
to enforce correct execution even when clients behave maliciously.
Third, step $\dtle$ transforms previously generated let bindings for remote communication
($x \gets \subexp_1\seq \mainexp_2$) 
to monadic bindings $\subexp_1 \bind x \to \mainexp_2$.
Fourth, step $\defun$ transforms functions into data structures that
can be sent over the network and are interpreted by a function (i.e., an FSM) 
on the other side.
Compared to standard defunctionalization, we handle two more issues.
First, we defunctionalize the built-in higher-order operator ($\bind$) 
by wrapping the program expression into a call to a trampoline $\trampoline(...)$
  and transforming the bind operator ($... \bind x \to ...$)
  to the $(\mathsf{More}, ..., ...)$ data structure;
 the trampoline repeatedly interprets the argument of $\mathsf{More}$ 
  until it returns $\mathsf{Done}$ instead of $\mathsf{More}$ signaling the program's result.
Second, we keep contract and client functions separate
by generating separate synthesized interpreter functions, called $\mathsf{cofn}$ and $\mathsf{clfn}$,
  thereby splitting the code into the parts specific to contract and client.

\subparagraph{MNF transformation (Fig.~\ref{fig:mnf}).}
The $\mnfp$ function wraps the main expression $\mainexp$ into a call to the trampoline
with the pair $(\text{Done}, \mainexp)$ -- signaling the final result -- as argument.
Then, $\mnfe$ transforms expressions recursively,
binding sub-expressions to variables,
resulting in a program prefixed by a sequence of let bindings.
As recursive calls to $\mnfe$ may return chains of let bindings,
we apply $\assoc$ to produce a flat chain of let bindings.
Given a let binding, whose sub-expressions are in MNF,
associativity recursively flattens the expression,
by moving nested let bindings to the front,
($ ...~(...~\mainexp_0;~\mainexp_1);~\mainexp_2 = ...~\mainexp_0;~(...~\mainexp_0;~\mainexp_2) $),
creating a single MNF expression
(i.e., $\assoc$ is composition for MNF terms).

\begin{figure}
  \figurehead
  \centering
  \input{notes/def_mnf.tex}
  \caption{Monadic normal form transformation.}
  \label{fig:mnf}
\end{figure}

\begin{figure}
  \figurehead
  \centering
  \input{notes/def_guard.tex}
  \caption{Guarding.}
  \label{fig:guard}
\end{figure}

\subparagraph{Guarding (Fig.~\ref{fig:guard}).}
We insert access control guards
for remote communication expressions $\gets_s$ to enforce
(i)~the execution order of contract code after running the client expression and
(ii)~that the correct client invokes the contract continuation.
The transformation sets the synthetic variable $\mathsf{state}$ to a unique value before the client expression,
and stores the predicate to designate valid clients in the synthetic variable $\mathsf{who}$.
After the client expression, the generated code asserts that the contract is in the same state,
and checks that the sender fulfills the predicate.
The assertion trivially holds in the sequential execution of the source language,
but after more compilation steps the client will be responsible
for calling the correct continuation on the contract.
Since client code is untrusted, the contract needs to ensure
that only the correct client can invoke only the correct continuation.

\begin{figure}
  \figurehead
  \centering
  \input{notes/def_cps.tex}
  \caption{Continuation-passing style transformation.}
  \label{fig:rotation}
\end{figure}

\begin{figure}
  \figurehead
  \centering
  \input{notes/def_defun.tex}
  \caption{Defunctionalization.}
  \label{fig:defunctionalisation}
\end{figure}

\subparagraph{CPS transformation (Fig.~\ref{fig:rotation}).}
The  $\dtle$ transformation turns the
chains of let bindings produced by $\mnfe$ into
  CPS.
  The chain contains three cases
of syntax forms:
(1)~monadic binding ($x \gets ...\seq \mainexp_1$),
(2)~let binding ($\letexp x = \subexp_0\seq \mainexp_1$),
or (3)~final expression.
For (1), $\dtle$ replaces the monadic binding with an explicit call to the bind operator ($... \bind (x \to cps(\mainexp_1))$).
For (2) and (3), $\dtle$ recurses into the tail of the chain.
This resembles do-notation desugaring (e.g., in Haskell).

\subparagraph{Defunctionalization (Fig.~\ref{fig:defunctionalisation}).}
The $\defun$ function
transforms the chains of let bindings and bind operators produced by $\dtle$,
which contains three cases of syntax forms:
(1)~a bind operator ($\subexp_1 \bind \subexp_2$),
or (2)~a let binding ($\letexp x = \subexp_1\seq \subexp_2$),
or (3)~the final expression.
For (1), 
$\subexp_1$ and $\subexp_2$ are replaced by data structures
that contain values for the free variables in $\subexp_1$ and $\subexp_2$
and are tagged with a fresh ID.
The body of the expression is lifted to top-level synthetic definitions.
For this, $\defun$ modifies the synthetic definitions $b$
by extracting the body $\subexp_{1,alt}$ of the synthetic $\clfnvar$ definition
and the body $\subexp_{2,alt}$ of $\cofnvar$,
and by adding an additional conditional clause to these definitions.
The added clause answers to requests for a given ID with evaluating the original expression.
For (2) and (3), $\defun$ recurses into the expressions.

After defunctionalization, lambdas
$x \to \subexp_0$ are lifted and assigned a top-level identifier $id_0$
and lambda applications, $id_0(\subexp_1)$,
are replaced with calls to a synthesized interpreter 
function $\text{\textsf{fn}}(id_0, \subexp_1)$.
The latter branches on the identifier
and executes the code that was lifted out of the original function.

\subparagraph{Compiling.}
The $\mathit{comp}$ function composes the compiler steps (not including $\mathit{mnf}$).
We also define the $\mathit{comp'}$ function,
which jumps over the wrapping $trmp$ expression
and initialises the defunctionalisation with an environment that contains the two functions $\cofnvar$ and $\clfnvar$, which assert false.

\vspace{-1.1em}
\[ \begin{array}{ll}
\mathit{comp} &~=~ \defun ~\circ~ \dtle ~\circ~ \guarde \\
\mathit{comp}' &~=~ \defunOuter ~\circ~ \dtlp ~\circ~ \guardp
\end{array} \]

\subsection{Semantics}
\label{subsec:semantics}

We model the semantics as a reduction relation
over configurations $\trace;\syntrace;c;\mainexp \to \trace';\syntrace';c';\mainexp'$.
Location $c=0$ denotes contract execution, otherwise 
execution of client of address $c$.
We distinguish good ($\textcolor{blue}{\to_g}$)
and  bad ($\textcolor{red}{\to_b}$) evaluations
(Fig.~\ref{fig:dynamic-semantics} and~\ref{fig:dynamic-semantics2});
shared rules are in black, without subscript ($\to$).

\textit{Attacker model.}
Attackers can control an arbitrary number of clients
and make them send arbitrary messages.
Hence, the bad semantics can answer a request to a client
with an arbitrary message from an arbitrary $id$.
We use evaluation with bad semantics to show that our compiler enforces access control against malicious clients.

Good evaluations of client expressions in the source language (\textcolor{blue}{\textsc{Rgs}})
  reduce to a client expression with a fixed client that fulfils the given predicate.
  We require that predicates evaluate purely.
  Hence, $\trace$ and $\syntrace$ do not change in the evaluation.
On the other hand,
  bad evaluation of client expressions in the source language (\textcolor{red}{\text{Rbs}})
  ignores the predicate, choosing an arbitrary client.
Similarly, bad evaluation also chooses an arbitrary client for
  the evaluation of a trampoline in the target (\textsc{Rtm}),
  which does not specify a predicate.
The trampoline ends when it reaches $\mathsf{Done}$ (\textsc{Rtd}).
Further, after choosing a client to evaluate,
  the good evaluation (\textcolor{blue}{\textsc{Rg}}) continues to reduce the client expression to a value,
while the bad evaluation (\textcolor{red}{\textsc{Rb}}) replaces the expression $\subexp$
   with a (manipulated) arbitrary value $v'$.
Both evaluations (\textcolor{blue}{\textsc{Rg}}, 
\textcolor{red}{\textsc{Rb}}) emit the message event $\msgevent(c,v)$ and an assignment to the special variable $\mathsf{sender}$,
  when a client expressions is reduced to a value $v$, to record the client--contract interaction.

\begin{figure}
  \figurehead
  \centering
  \input{notes/def_match.tex}
  \caption{Pattern matching.}
  \label{fig:pattern-matching}
\end{figure}

\begin{figure}
  \figurehead\scriptsize
  \centering
  \input{notes/dynamic_semantics.txt.tex}
  \caption{Evaluation (1/2).}
  \label{fig:dynamic-semantics}
\end{figure}

\subparagraph{Common Evaluation (Fig.~\ref{fig:dynamic-semantics2}).}
Expressions %
are reduced under the evaluation context $E$ on the current location (\textsc{Re}),
assignment to variables is recorded in the trace (\textsc{Rset$^\circ$}),
accessing 
a variable
  is answered by the most recent assignment to it from the trace
   in the current location (\textsc{Rget$^\circ$}).
For synthetic variables,
  we use the synthetic store (\textsc{Rget$^\dagger$}, \textsc{Rset$^\dagger$}).
Binary operators are defined as unsigned 256\,bit integer arithmetic;
  we only show the rule for addition (\textsc{Rop}).
Further, we give rules for
  conditionals (\textsc{Rt}, \textsc{Rf}),
  let binding (\textsc{Rlet}) and
  function application (\textsc{Rlam})
  using pattern matching.

\begin{figure}
  \figurehead
  \centering
  \input{notes/dynamic_semantics2.txt.tex}
  \caption{Evaluation (2/2).}
  \label{fig:dynamic-semantics2}
\end{figure}

\subparagraph{Pattern matching (Fig.~\ref{fig:pattern-matching}).}
Matching $[x{\Mapsto}v]$ is a partial function, 
matching patterns $x$ with values $v$,
  returning substitution of variables $id$ to values.
Matching is recursively defined over pairs; it
  matches constants to constants,
  identifiers to values by generating substitutions,
  and fails otherwise.
Substitutions $[id{\mapsto}v]$, in turn, can be applied to terms $\subexp$, written $[id{\mapsto}v]~\subexp$ (capture-avoiding substitution). Substitutions $\sigma$ compose right-to-left $(\sigma\sigma') x = \sigma (\sigma' x)$.

\subsection{Secure Compilation}
\label{subsec:secure-compilation}

We prove that the observable behavior of the contract
before and after compilation is equivalent.
We capture the observable behavior
by execution traces
and show that trace equivalence holds even when the program is attacked,
i.e., reduced by $\textcolor{red}{{} \to_b^* {}}$.

\subparagraph{Modelling Observable Behavior.}
The only source of observable nondeterminism in the bad semantics is
 the evaluation of $\downarrow_s$ and $\downarrow_t$.
As clients decisions on message sending
  are influenced by the state of contract variables, tracking
  incoming client messages
  and state changes in the trace
  suffices to capture the observable program behavior.
If the observable behavior is the same for the source and the compiled programs,
they are indistinguishable.
Thus, behavior preservation
  amounts to trace equality on programs before and after compilation.
Further, it suffices to model 
  equality for non-stuck traces.
The evaluation gets stuck (program crash) on assertions that guard against
    deviations from the intended program flow.
The Ethereum Virtual Machine
  reverts contract calls that crash, i.e., state changes of crashed calls do not take effect, hence, 
  stuck traces are not observable.

Since bad evaluation is nondeterministic,
  we work with 
  not just programs, expressions and configurations, but
  program sets, expression sets, and configuration sets.
Let $\trace;\syntrace;\mainexp \Downarrow$ be the trace set of the configuration $\trace;\syntrace;0;\mainexp$, e.g.,
  the set of tuples of the final event sequence $\trace'$ and value $v$
  of all reduction chains
  that start in $\trace;\syntrace;0;\mainexp$
  and end in $\trace';0;\syntrace';v$.
Our trace set definition does not include
synthetic events $\syntrace'$ of the final configuration.
Synthetic events are introduced through compilation; excluding them
  allows us to put source and target trace sets in relation.
Further, let the trace set of a configuration set $T~\textcolor{red}{\Downarrow}$,
  be the union of the trace sets for each element:

\vspace{-1.2em}
\[\trace;\syntrace;\mainexp~\textcolor{red}{\Downarrow} ~~=~ \{~ (\trace',v) ~|~ (\trace;\syntrace;0;\mainexp)~\textcolor{red}{\to_b^*}~(\trace';\syntrace';0;v) ~\} \]

\vspace{-1.2em}
\[ T~\textcolor{red}{\Downarrow} ~=~ \bigcup_{\trace;\syntrace;\mainexp \in T} ~ \trace;\syntrace;\mainexp~\textcolor{red}{\Downarrow}\]

We say that two configuration sets $T$ and $S$ are equivalent, denoted by $T \approx S$, iff
$T$ and $S$ have the same traces sets:

\vspace{-1.2em}
\[(T~\textcolor{red}{\approx}~S)  ~~~\Leftrightarrow~~~  (T~\textcolor{red}{\Downarrow} ~=~ S~\textcolor{red}{\Downarrow})\]

\vspace{-0.2em}
By this definition,
  two expressions that eventually evaluate to the same value with the same trace are related by trace equality.
We use this notion of trace equality to prove that a source program is trace-equal to its compiled version
  by evaluating the compiled program forward \textcolor{red}{$\to_b^*$}
  and the original program backward \textcolor{red}{$\gets_b^*$}
  until configurations converge.

\subparagraph{Secure Compilation.}
Theorem~\ref{theorem:trace-equality} states our correctness property,
which says that observable traces generated by the malicious evaluation of programs
are preserved ($\textcolor{red}{\approx}$) by compilation.
The malicious evaluation models
  that client code has been replaced with arbitrary code,
  while contract code is unchanged.
The preservation of observable traces implies
  the integrity %
  of the (unchanged) contract code.
Secure compilation
guarantees that developers can write safe programs
in the source language without knowledge about the compilation
or the distributed execution of client/contract tiers.

\begin{globaltheorem}[Secure Compilation]\label{theorem:trace-equality}\emph{
For each program $P$ over closed terms, the trace set of
the program under attack
equals the trace set of
the compiled program under attack: \\
$ \forall P.\, \{~ init_A(comp'(mnf'((P)))) ~\}~\textcolor{red}{\approx}~\{~ init_A(P) ~\} $.
}\end{globaltheorem}

We first show that trace equality holds for the different compiler steps.
Some compiler steps are defined as a recursive term-to-term transformation on open terms,
whereas traceset equality is defined by reducing terms to values, i.e., on closed terms.
Since all evaluable programs are closed terms,
we show that the compiler steps preserve the traceset of an open term $e$ that 
is closed by substitution $[x\Mapsto v]$. We formulate the necessary lemmas and sketch the proofs -- the detailed proof is in Appendix~\ref{appendix:proofs}.

{\small
\begin{globallemma}[assoc correct]\label{lemma:assoc}
$\{\, \trace;\syntrace; [x{\Mapsto}v]\,assoc(\mainexp) \,\} \approx \{\, \trace;\syntrace; [x{\Mapsto}v]\,\mainexp \,\}$
\end{globallemma}
\begin{globallemma}[mnf correct]\label{lemma:mnf}
$\{\, \trace;\syntrace; [x{\Mapsto}v]\,mnf(\mainexp) \,\} \approx \{\, \trace;\syntrace; [x{\Mapsto}v]\,\mainexp \,\}$
\end{globallemma}
\begin{globallemma}[mnf' correct]\label{lemma:mnf'}
$\{\, init_C(mnf'(\decl;\syndecl; \mainexp)) \,\} \approx \{\, init_C(\decl;\syndecl; \mainexp) \,\}$
\end{globallemma}
\begin{globallemma}[comp correct]\label{lemma:comp}
$\{\, [x{\Mapsto}v]\,init_A(comp(\decl;\syndecl; \trampoline(\mainexp))) \,\} \approx \{\, init_A(\decl;\syndecl; \trampoline([x{\Mapsto}v]\,\mainexp)) \,\}$
\end{globallemma}
\begin{globallemma}[comp' correct]\label{lemma:comp'}
$\{\, [x{\Mapsto}v]\,init_A(comp'(\decl;\syndecl; \trampoline(\mainexp))) \,\} \approx \{\, init_A(\decl;\syndecl; \trampoline([x{\Mapsto}v]\,\mainexp)) \,\}$
\end{globallemma}
}
\subparagraph{Proof sketch.}
Lemma~\ref{lemma:assoc}--\ref{lemma:comp'} hold by chain of transitive trace equality relations.
We show that a term is trace-equal to the same term after compilation,
  by evaluating the compiled program (${\to}^*$)
  and the original program (${\gets}^*$)
  until configurations converge.
In the inductive case,
  we can remove the current compiler step in redex position under traceset equality ($\approx$)
  since traces before and after applying the compiler step are equal by induction hypothesis.

An interesting case is the proof of $\mathit{comp}$ for $P = \decl;\syndecl; \downarrow(\subexp_0, () \to \subexp_1)$.
The compiler transforms the remote communication $\downarrow_s$
  into the use of a guard and a trampoline.
The compiled program steps to the use of $\downarrow_t$,
the source program to $\downarrow_s$.
In the attacker relation $\textcolor{red}{\to_b}$,
  arbitrary clients can send arbitrary values with $\downarrow_t$,
  leading to additional traces compared to the ones permitted in the source program
  where communication is modeled by $\downarrow_s$.
We observe that $\downarrow_s$ generates the trace elements
$\msgevent(c, \textcolor{red}{v}), \wrevent(0, \mathrm{sender}, c)$ for all $\textcolor{red}{v}$
and that $\downarrow_t$ generates the trace elements
$\msgevent(\textcolor{red}{c'}, \textcolor{red}{v}), \wrevent(0, \mathrm{sender}, \textcolor{red}{c'})$
for all $\textcolor{red}{v}, \textcolor{red}{c'}$,
which differ for $\textcolor{red}{c'} \ne c$.

Compilation adds an $\mathsf{assert}$ expression (Fig.~\ref{fig:guard})
evaluated after receiving a value from a client.
The $\mathsf{assert}$ gets stuck for
configurations
that produce trace elements with $\textcolor{red}{c'} \ne c$,
removing the traces of such configurations from the
trace set,
leaving only the traces where $\textcolor{red}{c'} = c$.
Hence, the trace set before and after compilation is equal under attack.

\section{Implementation}

\lang{} is embedded into Scala (the host language) with its features implemented as a source-to-source macro expansion.\footnote{The implementation entails 21 Scala files, 3\,412
lines of Scala source code (non-blank, non-comment) licensed under Apache 2.0 Open Source.
The compiler phases are macros that recurse over the Scala AST:
(a) the guarding phase, (b) the ``simplifying'' phase (including MNF translation, CPS translation of terms, defunctionalisation),
and (c) the translation phase of (a subset) of Scala expressions and types to a custom intermediate representation based on Scala case classes. The intermediate representation is translated to Solidity code and passed to the Solidity compiler (solc).}

The backend generating Solidity code is well separated.
One could disable the compilation step to Solidity in the compilation pipeline,
e.g., to run distributed code on multiple JVMs instead. In this case, the “contract code” would be executed 
by one computer (the “server”), and other computers would run the “client code”.

The Scala runtime of \lang{} contains the implementation of the serialisable datatypes, portable between Scala and the EVM (fixed-size arrays, dynamic arrays, unsigned integers of length of powers of two up to 256 bit).
Our runtime wraps \lstinline!web3j!~\cite{web3j-github} (for invoking transactions and interacting with the blockchain in general),
\lstinline!headlong!~\cite{headlong-github} (for serialisation/deserialisation in the Ethereum-specific serialisation format),
as well as code to parse Solidity and Ethereum error messages and to translate them to Scala error messages.

%% file: notes/syntax.tex
$id \in I\hspace{-.5pt}D$
\hspace{20pt} $\iden \in I$
\hspace{20pt} $\syniden \in \{ \whovar, \statevar, \clfnvar, \cofnvar \}$
\\[1em]

\begin{tabular}{lrll}
(definition)           &$\decl$    &::=  $\coVal i$ = $v$; $\decl$    | $\clVal i$ = $v$; $d$ | $()$\\
(synthetic definition) &$\syndecl$ &::=  $\coVal j$ = $v$; $\syndecl$ | $\clVal j$ = $v$; $b$ | $()$\\
(program)     &$P$        &::=  $\decl$; $\syndecl$; $\mainexp$
\\[0.7em]
(constant)    &$c$ &::= 0 | 1 | 2 | ... | $\true$ | $\false$
                    ~|~ () | \&\& | + | == | < | $\try$ \\
              &    &~~|~ >>= | $\trampoline$ | $\mathsf{Done}$ | $\mathsf{More}$ \\
(value)       &$v$ &::= $c$ | $v$ :: $v$ | $x{\to}\subexp$\\
(pattern)     &$x$ &::= $c$ | $x$ :: $x$ | $id$\\
(expression)  &$\subexp$ &::= $c$ | $\subexp$ :: $\subexp$ | $x{\to}\subexp$ | $id$ ~|~ \letexp $x$=$\subexp$; $\subexp$ | $\subexp~\subexp$ \\
              &    &~~|~ $\this$.$\iden$ | $\this$.$\iden$ := $\subexp$ ~|~ $\this$.$\syniden$ | $\this$.$\syniden$ := $\subexp$ \\
(main expression)  &$\mainexp$ &::= $c$ | $\mainexp$ :: $\mainexp$ | $x{\to}\subexp$ | $id$
                     ~|~ \letexp $x$=$\mainexp$; $\mainexp$ | $\mainexp~\mainexp$ \\
              &    &~~|~ $\this$.$\iden$ | $\this$.$\iden$ := $\mainexp$ ~|~ $\this$.$\syniden$ | $\this$.$\syniden$ := $\mainexp$ \\
              &    &~~|~ ${\downarrow}_s(\subexp, (){\to}\subexp)$ ~|~ ${\downarrow}_t(c, (){\to}\subexp)$
\end{tabular}

%% file: notes/def_fv_sugar_match.tex
\begin{tabular}{rlll}
$\mainexp_0 ~c~ \mainexp_1$
  &=& $c(\mainexp_0, \mainexp_1)$\\
($\mainexp_0$, ..., $\mainexp_n$)
  &=& $\mainexp_0 :: ... :: \mainexp_n :: ()$ \\
$\mainexp_0 \seq \mainexp_1$
  &=& $\letexp () = \mainexp_0\seq \mainexp_1$ \\
$\assert(\mainexp_0)\seq \mainexp_1$
  &=& $\letexp \true = \mainexp_0\seq \mainexp_1$ \\
$x \leftarrow \subexp_1 \seq \mainexp_2$
  &=& $\letexp x = \downarrow_t(() \rightarrow \subexp_1)\seq \mainexp_2$ \\
$\ifletexp x = \mainexp_1 \thenexp \subexp_2 \elseexp \subexp_3$
  &=& $\try(\mainexp_1, x \rightarrow \subexp_2, () \rightarrow \subexp_3)$ \\
\end{tabular}

%% file: notes/syntax2.tex
\begin{tabular}{llrll}
(frame)            &$F$         &::=& ${\downarrow}_s (\square, () \to e) ~|~ \square~e ~|~ e~\square ~|~ \square::e ~|~ e::\square$ \\
                   &            &~|~& $\letexp x=\square;~e ~|~ \letexp x=e;~\square ~|~ \this.\iden := \square ~|~ \this.\syniden := \square$ \\
(atom)             &$\Finv$     &::=& $\this.\iden ~|~ \this.\syniden ~|~ c ~|~ id ~|~ x \to e$ \\[.7em]
(context)          &$E$         &::=& $\square$ | $E$ :: $\mainexp$ | $v$ :: $E$ | $E$ $\mainexp$ | $v$ $E$
                                    | $\letexp x = E$; $\mainexp$ | $\this$.$\iden$ := $E$ | $\this$.$\syniden$ := $E$ \\[.7em]
(event)            &$\trace$    &::=& $\wrevent$($c,i,v$) $\trace$ | $\msgevent$($c,v$) $\trace$  | $()$ \\
(synthetic event)  &$\syntrace$ &::=& $\wrevent$($c,j,v$) $\syntrace$ | $()$ \\
(configuration)   &$C$          &::=& $\trace$; $\syntrace$; $c$; $\mainexp$\\
\end{tabular}

%% file: notes/def_mnf.tex
$\begin{array}{lll}
\mnfp[\decl;\syndecl;\mainexp] &= \decl;\syndecl; \trampoline(\mnfe[(\text{Done}, \mainexp)]) \\
\mnfe[F~e] &= \assoc[\letexp id_0 {=} e\seq \mnfe[F~id_0]] \\
\mnfe[\Finv] &= \Finv \\[.5em]
\assoc[\letexp x_0{=}(\letexp x_1{=}\mainexp_1\seq \mainexp_0)\seq \mainexp_2] &= \assoc[\letexp x_1{=}\mainexp_1\seq \assoc[\letexp x_0{=}\mainexp_0\seq \mainexp_2]]\\
\assoc[m] &= m\\
\end{array}$

%% file: notes/def_guard.tex
$\begin{array}{lll}
\guardp[\decl; \syndecl; \trampoline(\mainexp)] &\hspace{-.2cm}=~ \decl; \syndecl; \trampoline(\guarde[\mainexp])
\\[0.7em]
\guarde\left(\begin{array}{l}
  x \gets_s (\subexp_0, () \to \subexp_1)\seq\\
  \mainexp_2
\end{array}\right) &\hspace{-.2cm}=~ \left(\begin{array}{l}
  \this.\whovar := \subexp_0\seq \this.\statevar := c\seq \\
  x \gets_s (() \to \mathsf{true}, () \to \subexp_1)\seq \\
  \assert(\this.\statevar == c \,\operatorname{\&\&}\,\\
  \phantom{\assert(}\this.\whovar(\this.\sendervar))\seq \\
  \this.\statevar := 0\seq \guarde[\mainexp_2]
\end{array}\right)\\
&~\phantom{=}~~ \text{where }c ~\freshvar \\
\guarde[\letexp x = \subexp_0\seq \mainexp_1] &\hspace{-.2cm}=~ \letexp x = \subexp_0\seq \guarde[\mainexp_1]\\
\guarde[\mainexp] &\hspace{-.2cm}=~ \mainexp\\
\end{array}$

%% file: notes/def_cps.tex
$\begin{array}{lll}
\dtlp[d; b; \trampoline(m)] &~=~ d; b; \trampoline(\dtle[m])
\\[0.7em]
\dtle[x \gets_s (() \to \mathsf{true}, e_0)\seq m_1] &~=~ e_0 \bind (x \to \dtle[m_1])\\
\dtle[\letexp x = e_0\seq m_1] &~=~ \letexp x = e_0\seq \dtle[m_1]\\
\dtle[m] &~=~ m\\
\end{array}$

%% file: notes/def_defun.tex
$\begin{array}{llll}

\defunOuter[\decl; \syndecl; \subexp] &~=~& \defun[\decl; \coclfn(\syndecl, id, \assert(\false), \assert(\false)); \subexp]
\\[0.2em]
&& \begin{array}{ll} \text{where} & id \freshvar \end{array}
\\[1em]

\defun\left(\begin{array}{l}
  \decl; \coclfn(\syndecl, id,\\
  ~~\subexp_{1,alt},\\
  \\
  ~~\subexp_{2,alt}\\
  ~~);\\
  (() \to \subexp_1) \bind (x \to \subexp_2)
\end{array}\right) &~=~& \left(\begin{array}{l}
  \decl; \coclfn(\syndecl, id,\\
  ~~\ifletexp (c :: \fv(() \to \subexp_1)) = id\\
  ~~~~\thenexp \subexp_1 \elseexp \subexp_{1,alt}',\\
  ~~\ifletexp (c :: x :: \fv(x \to \subexp_2')) = id\\
  ~~~~\thenexp \subexp_2' \elseexp \subexp_{2,alt}');\\
  (\mathsf{More},~ c :: \fv(() \to \subexp_1),~ c :: \fv(x \to \subexp_2'))
\end{array}\right)
\\[2.5em]
&& \begin{array}{ll}
     \text{where} & c \freshvar\\
     \text{and}   & \decl; \coclfn(\syndecl, id, \subexp_{1,alt}', \subexp_{2,alt}'); \subexp_2' =\\
                  & ~\defun[\decl; \coclfn(\syndecl, id, \subexp_{1,alt}, \subexp_{2,alt}); \subexp_2]\\
   \end{array}
\\[1.8em]

\defun\left(\begin{array}{l}
  \decl; \coclfn(\syndecl, id,\\
  ~~\subexp_{1,alt}, \subexp_{2,alt});\\
  \letexp x = \subexp_0; ~ \subexp_1
\end{array}\right) &~=~& \decl; \coclfn(\syndecl, id, \subexp_{1,alt}, \subexp_{2,alt}); \letexp x = \subexp_0; ~ \defun(\subexp_1)
\\
\defun\left(\begin{array}{l}
  \decl; \coclfn(\syndecl, id,\\
  ~~\subexp_{1,alt}, \subexp_{2,alt});\\
  \subexp
\end{array}\right) &~=~& \decl; \coclfn(\syndecl, id, \subexp_{1,alt}, \subexp_{2,alt}); \subexp \\[.5em]
\coclfn(\syndecl, id, \subexp_{1,alt}, \subexp_{2,alt}) &~=~&
  \clVal \clfnvar = id \to \subexp_{1,alt};\\
  && \coVal \cofnvar = id \to \subexp_{2,alt}; \syndecl

\end{array}$

%% file: notes/def_match.tex
\begin{tabular}{rll}
$[c \Mapsto{} c]$ &= $[]$                 \\
$[id \Mapsto{} v]$ &= $[id \mapsto v]$                 \\
$[(\subexp_0$ :: $\subexp_1$) $\Mapsto$ ($\subexp'_0$ :: $\subexp'_1)]$ &= $[\subexp_0 \Mapsto{} \subexp'_0] \cdot [\subexp_1 \Mapsto{} \subexp'_1]$ \\
\end{tabular}

%% file: notes/dynamic_semantics.txt.tex
$\begin{array}{lllll}
(\textcolor{blue}{\textsc{Rgs}})&{\trace};{\syntrace};0;~{\downarrow}_s(v,()~{\to}~{\subexp})&\textcolor{blue}{~{\to}_g}&{\trace};{\syntrace};0;~{\downarrow}_t(c,()~{\to}~{\subexp})&\text{{\ensuremath{\mathsf{if}}}~}{\trace};{\syntrace};0;v(c)~~{\to}^*~{\trace};{\syntrace};0;{\ensuremath{\mathsf{true}}}\\
(\textcolor{red}{\textsc{Rbs}})&{\trace};{\syntrace};0;~{\downarrow}_s(v,()~{\to}~{\subexp})&\textcolor{red}{~{\to}_b}&{\trace};{\syntrace};0;~{\downarrow}_t(\textcolor{red}{c},()~{\to}~{\subexp})&\\
(\textsc{Rtm})&{\trace};{\syntrace};0;~{\trampoline}\left(\begin{array}{l}{\ensuremath{\mathsf{More}}},\\v_1~{\ensuremath{\mathsf{::}}}~{\subexp}_1,\\v_2~{\ensuremath{\mathsf{::}}}~{\subexp}_2\end{array}\right)&~{\to}&{\trace};{\syntrace};0;\left({\begin{array}{l}id~=~{\downarrow}_t(\textcolor{red}{c},~\this.{\ensuremath{\mathsf{clfn}}}(v_1~{\ensuremath{\mathsf{::}}}~{\subexp}_1));\\{\trampoline}(\this.{\ensuremath{\mathsf{cofn}}}(v_2~{\ensuremath{\mathsf{::}}}~id~{\ensuremath{\mathsf{::}}}~{\subexp}_2))\end{array}}\right)&\\
(\textsc{Rtd})&{\trace};{\syntrace};0;~{\trampoline}({\ensuremath{\mathsf{Done}}},~v)&~{\to}&{\trace};{\syntrace};0;~v&\\
(\textcolor{blue}{\textsc{Rg}})&{\trace};{\syntrace};0;~{\downarrow}_t(c,()~{\to}~{\subexp})&\textcolor{blue}{~{\to}_g}&{\trace};{\syntrace}~{\ensuremath{\mathsf{msg}}}(c,v)~{\ensuremath{\mathsf{wr}}}(0,{\ensuremath{\mathsf{sender}}},c);0;~v&\text{{\ensuremath{\mathsf{if}}}~}{\trace};{\syntrace};c;{\subexp}~~{\to}^*~{\trace}';{\syntrace}';c;v\\
(\textcolor{red}{\textsc{Rb}})&{\trace};{\syntrace};0;~{\downarrow}_t(c,()~{\to}~{\subexp})&\textcolor{red}{~{\to}_b}&{\trace};{\syntrace}~{\ensuremath{\mathsf{msg}}}(c,\textcolor{red}{v'})~{\ensuremath{\mathsf{wr}}}(0,{\ensuremath{\mathsf{sender}}},c);0;~\textcolor{red}{v'}&\\
[0.5em]
\end{array}$

%% file: notes/dynamic_semantics2.txt.tex
$\begin{array}{lllll}
(\textsc{Re})&\hspace{-.0cm}{\trace};{\syntrace};0;~E[{\mainexp}]&\hspace{-.0cm}~{\to}~{\trace}'\!;{\syntrace}'\!;0;~E[{\mainexp}']&\hspace{-.0cm}\text{{\ensuremath{\mathsf{if}}}~}{\trace};{\syntrace};0;{\mainexp}~~{\to}~{\trace}'\!;{\syntrace}'\!;0;{\mainexp}'&\\
(\textsc{Rget}^\circ)&\hspace{-.0cm}{\trace};{\syntrace};c;~\this.{\iden}&\hspace{-.0cm}~{\to}~{\trace};{\syntrace};c;~v&\hspace{-.0cm}\text{{\ensuremath{\mathsf{if}}}~}{\ensuremath{\mathsf{wr}}}(c,{\iden},v)~{\in}~{\trace}&\\
(\textsc{Rget}^\dagger)&\hspace{-.0cm}{\trace};{\syntrace};c;~\this.{\syniden}&\hspace{-.0cm}~{\to}~{\trace};{\syntrace};c;~v&\hspace{-.0cm}\text{{\ensuremath{\mathsf{if}}}~}{\ensuremath{\mathsf{wr}}}(c,{\syniden},v)~{\in}~{\syntrace}&\\
(\textsc{Rset}^\circ)&\hspace{-.0cm}{\trace};{\syntrace};c;~\this.{\iden}~{\ensuremath{\mathsf{:=}}}~v&\hspace{-.0cm}~{\to}~{\trace}~{\ensuremath{\mathsf{wr}}}(c,{\iden},v);{\syntrace};c;~()&&\\
(\textsc{Rset}^\dagger)&\hspace{-.0cm}{\trace};{\syntrace};c;~\this.{\syniden}~{\ensuremath{\mathsf{:=}}}~v&\hspace{-.0cm}~{\to}~{\trace};{\syntrace}~{\ensuremath{\mathsf{wr}}}(c,{\syniden},v);c;~()&&\\
(\textsc{Rop})&\hspace{-.0cm}{\trace};{\syntrace};c;~v_0~+~v_1&\hspace{-.0cm}~{\to}~{\trace};{\syntrace};c;~v'&\hspace{-.0cm}\text{{\ensuremath{\mathsf{if}}}~}~v'~=~v_0~+~v_1&\\
(\textsc{Rt})&\hspace{-.0cm}{\trace};{\syntrace};c;\left({\begin{array}{l}{\ensuremath{\mathsf{if}}}~{\ensuremath{\mathsf{let~}}}x=v\\{\ensuremath{\mathsf{then}}}~{\subexp}_0~{\ensuremath{\mathsf{else}}}~{\subexp}_1\end{array}}\right)&\hspace{-.0cm}~{\to}~{\trace};{\syntrace};c;~{\subexp}_0'&\hspace{-.0cm}\text{{\ensuremath{\mathsf{if}}}~}~{\subexp}_0'~=~[x{\Mapsto}v]~{\subexp}_0&\\
(\textsc{Rf})&\hspace{-.0cm}{\trace};{\syntrace};c;\left({\begin{array}{l}{\ensuremath{\mathsf{if}}}~{\ensuremath{\mathsf{let~}}}x=v\\{\ensuremath{\mathsf{then}}}~{\subexp}_0~{\ensuremath{\mathsf{else}}}~{\subexp}_1\end{array}}\right)&\hspace{-.0cm}~{\to}~{\trace};{\syntrace};c;~{\subexp}_1&\hspace{-.0cm}\text{otherwise}&\\
(\textsc{Rapp})&\hspace{-.0cm}{\trace};{\syntrace};c;~(x~{\to}~{\subexp})~v&\hspace{-.0cm}~{\to}~{\trace};{\syntrace};c;~{\subexp}'&\hspace{-.0cm}\text{{\ensuremath{\mathsf{if}}}~}~{\subexp}'~=~[x{\Mapsto}v]~{\subexp}&\\
(\textsc{Rlet})&\hspace{-.0cm}{\trace};{\syntrace};c;~(x=v;~{\mainexp})&\hspace{-.0cm}~{\to}~{\trace};{\syntrace};c;~{\mainexp}'&\hspace{-.0cm}\text{{\ensuremath{\mathsf{if}}}~}~{\mainexp}'~=~[x{\Mapsto}v]~{\mainexp}&\\
\end{array}$

%% file: paper/evaluation.tex
\section{Evaluation}
\label{sec:discussion}

We evaluate \lang{} along two research questions:

\begin{enumerate}[label=RQ\arabic*]
	\setlength{\itemindent}{2em}
	\item \label{rq-expressiveness} \emph{Does \lang{} support the most common dApps scenarios?}

	\item \label{rq-efficiency} \emph{Do \lang{}'s abstractions
        affect performance?}
\end{enumerate}

\subparagraph{Case Studies and Expressiveness (\ref*{rq-expressiveness}).}
Five classes of smart contract
applications have been identified \cite{bartoletti2017empirical}: Financial, Wallet, Notary, Game, and Library.
To answer \ref{rq-expressiveness}, we implemented at least one 
case study per category in \lang{}.
We implemented an ERC-20 Token,\!%
\footnote{A study investigating all blocks mined until Sep 15, 2018~\cite{oliva2020exploratory},
	found that 72.9\,\% of the high-activity contracts are token contracts compliant to ERC-20 or ERC-721,
	with an accumulated market capitalization of 12.7\,B~USD.}
a Crowdfunding, and an Escrowing dApp as representatives of financial dApps.
We cover \emph{wallets}
by implementing a multi-signature wallet, a special type of wallet that provides a transaction voting mechanism by only executing transactions, which are signed by a fixed fraction of the set of owners.
We implemented a general-purpose \emph{notary} contract enabling
users to store arbitrary data, e.g., document hashes or images, together with a submission timestamp and the data owner.
As \emph{games}, we implemented TicTacToe (Section~\ref{sec:the-language}),
Rock-Paper-Scissors, Hangman and Chinese Checkers.
Rock-Paper-Scissors makes use of timed commitments~\cite{DBLP:conf/sp/AndrychowiczDMM14},
i.e., all parties commit to a random seed share and open it
after all commitments have been posted.
The same technique can be used to generate randomness for dApps in a secure way.
To reduce expensive code deployment, developers outsource 
commonly used 
logic to library contracts.
We demonstrate \emph{library}-based development in \lang by including a TicTacToe library
to our case studies and another 
TicTacToe dApp
which uses that library instead of deploying the logic itself.

We also implemented a state channel~\cite{MBBKM19, DziembowskiFH18, DziembowskiEFHH19}
for TicTacToe in \lang,
which is an example for the class of \emph{scalability solutions}
that have emerged more recently. State channels enable parties to move parts of their dApp
to a non-blockchain consensus system, falling-back to the blockchain in case of disputes,
thereby making the dApps more efficient where possible.

Our case studies are between 1\,K and 7.5\,K bytes
which is a representative size: %
	Smart contracts are not built for large-scale applications
	since the gas model limits the maximal computation and storage volumes and
	causes huge fees for complex applications.
	The median (average, lower quantile, upper quantile) of the bytecode size of distinct contracts deployed at the time of writing is at 4\,K (5.5\,K, 1.5\,K, 7.5\,K)~\cite{google-cloud-big-query-contract-size}. %
We 
further elaborate on the case studies
including a comparison of the lines of code in Prisma
compared to the equivalent lines in Solidity and Javascript
in Appendix \ref{appendix:case-studies}.
Our case studies demonstrate that \lang{} supports most common dApps scenarios.

\begin{figure*}
    \input{paper/graphs/deploymentCosts}
    \input{paper/graphs/deploymentOverhead}
    \input{paper/graphs/executionCosts}
    \input{paper/graphs/executionOverhead}
    \caption{
        The cost of abstraction. Gas overhead of contracts written with \lang{} vs. Solidity. \\
        (The right plot displays minima, averages, maxima.)
    }
    \label{fig:gas-costs}
    \label{fig:gas-overhead}
    \vspace{-1em}
\end{figure*}
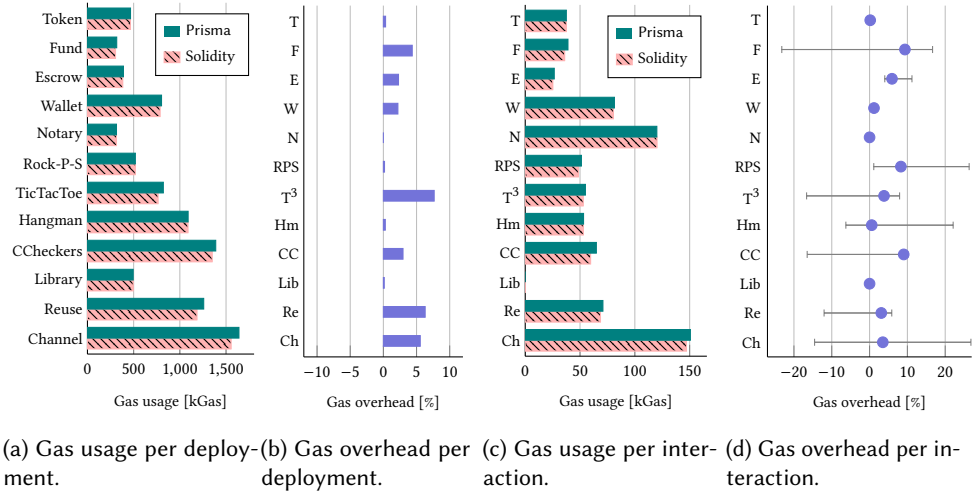

\subparagraph{Performance of Prisma DApps (\ref*{rq-efficiency}).}
Performance on the Ethereum blockchain is usually measured in terms of an Ethereum-specific metric called \textit{gas}.
Each instruction of the Ethereum Virtual Machine (EVM) consumes gas which needs to be paid for by the users in form of transaction fees credited to the miner.
We refer to the Ethereum yellow paper~\cite{ethereum-yellow-paper} for an overview of the gas consumption of the different EVM instructions.
To answer \ref{rq-efficiency}, we implement 
our case studies in both \lang{} and in Solidity/JavaScript
and compare their 
gas consumption.
Unlike prior work, 
we do not model a custom gas structure, but consider the real EVM gas costs \cite{wood2014ethereum}.

\emph{Experimental setup.}
We execute each case study on different inputs
to achieve different execution patterns that
cover all contract functions.
Each contract invocation that includes parameters with various sizes
(e.g., dynamic length arrays) is executed with a range of realistic inputs,
e.g., for Hangman, we consider several words (2 to 40 characters)
and different order of guesses, covering games in which the guesser wins
and those in which they lose.
\lang{} and Solidity/JavaScript implementations are executed on the same inputs.

We perform the measurements on a local setup.
As the execution
in the Ethereum VM is deterministic,
a single measurement suffices.
We set up the local Ethereum blockchain with \textit{Ganache} (Core v2.13.2)
on the latest supported hard fork (Muir Glacier).
All contracts are compiled to EVM byte code with \textit{solc} (v0.8.1, optimized on 20 runs).
We differentiate contract deployment and contract interaction.
Deployment means uploading the contract to the blockchain and initializing its state,
which occurs just once per application instance.
A single instance typically involves several contract interactions,
i.e., transactions calling public contract functions.

\emph{Results.}
Fig.~\ref{fig:gas-costs} shows the average gas consumption of contract deployment (Fig.~\ref{fig:eval-deployment-cost}) and interaction (Fig.~\ref{fig:eval-execution-cost}) as well as the relative overhead of Prisma vs.\@ Solidity/JS of deployment (Fig.~\ref{fig:eval-deployment-overhead}) and interaction (Fig.~\ref{fig:eval-execution-overhead}).
As the gas consumption of contract invocations depends heavily on the executed function,
the contract state, and the inputs, we provide the maximal, minimal and averaged overhead. 
The results show that the average gas consumption
of \lang{} is close to the one of Solidity/JS.
Our compiler achieves a deployment overhead of maximally 6\,\% (TicTacToe) or 86\,K gas (TicTacToe Channel).
The interaction overhead is below 10\,\% for all case studies
which at most amounts to 3.55\,K gas.\footnote{equals 0.59~USD based on gas price and exchange course of April 15, 2021}

\lang{}'s deployment overhead is mainly due to the automated flow control.
To guarantee correct execution, \lang{} manages a state variable
for dApps with more than one state.
The storage reserved for and the code deployed to maintain the state variable
cause a constant cost
of around 45\,K gas.
In Solidity, developers manually check whether flow control is needed
and, if so, may derive the state from existing
contract variables to avoid a state variable if possible.

The Token, Notary, Wallet and Library case studies do not require flow control:
each function can be called by any client at any time.
Hence, their overhead is small.
Escrow, Hangman and Rock-Paper-Scissors require a state variable, also in Solidity
-- which partially compensates the overhead of \lang{}'s automated flow control.
Crowdfunding, Chinese Checkers, TicTacToe (Library and Channel) do not require an explicit state variable in Solidity, as the state can be derived from the contract variables,
e.g., the number of moves.
Thus, these case studies have the largest deployment overhead.

While the average relative interaction overhead is constantly 
below 10\,\%, some contract invocations are far above,
e.g., in Crowdfunding, TicTacToe Channel, and Rock-Paper-Scissors.
Yet, case studies with such spikes also involve interactions
that are executed within the same dApp
instance with a negative overhead and amortize the costs of more costly transactions.
These deviations are also mainly due to 
automated flow control.
In EVM, setting a zero variable to some non-zero value costs more
gas (20\,K gas) than changing its value (5\,K gas)~\cite{wood2014ethereum}, and setting the value to zero saves gas.
Occupying and releasing storage via the state variable can cost or
save gas in a different way than in traditional dApps
without an explicit state variable, leading to different (and even negative) overhead in different transactions.

Besides the gas-overhead, we also consider the time-overhead of Prisma. In Ethereum, the estimated confirmation time for transactions is 3-5 minutes (assuming no congestion), which makes the number of on-chain interactions dominate the total execution time. As Prisma preserves the number of on-chain interactions, we assess the time-overhead of Prisma, if any, to be negligible.	

Note that per se it is not possible to achieve a better gas consumption in \lang{} than in Solidity -- every contract compiled from \lang{} can be implemented in Solidity.
Given the abstractions we offer beyond the traditional development approach,
and the sensibility of smart contracts
to small changes in instructions, we conclude that our abstractions come
with acceptable overhead.
We are confident that further engineering effort can eliminate the observed overhead.

\emph{Threats to validity.}
The main threat
is that the manually written code may be optimized better or worse than the code generated by the compiler.
We mitigate this threat by applying all gas optimizations, our compiler performs automatically,
to the Solidity implementations.
An external threat is that changes in the gas pricing of Ethereum may affect our evaluation.
For reproducibility, we state the Ethereum version (hard fork), we used in the paper.

\begin{table}\figurehead\centering
  \caption{Related work.}\label{tab:related-work-style}
  \small
  \begin{tabular}{lllll}
    \toprule
    Language & Encoding & Perspective  & Protocol      \\
    \midrule
    Solidity & FSM      & Local        & Assertions    \\
    Obsidian & FSM      & Local        & Type states   \\
    Nomos    & MNF      & Local        & Session types \\
    \midrule
    \lang{} & DS       & Global       & Control flow  \\
    \bottomrule
  \end{tabular}
  \vspace{-1.2em}
\end{table}

%% file: paper/graphs/deploymentCosts.tex
\begin{subfigure}[b]{.24\linewidth}\centering
	\begin{tikzpicture}
	\tikzset{font=\tiny}
	\begin{axis}[xbar=0pt,
	height = 6.2cm, width = 3.8cm, bar width=4pt,
	major y tick style = transparent,
	xmajorgrids = true,
	xlabel = {Gas usage [kGas]},
	symbolic y coords={Channel, Reuse, Library, CCheckers, Hangman, TicTacToe, Rock-P-S, Notary, Wallet, Escrow, Fund, Token},
	ytick = data,
	y tick label style = { xshift=6pt, rotate=0, anchor=east },
	scaled x ticks = false,
	enlarge y limits=0.05,
	xmin=0,
	legend cell align=left,
	legend pos=north east,
	legend image code/.code={\draw[#1, draw=none] (-3pt,-2pt) rectangle (3pt,2pt);},
	reverse legend,
	axis y line*=left,
	axis x line*=bottom,
	]
	\addplot[white!70!red, fill=white!70!red, mark=none, postaction={pattern = north west lines}]
	coordinates {
		(\varMultiSigManualDeployAvg,Wallet) +- (5, 5)
		(\varTokenManualDeployAvg,Token)
		(\varCrowdfundingManualDeployAvg,Fund)
		(\varEscrowManualDeployAvg,Escrow)
		(\varTTTChannelManualDeployAvg,Channel)
		(\varNotaryManualDeployAvg,Notary)
		(\varHangmanManualDeployAvg,Hangman)
		(\varTTTManualDeployAvg,TicTacToe)
		(\varRPSManualDeployAvg,Rock-P-S)
		(\varChineseCheckersManualDeployAvg,CCheckers)
		(\varTTTLibraryManualDeployAvg,Library)
		(\varTTTViaLibManualDeployAvg,Reuse)
	};
	
	\addplot[teal, fill=teal, mark=none]
	coordinates {
		(\varMultiSigCompiledDeployAvg,Wallet)
		(\varTokenCompiledDeployAvg,Token)
		(\varCrowdfundingCompiledDeployAvg,Fund)
		(\varEscrowCompiledDeployAvg,Escrow)
		(\varTTTChannelCompiledDeployAvg,Channel)
		(\varNotaryCompiledDeployAvg,Notary)
		(\varHangmanCompiledDeployAvg,Hangman)
		(\varTTTCompiledDeployAvg,TicTacToe)
		(\varRPSCompiledDeployAvg,Rock-P-S)
		(\varChineseCheckersCompiledDeployAvg,CCheckers)
		(\varTTTLibraryCompiledDeployAvg,Library)
		(\varTTTViaLibCompiledDeployAvg,Reuse)
	};
	
	\legend{Solidity,\lang{}}
	\end{axis}
	\end{tikzpicture}%
	\caption{Gas usage per deployment.}
	\label{fig:eval-deployment-cost}
\end{subfigure}%

%% file: paper/graphs/deploymentOverhead.tex
\begin{subfigure}[b]{.20\linewidth}\centering
	\begin{tikzpicture}
	\tikzset{font=\tiny}
	\begin{axis}[xbar,
	height = 6.2cm, width = 3.7cm, bar width=4pt,
	major y tick style = transparent,
	xmajorgrids = true,
	xlabel = {Gas overhead [\%]},
	symbolic y coords={Ch, Re, Lib, CC, Hm, T$^3$, RPS, N, W, E, F, T},
	ytick = data,
	y tick label style = { xshift=6pt, rotate=0, anchor=east },
	scaled x ticks = false,
	enlarge y limits=0.05,
	legend cell align=left,
	legend pos=north east,
	legend image code/.code={\draw[#1, draw=none] (-3pt,-2pt) rectangle (3pt,2pt);},
	reverse legend,
	xmin=-12, xmax=12, xtick={-10,-5,...,10},
	axis y line*=left,
	axis x line*=bottom,
	]
	\addplot[lightgray!60!blue, fill=lightgray!60!blue, mark=none]
	coordinates {
		(\varMultiSigDiffDeployRelAvg,W)
		(\varTokenDiffDeployRelAvg,T)
		(\varCrowdfundingDiffDeployRelAvg,F)
		(\varEscrowDiffDeployRelAvg,E)
		(\varTTTChannelDiffDeployRelAvg,Ch)
		(\varNotaryDiffDeployRelAvg,N)
		(\varHangmanDiffDeployRelAvg,Hm)
		(\varTTTDiffDeployRelAvg,T$^3$)
		(\varRPSDiffDeployRelAvg,RPS)
		(\varChineseCheckersDiffDeployRelAvg,CC)
		(\varTTTLibraryDiffDeployRelAvg,Lib)
		(\varTTTViaLibDiffDeployRelAvg,Re)
	};
	\end{axis}
	
	\end{tikzpicture}%
	\caption{Gas overhead per deployment.}
	\label{fig:eval-deployment-overhead}
\end{subfigure}

%% file: paper/graphs/executionCosts.tex
\begin{subfigure}[b]{.22\linewidth}\centering
	\begin{tikzpicture}
	\tikzset{font=\tiny}
	\begin{axis}[xbar=0pt,
	height = 6.2cm, width = 4cm, bar width=4pt,
	major y tick style = transparent,
	xmajorgrids = true,
	xlabel = {Gas usage [kGas]},
	symbolic y coords={Ch, Re, Lib, CC, Hm, T$^3$, RPS, N, W, E, F, T},
	ytick = data,
	y tick label style = { xshift=6pt, rotate=0, anchor=east },
	scaled x ticks = false,
	enlarge y limits=0.05,
	xmin=0,
	legend cell align=left,
	legend pos=north east,
	legend image code/.code={\draw[#1, draw=none] (-3pt,-2pt) rectangle (3pt,2pt);},
	reverse legend,
	axis y line*=left,
	axis x line*=bottom,
]
	\addplot[white!70!red, fill=white!70!red, mark=none, postaction={pattern = north west lines}]
	coordinates {
		(\varMultiSigManualExecutionAvg,W) +- (5, 5)
		(\varTokenManualExecutionAvg,T)
		(\varCrowdfundingManualExecutionAvg,F)
		(\varEscrowManualExecutionAvg,E)
		(\varTTTChannelManualExecutionAvg,Ch)
		(\varNotaryManualExecutionAvg,N)
		(\varHangmanManualExecutionAvg,Hm)
		(\varTTTManualExecutionAvg,T$^3$)
		(\varRPSManualExecutionAvg,RPS)
		(\varChineseCheckersManualExecutionAvg,CC)
		(\varTTTViaLibManualExecutionAvg,Re)
		(0,Lib)
	};
	
	\addplot[teal, fill=teal, mark=none]
	coordinates {
		(\varMultiSigCompiledExecutionAvg,W)
		(\varTokenCompiledExecutionAvg,T)
		(\varCrowdfundingCompiledExecutionAvg,F)
		(\varEscrowCompiledExecutionAvg,E)
		(\varTTTChannelCompiledExecutionAvg,Ch)
		(\varNotaryCompiledExecutionAvg,N)
		(\varHangmanCompiledExecutionAvg,Hm)
		(\varTTTCompiledExecutionAvg,T$^3$)
		(\varRPSCompiledExecutionAvg,RPS)
		(\varChineseCheckersCompiledExecutionAvg,CC)
		(\varTTTViaLibCompiledExecutionAvg,Re)
		(0,Lib)
	};
	
	\legend{Solidity,\lang{}}
	\end{axis}
	\end{tikzpicture}
	\caption{Gas usage per interaction.}
	\label{fig:eval-execution-cost}
\end{subfigure}%

%% file: paper/graphs/executionOverhead.tex
\pgfplotsset{boxplot/every whisker/.style={gray}}
\begin{subfigure}[b]{.22\linewidth}\centering
	\begin{tikzpicture}
	\tikzset{font=\tiny}
	\begin{axis}[
	height = 6.2cm, width = 4.3cm,
	xmin=-27, xmax=27, xtick={-20,-10,...,30},
	boxplot/draw direction=x,
	scaled x ticks = false,
	xmajorgrids = true,
	xlabel style={align=center},
	xlabel={Gas overhead [\%]},
	boxplot={
		draw position={0.5 + \plotnumofactualtype},
		box extend=0.3
	},
	cycle list={{black,fill=lightgray!60!blue,mark=none},{black,fill=lightgray!60!blue,mark=none}},
	ymin=0, ymax=12, ytick={0,1,2,...,12},
	y tick label as interval,
	yticklabels={Ch, Re, Lib, CC, Hm, T$^3$, RPS, N, W, E, F, T},
	y tick label style = { xshift=1pt, rotate = 0, anchor = east },
	ytick style={draw=none},
	name=boundary,
	axis y line*=left,
	axis x line*=bottom,
]

	\addplot+[mark = *, mark options = {draw=lightgray!60!blue, fill=lightgray!60!blue},
	boxplot prepared={
		upper whisker=\varTTTChannelDiffExecRelBoxUW,
		lower whisker=\varTTTChannelDiffExecRelBoxLW
	},
	] coordinates {(0,\varTTTChannelDiffExecRelAvg)};
	
	\addplot+[mark = *, mark options = {draw=lightgray!60!blue, fill=lightgray!60!blue},
	boxplot prepared={
		upper whisker=\varTTTViaLibDiffExecRelBoxUW,
		lower whisker=\varTTTViaLibDiffExecRelBoxLW
	},
	] coordinates {(0,\varTTTViaLibDiffExecRelAvg)};
	
	\addplot+[mark = *, mark options = {draw=lightgray!60!blue, fill=lightgray!60!blue},
	boxplot prepared={
		upper whisker=0,
		lower whisker=0
	},
	] coordinates {(0,0)};
	
	\addplot+[mark = *, mark options = {draw=lightgray!60!blue, fill=lightgray!60!blue},
	boxplot prepared={
		upper whisker=\varChineseCheckersDiffExecRelBoxUW,
		lower whisker=\varChineseCheckersDiffExecRelBoxLW
	},
	] coordinates {(0,\varChineseCheckersDiffExecRelAvg)};

	\addplot+[mark = *, mark options = {draw=lightgray!60!blue, fill=lightgray!60!blue},
	boxplot prepared={
		upper whisker=\varHangmanDiffExecRelBoxUW,
		lower whisker=\varHangmanDiffExecRelBoxLW
	},
	] coordinates {(0,\varHangmanDiffExecRelAvg)};

	\addplot+[mark = *, mark options = {draw=lightgray!60!blue, fill=lightgray!60!blue},
	boxplot prepared={
		upper whisker=\varTTTDiffExecRelBoxUW,
		lower whisker=\varTTTDiffExecRelBoxLW
	},
	] coordinates {(0,\varTTTDiffExecRelAvg)};

	\addplot+[mark = *, mark options = {draw=lightgray!60!blue, fill=lightgray!60!blue},
	boxplot prepared={
		upper whisker=\varRPSDiffExecRelBoxUW,
		lower whisker=\varRPSDiffExecRelBoxLW
	},
	] coordinates {(0,\varRPSDiffExecRelAvg)};

	\addplot+[mark = *, mark options = {draw=lightgray!60!blue, fill=lightgray!60!blue},
	boxplot prepared={
		upper whisker=\varNotaryDiffExecRelBoxUW,
		lower whisker=\varNotaryDiffExecRelBoxLW
	},
	] coordinates {(0,\varNotaryDiffExecRelAvg)};
	
	\addplot+[mark = *, mark options = {draw=lightgray!60!blue, fill=lightgray!60!blue},
	boxplot prepared={
		upper whisker=\varMultiSigDiffExecRelBoxUW,
		lower whisker=\varMultiSigDiffExecRelBoxLW
	},
	] coordinates {(0,\varMultiSigDiffExecRelAvg)};

	\addplot+[mark = *, mark options = {draw=lightgray!60!blue, fill=lightgray!60!blue},
	boxplot prepared={
		upper whisker=\varEscrowDiffExecRelBoxUW,
		lower whisker=\varEscrowDiffExecRelBoxLW
	},
	] coordinates {(0,\varEscrowDiffExecRelAvg)};
	
	\addplot+[mark = *, mark options = {draw=lightgray!60!blue, fill=lightgray!60!blue},
	boxplot prepared={
		upper whisker=\varCrowdfundingDiffExecRelBoxUW,
		lower whisker=\varCrowdfundingDiffExecRelBoxLW
	},
	] coordinates {(0,\varCrowdfundingDiffExecRelAvg)};

	\addplot+[mark = *, mark options = {draw=lightgray!60!blue, fill=lightgray!60!blue},
	boxplot prepared={
		upper whisker=\varTokenDiffExecRelBoxUW,
		lower whisker=\varTokenDiffExecRelBoxLW
	},
	] coordinates {(0,\varTokenDiffExecRelAvg)};

	\end{axis}

	\end{tikzpicture}%
	\caption{Gas overhead per interaction.}
	\label{fig:eval-execution-overhead}
\end{subfigure}%

%% file: paper/discussion.tex
\newcommand{\myrightarrow}[1]{\mathrel{\raisebox{-4pt}{$\xrightarrow{#1}$}}}

\section{Discussion and Related Work}
\label{sec:background}

\subsection{Smart Contract Languages for Enforcing Protocols}
We compare \lang{} to Solidity,
Obsidian~\cite{Coblenz17, Coblenz19, Coblenz20}, and
Nomos~\cite{DasBHP19, Nomos21}.
We highlight these languages as those also address the 
correctness of the client--contract interactions.
Tab.~\ref{tab:related-work-style} overviews the features of the surveyed languages for
(a)~the \emph{perspective} of defining interacting parties,
(b)~the used \emph{encoding} of the interaction effects, 
and (c)~the method used to check the contract-client interaction \emph{protocol}.
Fig.~\ref{fig:solidity}, \ref{fig:type-states}, and~\ref{fig:session-types}
show code snippets in these languages,
each encoding the \emph{TicTacToe} state machine from Fig.~\ref{fig:funding-diagram}.
All three languages focus solely on the contract and do not state how clients are developed, 
hence only contract code is shown.

\begin{figure}\figurehead
  \centering
  \begin{minipage}{.6\textwidth}
  \input{paper/listings/obsidian_ttt.tex}
  \end{minipage}
  \caption{Obsidian.}\label{fig:type-states}
\end{figure}

\begin{figure}\figurehead
  \centering
  \begin{minipage}{.6\textwidth}
  \input{paper/listings/nomos_ttt.tex}
  \end{minipage}
  \caption{Nomos.}\label{fig:session-types}
\end{figure}

\begin{figure}
    \begin{mathpar}
        \inferrule[NomosR]{
            \Psi;\Gamma,(y{:}A) \vdash P :: (c : B)
        }{
            \Psi;\Gamma \vdash (y {\gets} \mathsf{recv}~c;~P) :: (c : A \multimap B)
        }
        
        \inferrule[NomosS]{
            \Psi;\Gamma \vdash P :: (c : B)
        }{
            \Psi;\Gamma,(w{:}A) \vdash (\mathsf{send}~c~w;~P) :: (c : A \otimes B)
        }
        
        \inferrule[Obsidian]{
            (\mathsf{transaction}~T~m(\overline{t.(s \text{>>} s')~x}) \{ ... \}) \in \text{members}_{t_0}
        }{
            \Delta,\overline{e{:}t.s}
            \vdash
            e_0.m(\overline{e}) : T \dashv \Delta,\overline{e{:}t.s'}
        }
    \end{mathpar}
    \caption{Excerpts of simplified Nomos and Obsidian typing rules.}\label{fig:typing}
\end{figure}

All three approaches take a \textbf{local perspective} on interacting parties:
Contract and clients are defined separately, and 
their interaction is encoded by explicit send and receive side effects.
In Solidity and Obsidian, receive corresponds to arguments and send to return values
of methods defined in the contract classes.
In Nomos, send and receive are expressed as procedures operating over a channel -- given a channel \lstinline!c!, sending and receiving is represented by explicit statements (\lstinline!x = recv c; ...! and \lstinline!send c x; ...!).

The approaches differ in the \textbf{encoding style}
of communication effects.
Solidity and Obsidian adopt an \emph{FSM-style encoding}:
Contract fields encode states, methods encode transitions.
The contract in Fig.~\ref{fig:solidity}
represents FSM states via the \lstinline!phase! field with initial state 
\lstinline!Funding! (Line~\ref{code:solidity-init-phase}).
The
\lstinline!Fund!, \lstinline!Move! and \lstinline!Payout! methods are transitions, e.g.,
\lstinline!Payout! transitions the contract into 
the final state  \lstinline!Closed! (Line~\ref{code:solidity-to-final-phase}).
The FSM-style encoding results in an implicitly-everywhere 
concurrent programming model, which is complex to reason about and unfitting for dApps
because the execution model of blockchains is inherently sequential -- all method invocations are brought into a world-wide total order.
Nomos adopts the \emph{monadic normal form~(MNF)}
via do-notation to order effects. 
While the implementation of TicTacToe  in FSM style requires three methods(\lstinline!Fund!, \lstinline!Move!, \lstinline!Payout! -- one per transition), we only need two methods in MNF-style (\lstinline!funding!, \lstinline!executing! -- one per state with multiple entry points), and a single method in DS-style (\lstinline!init!).
For instance, the sequence of states and transitions
$Executing
\hspace{-0.35em}\raiseoperator{-0.13ex}{\closeoverset[0.43ex]{Move(x, y)}\longrightarrow}
Finished
\hspace{-0.35em}\raiseoperator{-0.13ex}{\closeoverset[0.43ex]{Payout()}\longrightarrow}
Closed$
in Nomos
can be written sequentially in do-notation
by inlining the last function which only has a single entry point.
Still, do-notation can be cumbersome
(e.g., funding and executing in Nomos are separate methods
that cannot be inlined since they have multiple entry points and model loops).

All three languages require an \textbf{explicit protocol}
for governing the 
send--receive interactions, to
ensure that every send effect has a corresponding receive effect 
in an interacting -- separately defined -- party.
In Solidity, developers express the protocol via 
run-time assertions to guard 
against invoking the methods
in an incorrect order
(e.g., \lstinline!require(phase==Finished)! in Fig.~\ref{fig:solidity}, Line~\ref{code:solidity-require-funding}).
Unlike Solidity, which does not support statically checking protocol compliance,
Nomos and Obsidian employ \emph{behavioral typing} for static checks.
Deployed contracts may interact with 
third-party, potentially manipulated clients.
Compile-time checking alone
cannot provide security guarantees.
Yet, 
complementing run-time enforcement with static checks 
helps detecting cases that are guaranteed to fail at run time ahead of time.

\subparagraph{Obsidian}
Obsidian employs typestates to increase safety of contract--client communication.
Contracts define a number of typestates; 
A method call can change the typestate of an object,
and calling a 
method on a receiver that is in the wrong typestate results in a typing error.
Each method in Fig.~\ref{fig:type-states}
is annotated with the state in which it can be called, e.g., \lstinline!Payout! 
requires state
\lstinline!Finished!,
and transitions to
\lstinline!Closed! (Line~\ref{code:obsidian-payout}).

\subparagraph{Nomos} Nomos employs session types.
The session types \lstinline!Funding!, \lstinline!Executing!, \lstinline!Finished!
in Fig.~\ref{fig:session-types} encode the protocol.
Receiving a message is represented by a function type,
e.g., in the \lstinline!Funding! state,
we receive an integer \lstinline!int -> ...! (Line~\ref{code:nomos-type-funding}).
We respond by either repeating the funding (\lstinline!Funding!), or continuing
to the next state of the protocol (\lstinline!Executing!). This 
is represented by internal choice \lstinline!+{ ... }!
that takes multiple possible responses giving each of them a unique label (\lstinline!notenough! and \lstinline!enough!).
Type \lstinline!1! indicates the end of a protocol (Line~\ref{code:nomos-type-finished}).
The contract processes \lstinline!funding! (Line~\ref{code:nomos-process-funding}) and \lstinline!executing! (Line~\ref{code:nomos-process-executing}) implement the protocol.
The \lstinline!recv! operation (Line~\ref{code:nomos-receive}) takes a session-typed channel
of form \lstinline'T -> U', returns a value of type \lstinline!T!
and changes the type of the channel to \lstinline!U!.
A session type for internal choice (\lstinline!+{ ... }!),
requires the program to select
one of the offered labels (e.g., \lstinline!$$s.notenough! in Line~\ref{code:nomos-decision-notenough} and \lstinline!$$s.enough! in Line~\ref{code:nomos-decision-enough}),
e.g., in the left and right branch of a conditional statement.

\subparagraph{Type systems.}
\label{sec:typing-comparison}
We show excerpts of simplified typing rules
for Nomos and Obsidian (Fig.~\ref{fig:typing}).
Nomos rules
have the form $\Psi;\Gamma \vdash P :: (c{:}A)$.
A process $P$ offers a channel $c$ of type $A$ with
values in context $\Psi$ and channels in $\Gamma$.
We can see that variables change their type to model the linearity of session types
in the \textsc{NomosS} (and \textsc{NomosR}) rule:
Sending (and receiving) changes the type of the channel $c$
from $A {\multimap} B$ to $B$
(and $A {\otimes} B$ to $B$).
Obsidian rules have the form $\Delta \vdash e{:}t \dashv \Delta'$.
An expression $e$ has type $t$ in context $\Delta$ and changes $\Delta$ to $\Delta'$.
We can see that variables change their type on method invocation (\textsc{Obsidian}):
A method $m$ in class $t_0$
with arguments $e_i$ of type $t_i$,
returning $T$\!,
changes the type state of the arguments from $s_i$ to $s'_i$.
For \lang, instead, a standard judgement $\Gamma \vdash e:T$
suffices for communication.
Variables do not change their type.
$\mathsf{awaitCl}(p)\{b\}$ has type $T$ in context $\Gamma$
if $p$ is a predicate of $Addr$ and $b$ is a pair of $Ether$ and $T$:

\[
  \inferrule[Prisma]{
    \Gamma \vdash p:Addr \to Bool \and
    \Gamma \vdash b:Ether \times T
  }{
    \Gamma \vdash \mathsf{awaitCl} (p) \{ ~b~ \}:T
  }
\]

\subparagraph{\lang.}
As shown in Tab.~\ref{tab:related-work-style}, \lang occupies an unexplored point in the design space: 
\emph{global} instead of local perspective on interacting parties, \emph{direct style (DS)} instead of FSM or MNF 
encoding of effects, and \emph{control flow} instead of extra protocol for governing interactions.   

\lang takes a \textbf{global perspective} on interacting parties.
The parties execute the same program,
where pairs of send and receive actions that ``belong together''
are encapsulated into a single \textbf{direct-style} operation, which
is executed differently by sending and receiving parties.
Hence, dApps are modeled
as sequences and loops of send-receive-instructions 
shared by interacting parties.
Due to the global direct style perspective, it is syntactically impossible to define parties 
with mismatching send and receive pairs.
Hence, a standard System-F-like type system suffices.
The interaction protocol follows directly from 
the sequential \textbf{control flow} of interaction points in the program -- the 
compiler can automatically generate access and control guards with correctness guarantees. 
Semantically, \lang features a by-default-sequential programming model, 
intentionally making the sequential execution of methods explicit, 
including interaction effects. 

The global direct-style model also leads to improved
design of dApps:
No programmatic state management on the contract
and no so-called \textit{callback hell} \cite{Edwards09} on the client. 
The direct style is also superior to Nomos' MNF style. 
The tierless model avoids boilerplate:
Client code can directly access public contract variables, 
unlike JavaScript code, which has to access them via a function call that  
requires either an await expression or a callback.%
\footnote{Obsidian and Nomos do not provide any client design, so we can only compare to Solidity/JavaScript.}
Additionally, the developer has to implement getters for 
public variables with complex data types such as arrays.\footnote{For simple data types the getter is generated automatically.}
We provide some code measurements (lines of code and number of cross-tier control-flow calls) 
of our \lang{} and Solidity/JS dApp case studies
in Appendix~\ref{appendix:loc-evaluation}.

Finally, using one language for both the contract and the clients
naturally enables static type safety 
of values that cross the contract--client boundary:
an honest, non-compromized client cannot provide inconsistent input,
 e.g., 
with wrong number of parameters or falsely encoded types.\footnote{Recall that in dApps checking cross-tier type-safety is not a security feature but a design-time safety feature (due to the open-world assumption of the execution model of public ledgers).}
In a setting with different language stacks,
it is not possible to statically detect type mismatches in the client--contract interaction;
e.g., Solidity has a type \textit{bytes} for byte arrays, which 
does not exist in JavaScript (commonly used to implement clients of a Solidity contract).
Client developers need to encode byte arrays using 
hexadecimal string representations starting with ``0x'', otherwise
they cannot be interpreted by the contract.

%% file: paper/listings/obsidian_ttt.tex
\begin{lstlisting}[language=Solidity]
asset contract TTT {
  state Funding{}; state Executing{}; state Finished{}; state Closed{}
  transaction Fund(TTT@Funding>>(Funding|Executing) this, int c) {
    /*...*/;  if (/* enough funds? */) -> Executing else -> Funding }
  transaction Move(TTT@Executing>>(Executing|Finished) this, int x, int y) {
    /*...*/; if (/* game over? */) -> Finished else -> Executing }
  transaction Payout(TTT@Finished>>Closed this) {(*@\label{code:obsidian-payout}@*)
    /*...*/; -> Closed } }
\end{lstlisting}

%% file: paper/listings/nomos_ttt.tex
\begin{lstlisting}[language=nomos]
type Funding   =        int -> +{ notenough: Funding, enough: Executing }(*@\label{code:nomos-type-funding}@*)
type Executing = int -> int -> +{ notdone: Executing, done: Finished }
type Finished  =        int -> 1(*@\label{code:nomos-type-finished}@*)
proc contract funding : . |{*}- ($$s : Funding) = {(*@\label{code:nomos-process-funding}@*)
  a = recv $$s; /* ... */(*@\label{code:nomos-receive}@*)
  if /* enough funds? */ then  $$s.notenough; $$s <- funding(*@\label{code:nomos-decision-notenough}@*)
                         else  $$s.enough; $$s <- executing }(*@\label{code:nomos-decision-enough}@*)
proc contract executing : . |{*}- ($$s : Executing) = {(*@\label{code:nomos-process-executing}@*)
  x = recv $$s; y = recv $$s; /* ... */
  if /* game over? */ then $$s.notdone; $$s <- executing
                      else $$s.done; z = recv $$s; close $$s }
\end{lstlisting}

%% file: paper/relatedwork.tex
\subsection{Other Related Work}
\subparagraph{Smart contract languages.}
Harz and Knottenbelt~\cite{Harz18} survey smart contract  languages, Hu et al.~\cite{survey-contracts} survey smart contract tools and systems, W\"ohrer and Zdun~\cite{WohrerZ20} give an overview of design patterns in smart contracts.
Br\"unjes and Gabbay~\cite{BrunjesG20} distinguish between imperative and functional smart contract programming.
\emph{Imperative contracts} are based on the account model; the most prominent language
is Solidity~\cite{solidity}.
\emph{Functional} ones~\cite{Chakravarty19,SeijasT18, Seijas20}  are based on EUTxO (Extended Unspent Transaction Output) model~\cite{eutxo}.  
State channels~\cite{ChakravartyCFGK20, MBBKM19, DziembowskiFH18, DziembowskiEFHH19} 
optimistically optimize contracts for the functional model.
\lang{} does not yet support compilation to state channels but we plan to treat them
as another kind of tier.

\subparagraph{Smart contracts as state machines.}
Scilla~\cite{SergeyNJ0TH19} is an automata-based compiler
for contracts.
FSolidM~\cite{MavridouL18} enables creating contracts via a graphical interface.
VeriSolid~\cite{MavridouLSD19} generates contracts from graphical models 
enriched with predicates based on computational tree logic. 
EFSM tools~\cite{SuvorovU19} generate contracts from state machines and linear temporal logic.
\lang{} avoids a separate specification but infers 
transactions and their  order from the control flow
of a multitier dApp.

\subparagraph{Analysis tools.} Durieux et al.~\cite{DurieuxFAC20} and Ferreira et al.~\cite{smartbugs} empirically validate languages and tools
and relate design patterns to security vulnerabilities, extending the 
survey
by Di Angelo and Salzer~\cite{AngeloS19}.
Our work is complementary, targeting the correctness of the distributed program flow.
For
vulnerabilities not related to program flow
	(e.g., front-running, or bad randomness), 
	developers (using Solidity/JavaScript or \lang) 
	can use the surveyed analysis tools.

\subparagraph{Multitier languages.}
Multitier programming 
was pioneered by Hop~\cite{Serrano:2006,Serrano:2016}.
Modeling a persistent 
session in client--server applications with continuations was 
mentioned by Queinnec~\cite{Queinnec00} and elaborated in 
Links~\cite{Cooper:2006:LWP:1777707.1777724,Fowler:2019}.
Eliom~\cite{Radanne:2016}
supports bidirectional client--server communication for web applications.
ScalaLoci~\cite{Weisenburger:2018} generalizes the multitier model to generic distributed architectures. 
Our work specializes it 
to the dApp domain and its specific properties.
Giallorenzo et al.~\cite{GiallorenzoMPRS21} establish interesting 
connections between multitier (subjective) and choreographic (objective) languages -- two variants of the global model.
\lang adopts the subjective view, which naturally fits the dApp domain, 
where a dominant role (contract) controls the execution
and diverts control to other parties (clients)
to collect their input.

Mashic~\cite{LuoR12} is a compiler
for a \emph{mashup} between two JavaScript programs:
the untrusted embedded (iframe) \emph{gadget(s)} and the trustworthy hosting \emph{integrator} program,
which communicate via messages.
The authors prove that the compiler guarantees
integrity and confidentiality.
More specifically, the gadget(s) cannot learn more
than what the integrator sends and,
analogously, the gadget's influence is limited to the integrators interface.
In Mashic, the two programs are separate
and the compiler checks that they communicate only via specified messages.
In contrast, in \lang{}, client and contract code are mixed.
Thus, in addition to checking that only the specified messages are used,
we can also check the interaction protocol -- expressed by the structure of the control flow of the program --
and ensure that it is followed by the target program after compilation.

Swift's~\cite{Chong07} secure automatic partitioning approach
uses information flow policies to derive placements.
Based on the policies, a constraint solver with integer programming
heuristically picks a placement such that network traffic is minimal and information flow integrity is preserved.
In contrast, placements in \lang{} are explicit to the developer.
Further, in blockchain programming, every single instruction generated
by the compiler potentially incurs high costs.
Therefore, we demonstrated that our compiler generates inexpensive programs,
whereas Swift does not consider the program's execution cost.

\subparagraph{Effectful programs and meta-programming.}
MNF and CSP are widely discussed as intermediate compiler forms~\cite{Appel1992, FlanaganSDF93, Kennedy07, MaurerDAJ17, CongOER19}.
F\# computation expressions~\cite{PetricekS14}
support control-flow operators in monadic expressions.
OCaml supports a monadic and applicative let~\cite{ocaml-applicative-let}:
more flexible than do-notation but still restricted to MNF.
Idris' !-notation~\cite{idris-bang-notation} 
inspired the GHC proposal for monadic inline binding~\cite{ghc-inline-binding}.
Scala supports effectful programs through
coroutines~\cite{scala-coroutines},
async/await~\cite{scala-async},
monadic inline binding~\cite{monadless}, 
Dsl.scala~\cite{dslscala}
and a (deprecated) compiler plugin for CPS translation~\cite{scala-continuations}.
The dotty-cps-async macro~\cite{dotty-cps-async} supports 
async/await and similar effects for the Dotty compiler.

%% file: paper/conclusion.tex
\section{Conclusion}
\label{sec:conclusion}

We proposed \lang{}, the first global language for dApps that features direct style communication. 
Compared to the state of the art, \lang{} (a)~enables the implementation of contract and client logic within the same development unit, rendering intricacies of the heterogeneous technology stack obsolete and avoiding boilerplate code, (b)~provides support for explicitly encoding the intended program flow and (c)~reduces the risk of human failures by enforcing the intended program flow and forcing developers to specify access control.

Unlike previous work that targeted challenges in the development of dApps with advanced typing disciplines 
e.g., session types, our model does not exhibit visible side effects
and  gets away with a standard System-F-style type system.
We describe the design and the main features of \lang{} informally, define its formal semantics, formalize the compilation process and prove it correct.
We demonstrate \lang{}'s applicability via case studies and performance benchmarks.

We plan to generate state channels
-- to optimistically cost-optimize dApps --
similar to how we generate state machines from high-level logic.
Further, we believe that
our technique for deriving the communication protocol from direct-style control flow
generalizes beyond the domain of smart contracts and we will
explore its further applicability.

%% file: paper/acknowledgements.tex
This work has been funded %
by \grantsponsor{BMBF}{the German Federal Ministry of Education and Research}{} iBlockchain project (BMBF No. \grantnum{BMBF}{16KIS0902}),
by \grantsponsor{DFG}{the German Research Foundation}{} (DFG, \grantnum{DFG}{SFB 1119} -- CROSSING Project), %
by the BMBF and \grantsponsor{HMWK}{the Hessian Ministry of Higher Education, Research, Science and the Arts}{} within their joint support of the \grantnum{BMBF & HMWK}{National Research Center for Applied Cybersecurity ATHENE},
by \grantsponsor{LOEWE}{the Hessian LOEWE initiative}{} (\grantnum{LOEWE}{emergenCITY}),
by \grantsponsor{SNSF}{the Swiss National Science Foundation}{} (SNSF, No. \grantnum{SNSF}{200429}),
and by \grantsponsor{unisig}{the University of St. Gallen}{} (IPF, No. \grantnum{unisig}{1031569}).

Thanks to George Zakhour for feedback on the initial draft.

%% file: paper/case_studies.tex
\section{Case studies}
\label{appendix:case-studies}

This section describes the implemented case studies in detail.
Bartoletti and Pompianu~\cite{bartoletti2017empirical} identify five classes of smart contract
applications: Financial, Notary, Game,  Wallet, and Library.
Our case studies include at least one application per category
(Table~\ref{tab:categories-and-x-tier-passes}). 
In addition, we consider scalability solutions.

\emph{Financial.}
These apps include digital tokens, crowdfunding, escrowing,
advertisement, insurances and sometimes Ponzi schemes. %
 A study investigating all blocks mined until September 15th, 2018~\cite{oliva2020exploratory}, 
 found that 72.9\,\% of the high-activity contracts are token contracts compliant to ERC-20 or ERC-721, 
 which have an accumulated market capitalization of US\,\$\,12.7 billion.
We have implemented a fungible \lang{} token of the ERC-20 standard.
Further, we implemented crowdfunding and escrowing case studies.
These case studies demonstrate how to send and receive coins with \lang{}, which is 
the basic functionality of financial applications.
Other financial use cases can be implemented in \lang{} with similar techniques.

\emph{Notary.}
These contracts use the blockchain to store data immutably and persistently, e.g., to certify their ownership.
We implemented a general-purpose notary contract enabling 
users to store arbitrary data, e.g., document hashes or images, together with a submission timestamp and the data owner.
This case study demonstrates that Notaries are  expressible with \lang{}.

\emph{Games.}
We implemented TicTacToe (Section~\ref{sec:the-language}),
Rock-Paper-Scissors, Hangman and Chinese Checkers.
Hangman evolves through multiple phases and hence benefits
from the explicit control flow definition in \lang{} more than the other game case studies.
The game Chinese Checkers is more complex than the others, in regard to the number of parties, the game logic and the number of rounds, and hence, represents larger applications. 
Rock-Paper-Scissors illustrates how randomness for dApps is securely generated.
Every Ethereum transaction, including the executions of contracts, 
is deterministic -- all participants can validate the generation of new blocks.
Hence, secure randomness is negotiated among parties:
in this case, by making use of timed commitments~\cite{DBLP:conf/sp/AndrychowiczDMM14},
i.e., all parties commit to a random seed share and open it 
after all commitments have been posted.
The contract uses the sum of all seed shares as randomness.
If one party aborts prior to opening its commitment, it is penalized.
In Rock-Paper-Scissors both parties commit to their choice -- their random share -- and open it afterwards.
Other games of chance, e.g., gambling contracts, use the same technique.

\emph{Wallet.}
A wallet contract manages digital assets, i.e., cryptocurrencies and tokens, and offers additional features 
such as shared ownership or daily transaction limits.
At August 30, 2019, 3.9 M of 17.9 M (21\,\%) deployed smart contracts have been different types of wallet contracts \cite{DBLP:journals/corr/abs-2001-06909}.
Multi-signature wallets are a special type of wallet that provides a transaction voting mechanism by only executing transactions, which are signed by a fixed fraction of the set of owners.
Wallets transfer money and call other contracts in their users stead depending on run-time input, %
demonstrating calls among contracts in \lang{}.
Further, a multi-signature wallet uses built-in features of the Ethereum VM for signature validation, i.e., data encoding, hash calculation, and signature verification, showing that these
features are supported in \lang{}.

\emph{Libraries.}
As the cost of deploying a contract increases with the amount of code in Ethereum,
developers try to avoid code repetitions.
Contract inheritance does not help: child contracts simply copy the attributes 
and functions from the parent.
Yet, one can outsource commonly used logic to \emph{library contracts} 
that are deployed once and called by other contracts.
For example, the TicTacToe dApp and the TicTacToe channel in our case studies 
share some logic, e.g., to check the win condition.
To demonstrate libraries in \lang{}, we include a TicTacToe library 
to our case studies and another on-chain executed TicTacToe dApp 
which uses such library instead of deploying the logic itself.
Libraries use a call instruction similar to wallets, although 
the call target is typically known at deployment and can be hard-coded.

\emph{Scalability solutions.}
State channels~\cite{MBBKM19, DziembowskiFH18, DziembowskiEFHH19} 
are scalability solutions, which enable a fixed group of parties to move their dApp 
to a non-blockchain consensus protocol: the execution falls-back 
to the blockchain in case of disputes.
Similar to multi-signature wallets, state channels use built-in signature validation.
We implemented a state channel for TicTacToe\footnote{A general solution is a much larger engineering effort and subject of industrial projects~\cite{perun-network, state-channels}}
to demonstrate that \lang{} supports state channels.

%% file: paper/loc_evaluation.tex
\section{Empirical Evaluation of Design Quality}
\label{appendix:loc-evaluation}

\begin{figure}
	\figurehead
	\centering
	\caption{Categories and Cross-tier calls.}
	\label{tab:categories-and-x-tier-passes}
	\input{paper/graphs/crossTierPassesTable}
\end{figure}

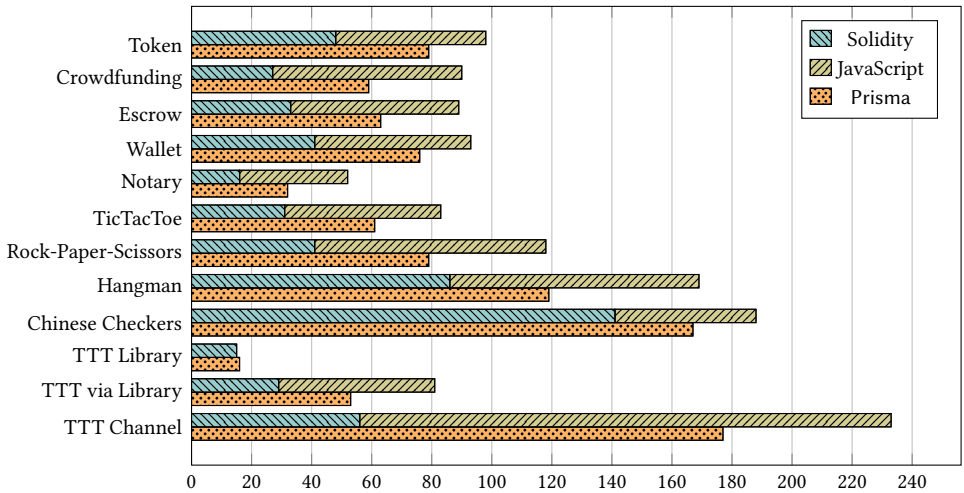
\begin{figure}
	\figurehead
	\centering
	\input{paper/graphs/locComparison}
	\caption{LOC in Solidity/JavaScript and \lang{}.}
	\label{fig:loc-comparison}
\end{figure}

In Section~\ref{sec:background}, we argued 
that with \lang{}, 
(a) we provide communication safety with a standard system-F-like type-system,
(b) the program flow can be defined explicitly and is enforced automatically,
(c) dApp developers need to master a single technology that covers both tiers, 
(d) cross-tier type-safety can be checked at compile-time, 
and (e) the code is simpler and less verbose due to reduced boilerplate code for communication and less control flow jumps.
The claims (a), (c), and (d) are a direct consequence of \lang{}'s design and do not require further evidence.
Claim (c) has been formally proven in Section~\ref{sec:formalization}.
It remains to investigate claim (e), i.e., in which extent \lang{} reduces the amount of code and error-prone control-flow jumps.

To this end, 
we implemented all case studies with equivalent functionality in \lang{} and in Solidity/JavaScript.
The JavaScript client logic is in direct style using async/await -- 
the Solidity contract needs to be implemented as a finite-state-machine.
We keep the client logic of our case studies (in both, the \lang{} and the Solidity implementation) as basic as possible, not to compare the client logic in Scala and in JavaScript but rather focus on the dApp semantics.
A complex client logic would shadow the interaction with the contract logic 
-- limited in size due to the gas semantics.

We start with comparing LOCs in the case studies (Figure~\ref{fig:loc-comparison}).
The results in Figure~\ref{fig:loc-comparison} show that case studies 
written in \lang{} require only 55\,--\,89\,\% LOC compared to 
those implemented in Solidity/JavaScript.
One exception is the standalone library, which has no client code and hence does not directly 
profit from the tierless design.

Second, we consider occurrences of explicit cross-tier control-flow calls in the Solidity/JavaScript dApps (cf. Table~\ref{tab:categories-and-x-tier-passes}),
which complicate control flow,
compared to \lang{}, where cross-tier access is seamless.
In the client implementations, 6\,--\,18\,\% of all lines trigger a contract interaction passing the control flow to the contract and waiting for the control flow to return.
From the contract code in finite-state-machine style, it is not directly apparent at which position the program flow continues, once passed back from clients to contract, i.e., which function is called by the clients next.
Direct-style code, on the other hand, ensures that the control flow of the contract always continues in the line that passed the control flow to the client by invoking an \lstinline!awaitCl! expression.

%% file: paper/graphs/crossTierPassesTable.tex
\begin{tabular}{>{\footnotesize}l>{\footnotesize}l>{\footnotesize}r>{\footnotesize}r>{\footnotesize}r>{\footnotesize}r}
	\toprule
	Category       & Case study               & \hspace{-4em}Cross-tier calls & \lang{} LoC & Solidity + JavaScript LoC \\
	\midrule %
	
	Financial      & Token                    & \varTokenRC & \varTokenLoCPrisma & \varTokenLoCSol{} + \varTokenLoCJS \\
	               & Crowdfunding             & \varCrowdfundingRC & \varCrowdfundingLoCPrisma & \varCrowdfundingLoCSol{} + \varCrowdfundingLoCJS \\
	               & Escrow                   & \varEscrowRC & \varEscrowLoCPrisma & \varEscrowLoCSol{} + \varEscrowLoCJS \\[6pt] %
	               
	Wallet         & Multi-signature wallet   & \varMultiSigRC & \varMultiSigLoCPrisma & \varMultiSigLoCSol{} + \varMultiSigLoCJS \\[6pt] %
	
	Notary         & General-purpose notary   & \varNotaryRC & \varNotaryLoCPrisma & \varNotaryLoCSol{} + \varNotaryLoCJS \\[6pt] %
	
	Game           & Rock Paper Scissors      & \varRPSRC & \varRPSLoCPrisma & \varRPSLoCSol{} + \varRPSLoCJS \\
				   & TicTacToe                & \varTTTRC & \varTTTLoCPrisma & \varTTTLoCSol{} + \varTTTLoCJS \\
				   & Hangman                  & \varHangmanRC & \varHangmanLoCPrisma & \varHangmanLoCSol{} + \varHangmanLoCJS \\
	               & Chinese Checkers		  & \varChineseCheckersRC & \varChineseCheckersLoCPrisma & \varChineseCheckersLoCSol{} + \varChineseCheckersLoCJS \\[6pt]

	Library        & TicTacToe library        & -- & \varChineseCheckersLoCPrisma & \varChineseCheckersLoCSol{} + -- \\
	               & TicTacToe using library  & \varTTTViaLibRC & \varTTTViaLibLoCPrisma & \varTTTViaLibLoCSol{} + \varTTTViaLibLoCJS \\[6pt] %
	               
		Scalability & TicTacToe channel        & \varTTTChannelRC & \varTTTChannelLoCPrisma & \varTTTChannelLoCSol{} + \varTTTChannelLoCJS \\ 
	\bottomrule
\end{tabular}

%% file: paper/graphs/locComparison.tex
\makeatletter
\newcommand\resetstackedplots{
  \makeatletter
  \pgfplots@stacked@isfirstplottrue
  \makeatother
  \addplot [forget plot,draw=none] coordinates{(0,TTT Channel) (0,Token) (0,Crowdfunding) (0,Escrow) (0,Notary) (0,Hangman) (0,TicTacToe) (0,Rock-Paper-Scissors) (0,Chinese Checkers) (0,Wallet) (0,TTT Library) (0,TTT via Library)};
}
\makeatother

\begin{tikzpicture}
\tikzset{font=\footnotesize}
\begin{axis}[
xbar stacked,
width = 0.85*\linewidth,
bar width=5pt, y=13pt,
xmin=0,
legend pos=north east,
legend entries={Solidity, JavaScript, \lang{}},
ytick=data,
symbolic y coords={Token, Crowdfunding, Escrow, Wallet, Notary, TicTacToe, Rock-Paper-Scissors, Hangman, Chinese Checkers, TTT Library, TTT via Library, TTT Channel},
y dir=reverse,
major y tick style = transparent,
xmajorgrids = true
]
\addplot +[bar shift=+2.5pt, black, fill=white!60!teal, mark=none, postaction={pattern = north west lines}] coordinates {(\varTTTChannelLoCSol,TTT Channel) (\varTokenLoCSol,Token) (\varCrowdfundingLoCSol,Crowdfunding) (\varEscrowLoCSol,Escrow) (\varNotaryLoCSol,Notary) (\varHangmanLoCSol,Hangman) (\varTTTLoCSol,TicTacToe) (\varRPSLoCSol,Rock-Paper-Scissors)
	(\varChineseCheckersLoCSol,Chinese Checkers) (\varMultiSigLoCSol,Wallet) (\varTTTLibraryLoCSol,TTT Library) (\varTTTViaLibLoCSol,TTT via Library)};

\addplot +[bar shift=+2.5pt, black, fill=white!60!olive, mark=none, postaction={pattern = north east lines}] coordinates {(\varTTTChannelLoCJS,TTT Channel) (\varTokenLoCJS,Token) (\varCrowdfundingLoCJS,Crowdfunding) (\varEscrowLoCJS,Escrow) (\varNotaryLoCJS,Notary) (\varHangmanLoCJS,Hangman) (\varTTTLoCJS,TicTacToe) (\varRPSLoCJS,Rock-Paper-Scissors)
	(\varChineseCheckersLoCJS,Chinese Checkers) (\varMultiSigLoCJS,Wallet) (\varTTTLibraryLoCJS,TTT Library) (\varTTTViaLibLoCJS,TTT via Library)};

\resetstackedplots

\addplot +[bar shift=-2.5pt, black, fill=white!40!orange, mark=none, postaction={pattern = crosshatch dots}] coordinates {(\varTTTChannelLoCPrisma,TTT Channel) (\varTokenLoCPrisma,Token) (\varCrowdfundingLoCPrisma,Crowdfunding) (\varEscrowLoCPrisma,Escrow) (\varNotaryLoCPrisma,Notary) (\varHangmanLoCPrisma,Hangman) (\varTTTLoCPrisma,TicTacToe) (\varRPSLoCPrisma,Rock-Paper-Scissors)
	(\varChineseCheckersLoCPrisma,Chinese Checkers) (\varMultiSigLoCPrisma,Wallet) (\varTTTLibraryLoCPrisma,TTT Library) (\varTTTViaLibLoCPrisma,TTT via Library)};

\end{axis}
\end{tikzpicture}%

%% file: paper/appendix.tex
\section{Proofs}
\label{appendix:proofs}

\setcounter{globallemma}{0}
\setcounter{globaltheorem}{0}

We provide the definition of $\mathsf{comp}'$ and $\mathsf{comp}$ in Figure~\ref{fig:comp},
the definition for the free variables for a given term $\fv$ in Figure~\ref{fig:free-variables}
and the detailed proofs for the theorem and the lemmas on the following pages.

\vspace{3em}

\begin{figure}
$\begin{array}{llll}
comp'(\decl; \syndecl; \trampoline(\mainexp)) &~=~&
  \decl\seq \coclfn(\syndecl, id, \assert(\false), \assert(\false))\seq \trampoline(comp(\mainexp))\\
&& \begin{array}{ll} \text{where} & id \freshvar \end{array}\\[1em]

comp\left(\begin{array}{l}
  \decl; \coclfn(\syndecl, id,\\~~\subexp_{1,alt},\\~~\subexp_{2,alt});\\
  tmp \gets_s (() \to \subexp_1)\seq \subexp_2
\end{array}\right) &~=~& \left(\begin{array}{l}
  \decl; \coclfn(\syndecl, id,\\
  ~~\ifletexp (c :: \fv(() \to \subexp_1)) = id \thenexp \subexp_1 \elseexp \subexp_{1,alt}',\\
  ~~\ifletexp (c :: x :: \fv(x \to \subexp_2')) = id \thenexp \\
  ~~~~~~\assert(\this.\statevar == c \,\operatorname{\&\&}\, \this.\whovar(\this.\sendervar))\seq \\
  ~~~~~~\this.\statevar := 0\seq \subexp_2' \\
  ~~\elseexp\\
  ~~~~~~ \subexp_{2,alt}');\\
  \this.\whovar := \subexp_0\seq \this.\statevar := c\seq \\
  (\mathsf{More},~ c :: \fv(() \to \subexp_1),~ c :: \fv(x \to \subexp_2'))
\end{array}\right)
\\[2.5em]
&& \begin{array}{ll}
     \text{where} & c \freshvar\\
     \text{and}   & \decl; \coclfn(\syndecl, id, \subexp_{1,alt}', \subexp_{2,alt}'); \subexp_2' =\\
                  & ~~ \defun[\decl; \coclfn(\syndecl, id, \subexp_{1,alt}, \subexp_{2,alt}); \subexp_2]\\
   \end{array}
\\

comp\left(\begin{array}{l}
  \decl; \coclfn(\syndecl, id, \subexp_{1,alt}, \subexp_{2,alt});\\
  \letexp x = \subexp_0; ~ \subexp_1
\end{array}\right) &~=~& \left(\begin{array}{l}
  \decl; \coclfn(\syndecl, id, \subexp_{1,alt}, \subexp_{2,alt});\\
  \letexp x = \subexp_0; ~ \defun(\subexp_1)
\end{array}\right)
\\

comp\left(\begin{array}{l}
  \decl; \coclfn(\syndecl, id, \subexp_{1,alt}, \subexp_{2,alt});\\
  \subexp
\end{array}\right) &~=~& \left(\begin{array}{l}
  \decl; \coclfn(\syndecl, id, \subexp_{1,alt}, \subexp_{2,alt});\\
  \subexp
\end{array}\right)
\\[2em]

\coclfn(\syndecl, id, \subexp_{1,alt}, \subexp_{2,alt}) &~=~&
  (\clVal \clfnvar = id \to \subexp_{1,alt}); (\coVal \cofnvar = id \to \subexp_{2,alt}); b
\\
\end{array}$
\caption{$\mathsf{comp}'$ and $\mathsf{comp}$.}
\label{fig:comp}
\end{figure}

\vspace{3em}

\begin{figure}
  \input{notes/def_fv.tex}
  \caption{Free variables.}
  \label{fig:free-variables}
\end{figure}

\vspace{5em}

\clearpage{}

\begin{globaltheorem}[Secure Compilation]\emph{
For all programs $P$ over closed terms, the trace set of evaluating the program under attack
equals the trace set of evaluating the compiled program under attack,
i.e.,
\[ \forall P. ~~~ \{~ init_A(comp'(mnf'((P)))) ~\} \approx ... \approx \{~ init_A(P) ~\}  \]
}\end{globaltheorem}

\begin{proof}
\[ \begin{array}{cl}
                                                   &~ init_A(comp'(mnf'(P)))\\
~\overset{\text{Lemma }\ref{lemma:mnf'}}{\approx}  &~ init_A(mnf'(P))\\
~\overset{\text{Lemma }\ref{lemma:comp'}}{\approx} &~ init_A(P)\\
\end{array} \]
\end{proof}

\subparagraph{Extensions}

For simplicity, our definition of initialization uses a fixed set of clients.
Yet, the malicious semantics does not actually depend on the fixed set of clients,
but instead models an attacker that is in control of all clients
with the capability of sending messages from any client, not bound to the fixed set.
Hence, it is straightforward to extend the proofs to the setting of a dynamic set of clients, e.g.,
clients joining and leaving at run time.

Further, our trace equality relation defines that all programs
in the relation eventually reduce to values, filtering out programs that loop or get stuck.
Below, we outline an approach to prove trace equality for looping or stuck programs
by showing that such programs loop with the same infinite trace
or get stuck at the same trace, respectively.
To this end, we track the number of steps done via a step-indexed trace equality relation:
\[ p;q;e \Downarrow^n ~~=~ \{~ (p',v) ~|~ (p;q;e) \to^n (p';q';v) ~\}
~~~~~~~~~~~~~~~~~~~~~~~~~
   T \Downarrow^n ~=~ \bigcup_{p;q;e \in T} ~ p;q;e \Downarrow^n      \]

With this definition, we can no longer use just equality of traces as the left and right program
may take a different number of steps to produce the same events.
Instead, we move from an equality relation to a relation stating non-disagreement, which says that
-- independently of how long we run either statement --
the traces will never be in disagreement:
\[ (T \approx^n S)  ~~~\Leftrightarrow~~~  (T\Downarrow^n ~\#_{\text{set}}~ S\Downarrow^n) \]

where $\#_\text{set}$ is defined on trace sets as
\[ T \#_{\text{set}} S  ~~~\Leftrightarrow~~~
(\forall t {\in} T.~ \exists s {\in} S.~~ t ~\#_{\text{trace}}~ s) \land
(\forall s {\in} S.~ \exists t {\in} T.~~ t~ \#_{\text{trace}}~ s) \]

and $\#_\text{trace}$ on event traces as
\[ \begin{array}{llllll}
(ev,   &())     &\#_{\text{trace}}& (ev,   &tail_2) &= true \\
(ev,   &tail_1) &\#_{\text{trace}}& (ev,   &())     &= true \\
(ev_1, &tail_1) &\#_{\text{trace}}& (ev_2, &tail_2) &= false \\
(ev,   &tail_1) &\#_{\text{trace}}& (ev,   &tail_2) &= tail_1 ~\#_{\text{trace}}~ tail_2 \\
\end{array}. \]

\clearpage{}\input{notes/lemma_assoc_preserves.txt.tex}
\clearpage{}\input{notes/lemma_mnfe_preserves.txt.tex}
\clearpage{}\input{notes/lemma_mnfp_preserves.txt.tex}
\clearpage{}\input{notes/lemma_dtle_preserves.txt.tex}
\clearpage{}\input{notes/lemma_dtlp_preserves.txt.tex}

\clearpage{}

%% file: notes/def_fv.tex
$\begin{array}{lll}
\fv(\mainexp_0 :: \mainexp_1)                   &=& \fv(\mainexp_0) \cup \fv(\mainexp_1) \\
\fv(x \to \mainexp)                             &=& \fv(\mainexp) \setminus \fv(x) \\
\fv(id)                                         &=& \{ id \} \\
\fv(\mainexp_0~\mainexp_1)                      &=& \fv(\mainexp_0) \cup \fv(\mainexp_1) \\
\fv(\downarrow^*((\mainexp_0, ()\to \mainexp_1) &=& \fv(\mainexp_0) \cup \fv(\mainexp_1) \\
\fv(\text{let }x = \mainexp_0; \mainexp_1)      &=& \fv(\mainexp_0) \cup \fv(\mainexp_1) \setminus \fv(x) \\
\fv(\text{this.}\iden := \mainexp_0)            &=& \fv(\mainexp_0) \\
\fv(\text{this.}\syniden := \mainexp_0)         &=& \fv(\mainexp_0) \\
\fv(\text{this.}\iden)                          &=& \{ \} \\
\fv(\text{this.}\syniden)                       &=& \{ \} \\
\fv(c)                                          &=& \{ \} \\
\end{array}$

%% file: notes/lemma_assoc_preserves.txt.tex
\bigskip{}
\begin{globallemma}[assoc preserves traces]
 \emph{$assoc$ is defined as a recursive term-to-term transformation on open terms,}
\emph{whereas traceset equality is defined by reducing terms to values, i.e., on closed terms.}
\emph{Since all valid programs are closed terms,}
\emph{we show that $assoc$ preserves the traceset of an open term $e$ that is closed by substitution $[x\Mapsto v]$.}

\emph{For all terms ${\subexp}$, traces ${\trace}$, traces ${\syntrace}$, values $v$, patterns $x$,}
\[ \{~{\trace};{\syntrace};~[x{\Mapsto}v]~assoc({\subexp})~\} ~\approx~ ... ~\approx~ \{~{\trace};{\syntrace};~[x{\Mapsto}v]~{\subexp}~\} \]
\end{globallemma}\begin{proof} By induction over term structure.
\subparagraph{Case} $ ~~~{\subexp}~=~({\ensuremath{\mathsf{let~}}}x_1=({\ensuremath{\mathsf{let~}}}x_0={\subexp}_0;~{\subexp}_1);~{\subexp}_2) $.

We know $x_0~{\notin}~\fv({\subexp}_2)$ since ${\subexp}_2$ is not in the scope of the $x_0$ binding,
and that all identifiers are distinct, which can always be achieved by $\alpha$-renaming.

{\small
\[ ~x_0~{\notin}~\fv({\subexp}_2) \]
}

According to $\approx$, we only consider terms that reduce to a value.
Therefore, let $\phi$ be the judgement that
the term ${\subexp}_0$ closed by $[x{\Mapsto}v]$ with trace ${\trace}$
evaluates to a value $v_0$ producing trace $p_0$.

{\small
\[ ~\phi~\equiv~~~~({\trace};{\syntrace};~[x{\Mapsto}v]{\subexp}_0~~~{\to}^*~~{\trace}~p_0;{\syntrace};~v_0) \]
}

The lemma holds by the following chain of transitive relations.
We evaluate the compiled program from top to bottom ($~{\to}^*$)
and the original program from bottom to top ($~{\gets}^*$) until configurations converge.
The induction hypothesis (IH) allows the removal of $assoc$ in redex position under traceset equality ($\approx$).

{\small\begin{longtable}{CLL}
   & \vspace{5pt}{\left\{\def\arraystretch{1}\begin{array}{l}\arrayrulecolor{white!80!black}{\trace};{\syntrace};~[x{\Mapsto}v]~assoc({\subexp})
\end{array}\right.} &\hspace{-1.5em}\left.\vphantom{\left\{\def\arraystretch{1}\begin{array}{l}  p;q; [x|=>v] assoc(e) \end{array}\right\}}\right\}\\
\overset{\text{def.~}{\subexp}}{=}
   & \vspace{5pt}{\left\{\def\arraystretch{1}\begin{array}{l}\arrayrulecolor{white!80!black}{\trace};{\syntrace};~[x{\Mapsto}v]~assoc({\ensuremath{\mathsf{let~}}}x_1=({\ensuremath{\mathsf{let~}}}x_0={\subexp}_0;~{\subexp}_1);~{\subexp}_2)
\end{array}\right.} &\hspace{-1.5em}\left.\vphantom{\left\{\def\arraystretch{1}\begin{array}{l}  p;q; [x|=>v] assoc(let x1=(let x0=e0; e1); e2) \end{array}\right\}}\right\}\\
\overset{\text{def.~}assoc}{=}
   & \vspace{5pt}{\left\{\def\arraystretch{1}\begin{array}{l}\arrayrulecolor{white!80!black}{\trace};{\syntrace};~[x{\Mapsto}v]~assoc({\ensuremath{\mathsf{let~}}}x_0={\subexp}_0;~assoc({\ensuremath{\mathsf{let~}}}x_1={\subexp}_1;~{\subexp}_2))
\end{array}\right.} &\hspace{-1.5em}\left.\vphantom{\left\{\def\arraystretch{1}\begin{array}{l}  p;q; [x|=>v] assoc(let x0=e0; assoc(let x1=e1; e2)) \end{array}\right\}}\right\}\\
\overset{IH}{\approx}
   & \vspace{5pt}{\left\{\def\arraystretch{1}\begin{array}{l}\arrayrulecolor{white!80!black}{\trace};{\syntrace};~[x{\Mapsto}v]~{\ensuremath{\mathsf{let~}}}x_0={\subexp}_0;~assoc({\ensuremath{\mathsf{let~}}}x_1={\subexp}_1;~{\subexp}_2)
\end{array}\right.} &\hspace{-1.5em}\left.\vphantom{\left\{\def\arraystretch{1}\begin{array}{l}  p;q; [x|=>v] let x0=e0; assoc(let x1=e1; e2) \end{array}\right\}}\right\}\\
\overset{\text{def.~}{\Mapsto}}{=}
   & \vspace{5pt}{\left\{\def\arraystretch{1}\begin{array}{l}\arrayrulecolor{white!80!black}{\trace};{\syntrace};~{\ensuremath{\mathsf{let~}}}x_0=[x{\Mapsto}v]~{\subexp}_0;~[x{\Mapsto}v]~assoc({\ensuremath{\mathsf{let~}}}x_1={\subexp}_1;~{\subexp}_2)
\end{array}\right.} &\hspace{-1.5em}\left.\vphantom{\left\{\def\arraystretch{1}\begin{array}{l}  p;q; let x0=[x|=>v] e0; [x|=>v] assoc(let x1=e1; e2) \end{array}\right\}}\right\}\\
\overset{\phi}{~{\to}^*}
   & \vspace{5pt}{\left\{\def\arraystretch{1}\begin{array}{l}\arrayrulecolor{white!80!black}{\trace}~p_0;{\syntrace};~{\ensuremath{\mathsf{let~}}}x_0=v_0;~[x{\Mapsto}v]~assoc({\ensuremath{\mathsf{let~}}}x_1={\subexp}_1;~{\subexp}_2)~|~\forall~v_0~p_0,~\phi
\end{array}\right.} &\hspace{-1.5em}\left.\vphantom{\left\{\def\arraystretch{1}\begin{array}{l}  p p0;q; let x0=v0; [x|=>v] assoc(let x1=e1; e2) | \forall v0 p0, \phi \end{array}\right\}}\right\}\\
\overset{\textsc{Rlet}}{~{\to}}
   & \vspace{5pt}{\left\{\def\arraystretch{1}\begin{array}{l}\arrayrulecolor{white!80!black}{\trace}~p_0;{\syntrace};~[x_0{\Mapsto}v_0,~x{\Mapsto}v]~assoc({\ensuremath{\mathsf{let~}}}x_1={\subexp}_1;~{\subexp}_2)~|~\forall~v_0~p_0,~\phi
\end{array}\right.} &\hspace{-1.5em}\left.\vphantom{\left\{\def\arraystretch{1}\begin{array}{l}  p p0;q; [x0|=>v0, x|=>v] assoc(let x1=e1; e2) | \forall v0 p0, \phi \end{array}\right\}}\right\}\\
\overset{IH}{\approx}
   & \vspace{5pt}{\left\{\def\arraystretch{1}\begin{array}{l}\arrayrulecolor{white!80!black}{\trace}~p_0;{\syntrace};~[x_0{\Mapsto}v_0,~x{\Mapsto}v]~{\ensuremath{\mathsf{let~}}}x_1={\subexp}_1;~{\subexp}_2~|~\forall~v_0~p_0,~\phi
\end{array}\right.} &\hspace{-1.5em}\left.\vphantom{\left\{\def\arraystretch{1}\begin{array}{l}  p p0;q; [x0|=>v0, x|=>v] let x1=e1; e2 | \forall v0 p0, \phi \end{array}\right\}}\right\}\\
\overset{\text{def.~}{\Mapsto};~x_0~{\notin}~\fv({\subexp}_2)}{=}
   & \vspace{5pt}{\left\{\def\arraystretch{1}\begin{array}{l}\arrayrulecolor{white!80!black}{\trace}~p_0;{\syntrace};~{\ensuremath{\mathsf{let~}}}x_1=[x_0{\Mapsto}v_0,~x{\Mapsto}v]{\subexp}_1;~[x{\Mapsto}v]{\subexp}_2~|~\forall~v_0~p_0,~\phi
\end{array}\right.} &\hspace{-1.5em}\left.\vphantom{\left\{\def\arraystretch{1}\begin{array}{l}  p p0;q; let x1=[x0|=>v0, x|=>v]e1; [x|=>v]e2 | \forall v0 p0, \phi \end{array}\right\}}\right\}\\
\overset{\textsc{Rlet}}{~{\gets}}
   & \vspace{5pt}{\left\{\def\arraystretch{1}\begin{array}{l}\arrayrulecolor{white!80!black}{\trace}~p_0;{\syntrace};~{\ensuremath{\mathsf{let~}}}x_1=({\ensuremath{\mathsf{let~}}}x_0=v_0;~[x{\Mapsto}v]{\subexp}_1);~[x{\Mapsto}v]{\subexp}_2~|~\forall~v_0~p_0,~\phi
\end{array}\right.} &\hspace{-1.5em}\left.\vphantom{\left\{\def\arraystretch{1}\begin{array}{l}  p p0;q; let x1=(let x0=v0; [x|=>v]e1); [x|=>v]e2 | \forall v0 p0, \phi \end{array}\right\}}\right\}\\
\overset{\phi}{~{\gets}^*}
   & \vspace{5pt}{\left\{\def\arraystretch{1}\begin{array}{l}\arrayrulecolor{white!80!black}{\trace};{\syntrace};~{\ensuremath{\mathsf{let~}}}x_1=({\ensuremath{\mathsf{let~}}}x_0=[x{\Mapsto}v]{\subexp}_0;~[x{\Mapsto}v]{\subexp}_1);~[x{\Mapsto}v]{\subexp}_2
\end{array}\right.} &\hspace{-1.5em}\left.\vphantom{\left\{\def\arraystretch{1}\begin{array}{l}  p;q; let x1=(let x0=[x|=>v]e0; [x|=>v]e1); [x|=>v]e2 \end{array}\right\}}\right\}\\
\overset{\text{def.~}{\Mapsto}}{=}
   & \vspace{5pt}{\left\{\def\arraystretch{1}\begin{array}{l}\arrayrulecolor{white!80!black}{\trace};{\syntrace};~[x{\Mapsto}v]~{\ensuremath{\mathsf{let~}}}x_1=({\ensuremath{\mathsf{let~}}}x_0={\subexp}_0;~{\subexp}_1);~{\subexp}_2
\end{array}\right.} &\hspace{-1.5em}\left.\vphantom{\left\{\def\arraystretch{1}\begin{array}{l}  p;q; [x|=>v] let x1=(let x0=e0; e1); e2 \end{array}\right\}}\right\}\\
\overset{\text{def.~}{\subexp}}{=}
   & \vspace{5pt}{\left\{\def\arraystretch{1}\begin{array}{l}\arrayrulecolor{white!80!black}{\trace};{\syntrace};~[x{\Mapsto}v]~{\subexp}
\end{array}\right.} &\hspace{-1.5em}\left.\vphantom{\left\{\def\arraystretch{1}\begin{array}{l}  p;q; [x|=>v] e \end{array}\right\}}\right\}\\
\end{longtable}}

\subparagraph{Case} $ ~~~{\subexp}~\ne~({\ensuremath{\mathsf{let~}}}x_1=({\ensuremath{\mathsf{let~}}}x_0={\subexp}_0;~{\subexp}_1);~{\subexp}_2) $.

If ${\subexp}$ is not of nested let form, we simply apply the definition of $assoc$.

{\small\begin{longtable}{CLL}
   & \vspace{5pt}{\left\{\def\arraystretch{1}\begin{array}{l}\arrayrulecolor{white!80!black}{\trace};{\syntrace};~[x{\Mapsto}v]~assoc({\subexp})
\end{array}\right.} &\hspace{-1.5em}\left.\vphantom{\left\{\def\arraystretch{1}\begin{array}{l}  p;q; [x|=>v] assoc(e) \end{array}\right\}}\right\}\\
\overset{\text{def.~}assoc}{=}
   & \vspace{5pt}{\left\{\def\arraystretch{1}\begin{array}{l}\arrayrulecolor{white!80!black}{\trace};{\syntrace};~[x{\Mapsto}v]~{\subexp}
\end{array}\right.} &\hspace{-1.5em}\left.\vphantom{\left\{\def\arraystretch{1}\begin{array}{l}  p;q; [x|=>v] e \end{array}\right\}}\right\}\\
\end{longtable}}

\end{proof}

%% file: notes/lemma_mnfe_preserves.txt.tex
\bigskip{}
\begin{globallemma}[mnf preserves traces]
 \emph{$mnf$ is defined as a recursive term-to-term transformation on open terms,}
\emph{whereas traceset equality is defined by reducing terms to values, i.e., on closed terms.}
\emph{Since all valid programs are closed terms,}
\emph{we show that $mnf$ preserves the traceset of an open term $e$ that is closed by substitution $[x\Mapsto v]$.}

\emph{For all terms ${\subexp}$, traces ${\trace}$, traces ${\syntrace}$, values $v$, patterns $x$,}
\[ \{~{\trace};{\syntrace};~[x{\Mapsto}v]~mnf({\subexp})~\} ~\approx~ ... ~\approx~ \{~{\trace};{\syntrace};~[x{\Mapsto}v]~{\subexp}~\} \]
\end{globallemma}\begin{proof} By induction over term structure.
\subparagraph{Case} $ {\subexp}~=~{\subexp}_0~{\subexp}_1 $.

According to $\approx$, we only consider terms that reduce to a value.
Therefore, let $\phi_0$ be the judgement that
the term ${\subexp}_0$ closed by $[x{\Mapsto}v]$ with trace ${\trace}$
evaluates to a value $v_0$ producing trace $p_0$.
Let $\phi_1$ be the judgement that
the term ${\subexp}_1$ closed by $[x{\Mapsto}v]$ with trace ${\trace}~p_0$
evaluates to a value $v_1$ producing trace ${\trace}~p_0~p_1$.

{\small
\[ ~\phi_0~\equiv~~~~({\trace};{\syntrace};~[x{\Mapsto}v]~{\subexp}_0~~~{\to}^*~~{\trace}~p_0;{\syntrace};~v_0) \]
\[ ~\phi_1~\equiv~~~~({\trace}~p_0;{\syntrace};~[x{\Mapsto}v]~{\subexp}_1~~~{\to}^*~~{\trace}~p_0~p_1;{\syntrace};~v_1) \]
}

Let $id_0$ be the fresh identifier $mnf$ produces.

{\small
\[ ~id_0~\freshvar \]
}

The lemma holds by the following chain of transitive relations.
We evaluate the compiled program from top to bottom ($~{\to}^*$)
and the original program from bottom to top ($~{\gets}^*$) until configurations converge.
The induction hypothesis (IH) allows the removal of $mnf$ in redex position under traceset equality ($\approx$).

{\small\begin{longtable}{CLL}
   & \vspace{5pt}{\left\{\def\arraystretch{1}\begin{array}{l}\arrayrulecolor{white!80!black}{\trace};{\syntrace};~[x{\Mapsto}v]~mnf({\subexp})
\end{array}\right.} &\hspace{-1.5em}\left.\vphantom{\left\{\def\arraystretch{1}\begin{array}{l}  p;q; [x|=>v] mnf(e) \end{array}\right\}}\right\}\\
\overset{\text{def.~}{\subexp}}{=}
   & \vspace{5pt}{\left\{\def\arraystretch{1}\begin{array}{l}\arrayrulecolor{white!80!black}{\trace};{\syntrace};~[x{\Mapsto}v]~mnf({\subexp}_0~{\subexp}_1)
\end{array}\right.} &\hspace{-1.5em}\left.\vphantom{\left\{\def\arraystretch{1}\begin{array}{l}  p;q; [x|=>v] mnf(e0 e1) \end{array}\right\}}\right\}\\
\overset{\text{def.~}mnf}{=}
   & \vspace{5pt}{\left\{\def\arraystretch{1}\begin{array}{l}\arrayrulecolor{white!80!black}{\trace};{\syntrace};~[x{\Mapsto}v]~assoc({\ensuremath{\mathsf{let~}}}id_0=mnf({\subexp}_0);~assoc({\ensuremath{\mathsf{let~}}}id_1=mnf({\subexp}_1);~id_0~id_1))
\end{array}\right.} &\hspace{-1.5em}\left.\vphantom{\left\{\def\arraystretch{1}\begin{array}{l}  p;q; [x|=>v] assoc(let id0=mnf(e0); assoc(let id1=mnf(e1); id0 id1)) \end{array}\right\}}\right\}\\
\overset{\text{Lemma~}\ref{lemma:assoc}}{\approx}
   & \vspace{5pt}{\left\{\def\arraystretch{1}\begin{array}{l}\arrayrulecolor{white!80!black}{\trace};{\syntrace};~[x{\Mapsto}v]~{\ensuremath{\mathsf{let~}}}id_0=mnf({\subexp}_0);~assoc({\ensuremath{\mathsf{let~}}}id_1=mnf({\subexp}_1);~id_0~id_1)
\end{array}\right.} &\hspace{-1.5em}\left.\vphantom{\left\{\def\arraystretch{1}\begin{array}{l}  p;q; [x|=>v] let id0=mnf(e0); assoc(let id1=mnf(e1); id0 id1) \end{array}\right\}}\right\}\\
\overset{\text{def.~}{\Mapsto}}{=}
   & \vspace{5pt}{\left\{\def\arraystretch{1}\begin{array}{l}\arrayrulecolor{white!80!black}{\trace};{\syntrace};~{\ensuremath{\mathsf{let~}}}id_0=[x{\Mapsto}v]~mnf({\subexp}_0);~[x{\Mapsto}v]~assoc({\ensuremath{\mathsf{let~}}}id_1=mnf({\subexp}_1);~id_0~id_1)
\end{array}\right.} &\hspace{-1.5em}\left.\vphantom{\left\{\def\arraystretch{1}\begin{array}{l}  p;q; let id0=[x|=>v] mnf(e0); [x|=>v] assoc(let id1=mnf(e1); id0 id1) \end{array}\right\}}\right\}\\
\overset{IH}{\approx}
   & \vspace{5pt}{\left\{\def\arraystretch{1}\begin{array}{l}\arrayrulecolor{white!80!black}{\trace};{\syntrace};~{\ensuremath{\mathsf{let~}}}id_0=[x{\Mapsto}v]~{\subexp}_0;~[x{\Mapsto}v]~assoc({\ensuremath{\mathsf{let~}}}id_1=mnf({\subexp}_1);~id_0~id_1))
\end{array}\right.} &\hspace{-1.5em}\left.\vphantom{\left\{\def\arraystretch{1}\begin{array}{l}  p;q; let id0=[x|=>v] e0; [x|=>v] assoc(let id1=mnf(e1); id0 id1)) \end{array}\right\}}\right\}\\
\overset{\phi_0}{~{\to}^*}
   & \vspace{5pt}{\left\{\def\arraystretch{1}\begin{array}{l}\arrayrulecolor{white!80!black}{\trace}~p_0;{\syntrace};~{\ensuremath{\mathsf{let~}}}id_0=v_0;~[x{\Mapsto}v]~assoc({\ensuremath{\mathsf{let~}}}id_1=mnf({\subexp}_1);~id_0~id_1)~|~\forall~v_0~p_0,~{\ensuremath{\mathsf{if}}}~\phi_0
\end{array}\right.} &\hspace{-1.5em}\left.\vphantom{\left\{\def\arraystretch{1}\begin{array}{l}  p p0;q; let id0=v0; [x|=>v] assoc(let id1=mnf(e1); id0 id1) | \forall v0 p0, if \phi_0 \end{array}\right\}}\right\}\\
\overset{\textsc{Rlet}}{~{\to}}
   & \vspace{5pt}{\left\{\def\arraystretch{1}\begin{array}{l}\arrayrulecolor{white!80!black}{\trace}~p_0;{\syntrace};~[id_0{\mapsto}v_0,~x{\Mapsto}v]~assoc({\ensuremath{\mathsf{let~}}}id_1=mnf({\subexp}_1);~id_0~id_1)~|~\forall~v_0~p_0,~{\ensuremath{\mathsf{if}}}~\phi_0
\end{array}\right.} &\hspace{-1.5em}\left.\vphantom{\left\{\def\arraystretch{1}\begin{array}{l}  p p0;q; [id0|->v0, x|=>v] assoc(let id1=mnf(e1); id0 id1) | \forall v0 p0, if \phi_0 \end{array}\right\}}\right\}\\
\overset{\text{Lemma~}\ref{lemma:assoc}}{\approx}
   & \vspace{5pt}{\left\{\def\arraystretch{1}\begin{array}{l}\arrayrulecolor{white!80!black}{\trace}~p_0;{\syntrace};~[id_0{\mapsto}v_0,~x{\Mapsto}v]~{\ensuremath{\mathsf{let~}}}id_1=mnf({\subexp}_1);~id_0~id_1~|~\forall~v_0~p_0,~{\ensuremath{\mathsf{if}}}~\phi_0
\end{array}\right.} &\hspace{-1.5em}\left.\vphantom{\left\{\def\arraystretch{1}\begin{array}{l}  p p0;q; [id0|->v0, x|=>v] let id1=mnf(e1); id0 id1 | \forall v0 p0, if \phi_0 \end{array}\right\}}\right\}\\
\overset{\text{def.~}{\Mapsto}}{=}
   & \vspace{5pt}{\left\{\def\arraystretch{1}\begin{array}{l}\arrayrulecolor{white!80!black}{\trace}~p_0;{\syntrace};~{\ensuremath{\mathsf{let~}}}id_1=[id_0{\mapsto}v_0,~x{\Mapsto}v]~mnf({\subexp}_1);~v_0~id_1~|~\forall~v_0~p_0,~{\ensuremath{\mathsf{if}}}~\phi_0
\end{array}\right.} &\hspace{-1.5em}\left.\vphantom{\left\{\def\arraystretch{1}\begin{array}{l}  p p0;q; let id1=[id0|->v0, x|=>v] mnf(e1); v0 id1 | \forall v0 p0, if \phi_0 \end{array}\right\}}\right\}\\
\overset{IH}{=}
   & \vspace{5pt}{\left\{\def\arraystretch{1}\begin{array}{l}\arrayrulecolor{white!80!black}{\trace}~p_0;{\syntrace};~{\ensuremath{\mathsf{let~}}}id_1=[id_0{\mapsto}v_0,~x{\Mapsto}v]~{\subexp}_1;~v_0~id_1~|~\forall~v_0~p_0,~{\ensuremath{\mathsf{if}}}~\phi_0
\end{array}\right.} &\hspace{-1.5em}\left.\vphantom{\left\{\def\arraystretch{1}\begin{array}{l}  p p0;q; let id1=[id0|->v0, x|=>v] e1; v0 id1 | \forall v0 p0, if \phi_0 \end{array}\right\}}\right\}\\
\overset{id_0~\freshvar}{=}
   & \vspace{5pt}{\left\{\def\arraystretch{1}\begin{array}{l}\arrayrulecolor{white!80!black}{\trace}~p_0;{\syntrace};~{\ensuremath{\mathsf{let~}}}id_1=[x{\Mapsto}v]~{\subexp}_1;~v_0~id_1~|~\forall~v_0~p_0,~{\ensuremath{\mathsf{if}}}~\phi_0
\end{array}\right.} &\hspace{-1.5em}\left.\vphantom{\left\{\def\arraystretch{1}\begin{array}{l}  p p0;q; let id1=[x|=>v] e1; v0 id1 | \forall v0 p0, if \phi_0 \end{array}\right\}}\right\}\\
\overset{\phi_1}{~{\to}^*}
   & \vspace{5pt}{\left\{\def\arraystretch{1}\begin{array}{l}\arrayrulecolor{white!80!black}{\trace}~p_0~p_1;{\syntrace};~{\ensuremath{\mathsf{let~}}}id_1=v_1;~v_0~id_1~|~\forall~v_0~v_1~p_0~p_1,~{\ensuremath{\mathsf{if}}}~\phi_0,~\phi_1
\end{array}\right.} &\hspace{-1.5em}\left.\vphantom{\left\{\def\arraystretch{1}\begin{array}{l}  p p0 p1;q; let id1=v1; v0 id1 | \forall v0 v1 p0 p1, if \phi_0, \phi_1 \end{array}\right\}}\right\}\\
\overset{\textsc{Rlet}}{~{\to}}
   & \vspace{5pt}{\left\{\def\arraystretch{1}\begin{array}{l}\arrayrulecolor{white!80!black}{\trace}~p_0~p_1;{\syntrace};~v_0~v_1~|~\forall~v_0~v_1~p_0~p_1,~{\ensuremath{\mathsf{if}}}~\phi_0,~\phi_1
\end{array}\right.} &\hspace{-1.5em}\left.\vphantom{\left\{\def\arraystretch{1}\begin{array}{l}  p p0 p1;q; v0 v1 | \forall v0 v1 p0 p1, if \phi_0, \phi_1 \end{array}\right\}}\right\}\\
\overset{\phi_1}{~{\gets}^*}
   & \vspace{5pt}{\left\{\def\arraystretch{1}\begin{array}{l}\arrayrulecolor{white!80!black}{\trace}~p_0;{\syntrace};~v_0~[x{\Mapsto}v]~{\subexp}_1~|~\forall~v_0~p_0,~{\ensuremath{\mathsf{if}}}~\phi_0
\end{array}\right.} &\hspace{-1.5em}\left.\vphantom{\left\{\def\arraystretch{1}\begin{array}{l}  p p0;q; v0 [x|=>v] e1 | \forall v0 p0, if \phi_0 \end{array}\right\}}\right\}\\
\overset{\phi_0}{~{\gets}^*}
   & \vspace{5pt}{\left\{\def\arraystretch{1}\begin{array}{l}\arrayrulecolor{white!80!black}{\trace};{\syntrace};~([x{\Mapsto}v]~{\subexp}_0)~[x{\Mapsto}v]~{\subexp}_1
\end{array}\right.} &\hspace{-1.5em}\left.\vphantom{\left\{\def\arraystretch{1}\begin{array}{l}  p;q; ([x|=>v] e0) [x|=>v] e1 \end{array}\right\}}\right\}\\
\overset{\text{def.~}{\Mapsto}}{=}
   & \vspace{5pt}{\left\{\def\arraystretch{1}\begin{array}{l}\arrayrulecolor{white!80!black}{\trace};{\syntrace};~[x{\Mapsto}v]~{\subexp}_0~{\subexp}_1
\end{array}\right.} &\hspace{-1.5em}\left.\vphantom{\left\{\def\arraystretch{1}\begin{array}{l}  p;q; [x|=>v] e0 e1 \end{array}\right\}}\right\}\\
\overset{\text{def.~}{\subexp}}{=}
   & \vspace{5pt}{\left\{\def\arraystretch{1}\begin{array}{l}\arrayrulecolor{white!80!black}{\trace};{\syntrace};~[x{\Mapsto}v]~{\subexp}
\end{array}\right.} &\hspace{-1.5em}\left.\vphantom{\left\{\def\arraystretch{1}\begin{array}{l}  p;q; [x|=>v] e \end{array}\right\}}\right\}\\
\end{longtable}}

\subparagraph{Case} $ ~~{\subexp}~=~{\ensuremath{\mathsf{let~}}}id_0={\subexp}_0;~{\subexp}_1 $.

According to $\approx$, we only consider terms that reduce to a value.
Therefore, let $\phi$ be the judgement that
the term ${\subexp}_0$ closed by $[x{\Mapsto}v]$ with trace ${\trace}$
evaluates to a value $v_0$ producing trace $p_0$.

{\small
\[ ~\phi_0~\equiv~~~~({\trace};{\syntrace};~[x{\Mapsto}v]~{\subexp}_0~~~{\to}^*~~{\trace}~p_0;{\syntrace};~v_0) \]
}

The lemma holds by the following chain of transitive relations.
We evaluate the compiled program from top to bottom ($~{\to}^*$)
and the original program from bottom to top ($~{\gets}^*$) until configurations converge.
The induction hypothesis (IH) allows the removal of $mnf$ in redex position under traceset equality ($\approx$).

{\small\begin{longtable}{CLL}
   & \vspace{5pt}{\left\{\def\arraystretch{1}\begin{array}{l}\arrayrulecolor{white!80!black}{\trace};{\syntrace};~[x{\Mapsto}v]~mnf({\subexp})
\end{array}\right.} &\hspace{-1.5em}\left.\vphantom{\left\{\def\arraystretch{1}\begin{array}{l}  p;q; [x|=>v] mnf(e) \end{array}\right\}}\right\}\\
\overset{\text{def.~}{\subexp}}{=}
   & \vspace{5pt}{\left\{\def\arraystretch{1}\begin{array}{l}\arrayrulecolor{white!80!black}{\trace};{\syntrace};~[x{\Mapsto}v]~mnf({\ensuremath{\mathsf{let~}}}id_0=v_0;~{\subexp}_1)
\end{array}\right.} &\hspace{-1.5em}\left.\vphantom{\left\{\def\arraystretch{1}\begin{array}{l}  p;q; [x|=>v] mnf(let id0=v0; e1) \end{array}\right\}}\right\}\\
\overset{\text{def.~}mnf}{=}
   & \vspace{5pt}{\left\{\def\arraystretch{1}\begin{array}{l}\arrayrulecolor{white!80!black}{\trace};{\syntrace};~[x{\Mapsto}v]~assoc({\ensuremath{\mathsf{let~}}}id_0=mnf({\subexp}_0);~mnf({\subexp}_1))
\end{array}\right.} &\hspace{-1.5em}\left.\vphantom{\left\{\def\arraystretch{1}\begin{array}{l}  p;q; [x|=>v] assoc(let id0=mnf(e0); mnf(e1)) \end{array}\right\}}\right\}\\
\overset{\text{Lemma~}\ref{lemma:assoc}}{\approx}
   & \vspace{5pt}{\left\{\def\arraystretch{1}\begin{array}{l}\arrayrulecolor{white!80!black}{\trace};{\syntrace};~[x{\Mapsto}v]~{\ensuremath{\mathsf{let~}}}id_0=mnf({\subexp}_0);~mnf({\subexp}_1)
\end{array}\right.} &\hspace{-1.5em}\left.\vphantom{\left\{\def\arraystretch{1}\begin{array}{l}  p;q; [x|=>v] let id0=mnf(e0); mnf(e1) \end{array}\right\}}\right\}\\
\overset{\text{def.~}{\Mapsto}}{=}
   & \vspace{5pt}{\left\{\def\arraystretch{1}\begin{array}{l}\arrayrulecolor{white!80!black}{\trace};{\syntrace};~{\ensuremath{\mathsf{let~}}}id_0=[x{\Mapsto}v]~mnf({\subexp}_0);~[x{\Mapsto}v]~mnf({\subexp}_1)
\end{array}\right.} &\hspace{-1.5em}\left.\vphantom{\left\{\def\arraystretch{1}\begin{array}{l}  p;q; let id0=[x|=>v] mnf(e0); [x|=>v] mnf(e1) \end{array}\right\}}\right\}\\
\overset{IH}{\approx}
   & \vspace{5pt}{\left\{\def\arraystretch{1}\begin{array}{l}\arrayrulecolor{white!80!black}{\trace};{\syntrace};~{\ensuremath{\mathsf{let~}}}id_0=[x{\Mapsto}v]~{\subexp}_0;~[x{\Mapsto}v]~mnf({\subexp}_1)
\end{array}\right.} &\hspace{-1.5em}\left.\vphantom{\left\{\def\arraystretch{1}\begin{array}{l}  p;q; let id0=[x|=>v] e0; [x|=>v] mnf(e1) \end{array}\right\}}\right\}\\
\overset{\phi_0}{~{\to}^*}
   & \vspace{5pt}{\left\{\def\arraystretch{1}\begin{array}{l}\arrayrulecolor{white!80!black}{\trace}~p_0;{\syntrace};~{\ensuremath{\mathsf{let~}}}id_0=v_0;~[x{\Mapsto}v]~mnf({\subexp}_1)~|~\forall~v_0~p_0,~{\ensuremath{\mathsf{if}}}~\phi_0
\end{array}\right.} &\hspace{-1.5em}\left.\vphantom{\left\{\def\arraystretch{1}\begin{array}{l}  p p0;q; let id0=v0; [x|=>v] mnf(e1) | \forall v0 p0, if \phi_0 \end{array}\right\}}\right\}\\
\overset{\textsc{Rlet}}{~{\to}}
   & \vspace{5pt}{\left\{\def\arraystretch{1}\begin{array}{l}\arrayrulecolor{white!80!black}{\trace}~p_0;{\syntrace};~[id_0{\mapsto}v_0,~x{\Mapsto}v]~mnf({\subexp}_1)~|~\forall~v_0~p_0,~{\ensuremath{\mathsf{if}}}~\phi_0
\end{array}\right.} &\hspace{-1.5em}\left.\vphantom{\left\{\def\arraystretch{1}\begin{array}{l}  p p0;q; [id0|->v0, x|=>v] mnf(e1) | \forall v0 p0, if \phi_0 \end{array}\right\}}\right\}\\
\overset{IH}{\approx}
   & \vspace{5pt}{\left\{\def\arraystretch{1}\begin{array}{l}\arrayrulecolor{white!80!black}{\trace}~p_0;{\syntrace};~[id_0{\mapsto}v_0,~x{\Mapsto}v]~{\subexp}_1~|~\forall~v_0~p_0,~{\ensuremath{\mathsf{if}}}~\phi_0
\end{array}\right.} &\hspace{-1.5em}\left.\vphantom{\left\{\def\arraystretch{1}\begin{array}{l}  p p0;q; [id0|->v0, x|=>v] e1 | \forall v0 p0, if \phi_0 \end{array}\right\}}\right\}\\
\overset{\textsc{Rlet}}{~{\gets}}
   & \vspace{5pt}{\left\{\def\arraystretch{1}\begin{array}{l}\arrayrulecolor{white!80!black}{\trace}~p_0;{\syntrace};~[x{\Mapsto}v]~{\ensuremath{\mathsf{let~}}}id_0=v_0;~{\subexp}_1~|~\forall~v_0~p_0,~{\ensuremath{\mathsf{if}}}~\phi_0
\end{array}\right.} &\hspace{-1.5em}\left.\vphantom{\left\{\def\arraystretch{1}\begin{array}{l}  p p0;q; [x|=>v] let id0=v0; e1 | \forall v0 p0, if \phi_0 \end{array}\right\}}\right\}\\
\overset{\phi_0}{~{\gets}^*}
   & \vspace{5pt}{\left\{\def\arraystretch{1}\begin{array}{l}\arrayrulecolor{white!80!black}{\trace};{\syntrace};~[x{\Mapsto}v]~{\ensuremath{\mathsf{let~}}}id_0={\subexp}_0;~{\subexp}_1
\end{array}\right.} &\hspace{-1.5em}\left.\vphantom{\left\{\def\arraystretch{1}\begin{array}{l}  p;q; [x|=>v] let id0=e0; e1 \end{array}\right\}}\right\}\\
\overset{\text{def.~}{\subexp}}{=}
   & \vspace{5pt}{\left\{\def\arraystretch{1}\begin{array}{l}\arrayrulecolor{white!80!black}{\trace};{\syntrace};~[x{\Mapsto}v]~{\subexp}
\end{array}\right.} &\hspace{-1.5em}\left.\vphantom{\left\{\def\arraystretch{1}\begin{array}{l}  p;q; [x|=>v] e \end{array}\right\}}\right\}\\
\end{longtable}}

\subparagraph{Case} $ \text{The~other~cases~of~${\subexp}$~are~proved~analogously} $.

\end{proof}

%% file: notes/lemma_mnfp_preserves.txt.tex
\bigskip{}
\begin{globallemma}[mnf' preserves trace]
 \emph{$mnf'$ is defined on programs.}
\emph{To evaluate a program, it is initialized with a set of clients $A$.}
\emph{$mnf'$ preserves the traceset of (closed) programs $P$ for any set of clients $A$.}

\emph{For all  $P$,}
\[ \{~init_A(mnf'(P))~\} ~\approx~ ... ~\approx~ \{~init_A(P)~\} \]
\end{globallemma}\begin{proof} By induction over term structure.
\subparagraph{Case} $ ~~P~=~({\decl};{\syndecl};~{\subexp}_0) $.

Initializing the definitions ${\decl};{\syndecl}$ with $A$ produces the trace ${\trace}$ and the state ${\syntrace}$.

{\small
\[ ~init_A({\decl};{\syndecl})~=~{\trace};{\syntrace} \]
}

According to $\approx$, we only consider terms that reduce to a value.
Therefore, let $\phi$ be the judgement that
the term ${\subexp}_0$ closed by $[x{\Mapsto}v]$ in trace ${\trace}$
produces a value $v_0$ and trace $p_0$.

{\small
\[ ~\phi~\equiv~~~~({\trace};{\syntrace};~{\subexp}_0~~~{\to}^*~~{\trace}~p_0;{\syntrace};~v_0) \]
}

The lemma holds by the following chain of transitive relations.
We evaluate the compiled program from top to bottom ($~{\to}^*$)
and the original program from bottom to top ($~{\gets}^*$) until configurations converge,
using Lemma \ref{lemma:mnf}.

{\small\begin{longtable}{CLL}
   & \vspace{5pt}{\left\{\def\arraystretch{1}\begin{array}{l}\arrayrulecolor{white!80!black}init_A(mnf'(P))
\end{array}\right.} &\hspace{-1.5em}\left.\vphantom{\left\{\def\arraystretch{1}\begin{array}{l}  init_A(mnf'(P)) \end{array}\right\}}\right\}\\
\overset{\text{def.~}P}{=}
   & \vspace{5pt}{\left\{\def\arraystretch{1}\begin{array}{l}\arrayrulecolor{white!80!black}init_A(mnf'({\decl};{\syndecl};~{\subexp}_0))
\end{array}\right.} &\hspace{-1.5em}\left.\vphantom{\left\{\def\arraystretch{1}\begin{array}{l}  init_A(mnf'(d;b; e0)) \end{array}\right\}}\right\}\\
\overset{\text{def.~}mnf'}{=}
   & \vspace{5pt}{\left\{\def\arraystretch{1}\begin{array}{l}\arrayrulecolor{white!80!black}init_A({\decl};{\syndecl};~{\trampoline}(mnfe({\ensuremath{\mathsf{Done}}}({\subexp}_0))))
\end{array}\right.} &\hspace{-1.5em}\left.\vphantom{\left\{\def\arraystretch{1}\begin{array}{l}  init_A(d;b; tramp(mnfe(Done(e0)))) \end{array}\right\}}\right\}\\
\overset{\text{def.~}init_A}{=}
   & \vspace{5pt}{\left\{\def\arraystretch{1}\begin{array}{l}\arrayrulecolor{white!80!black}{\trace};{\syntrace};~{\trampoline}(mnfe({\ensuremath{\mathsf{Done}}}({\subexp}_0)))
\end{array}\right.} &\hspace{-1.5em}\left.\vphantom{\left\{\def\arraystretch{1}\begin{array}{l}  p;q; tramp(mnfe(Done(e0))) \end{array}\right\}}\right\}\\
\overset{\text{Lemma~}\ref{lemma:mnf}}{\approx}
   & \vspace{5pt}{\left\{\def\arraystretch{1}\begin{array}{l}\arrayrulecolor{white!80!black}{\trace};{\syntrace};~{\trampoline}({\ensuremath{\mathsf{Done}}}({\subexp}_0))
\end{array}\right.} &\hspace{-1.5em}\left.\vphantom{\left\{\def\arraystretch{1}\begin{array}{l}  p;q; tramp(Done(e0)) \end{array}\right\}}\right\}\\
\overset{\phi}{~{\to}^*}
   & \vspace{5pt}{\left\{\def\arraystretch{1}\begin{array}{l}\arrayrulecolor{white!80!black}{\trace}~p_0;{\syntrace};~{\trampoline}({\ensuremath{\mathsf{Done}}}(v_0))~|~\forall~v_0~p_0,~{\ensuremath{\mathsf{if}}}~\phi
\end{array}\right.} &\hspace{-1.5em}\left.\vphantom{\left\{\def\arraystretch{1}\begin{array}{l}  p p0;q; tramp(Done(v0)) | \forall v0 p0, if \phi \end{array}\right\}}\right\}\\
\overset{\textsc{Rdone}}{~{\to}}
   & \vspace{5pt}{\left\{\def\arraystretch{1}\begin{array}{l}\arrayrulecolor{white!80!black}{\trace}~p_0;{\syntrace};~v_0~|~\forall~v_0~p_0,~{\ensuremath{\mathsf{if}}}~\phi
\end{array}\right.} &\hspace{-1.5em}\left.\vphantom{\left\{\def\arraystretch{1}\begin{array}{l}  p p0;q; v0 | \forall v0 p0, if \phi \end{array}\right\}}\right\}\\
\overset{\phi}{~{\gets}^*}
   & \vspace{5pt}{\left\{\def\arraystretch{1}\begin{array}{l}\arrayrulecolor{white!80!black}{\trace};{\syntrace};~{\subexp}_0
\end{array}\right.} &\hspace{-1.5em}\left.\vphantom{\left\{\def\arraystretch{1}\begin{array}{l}  p;q; e0 \end{array}\right\}}\right\}\\
\overset{\text{def.~}init_A}{=}
   & \vspace{5pt}{\left\{\def\arraystretch{1}\begin{array}{l}\arrayrulecolor{white!80!black}init_A({\decl};{\syndecl};~{\subexp}_0)
\end{array}\right.} &\hspace{-1.5em}\left.\vphantom{\left\{\def\arraystretch{1}\begin{array}{l}  init_A(d;b; e0) \end{array}\right\}}\right\}\\
\overset{\text{def.~}P}{=}
   & \vspace{5pt}{\left\{\def\arraystretch{1}\begin{array}{l}\arrayrulecolor{white!80!black}init_A(P)
\end{array}\right.} &\hspace{-1.5em}\left.\vphantom{\left\{\def\arraystretch{1}\begin{array}{l}  init_A(P) \end{array}\right\}}\right\}\\
\end{longtable}}

\end{proof}

%% file: notes/lemma_dtle_preserves.txt.tex
\bigskip{}
\begin{globallemma}[comp preserves traces]
 \emph{$comp$ is defined on programs.}
\emph{To evaluate a program, it is initialized with a set of clients $A$.}
\emph{$comp$ preserves the traceset of (closed) programs $P$ for any set of clients $A$.}

\emph{For all definitions ${\syndecl}$, definitions ${\decl}$, terms ${\subexp}$, values $v$, patterns $x$,}
\[ \{~[x{\Mapsto}v]~init_A(comp({\decl};{\syndecl};~{\trampoline}({\subexp})))~\} ~\approx~ ... ~\approx~ \{~init_A({\decl};{\syndecl};~{\trampoline}([x{\Mapsto}v]~{\subexp}))~\} \]
\end{globallemma}\begin{proof} By induction over term structure.
\subparagraph{Case} $ ~~~{\subexp}~=~{\ensuremath{\mathsf{let~}}}x={\downarrow}_s(({\subexp}_0,~()~{\to}{\subexp}_1));~{\subexp}_2 $.

$comp$ expects the definitions ${\syndecl}$ to be of form:

{\small
\[ {\syndecl}~=~\left(\begin{array}{l}~{\ensuremath{\mathsf{@cl}}}~\this.{\ensuremath{\mathsf{clfn}}}~=~id~{\to}~{\subexp}_{1,alt};\\~{\ensuremath{\mathsf{@co}}}~\this.{\ensuremath{\mathsf{cofn}}}~=~id~{\to}~{\subexp}_{2,alt};\\~b_{rest}\end{array}\right) \]
}

$comp$ is defined recursively and applied to the term ${\subexp}_2$.
Intuitively, $comp$ transforms ${\subexp}_2$ to ${\subexp}_2'$ and ${\syndecl}$ to ${\syndecl}'$ by moving the part of ${\subexp}_2$ that comes
after the ${\downarrow}_s$ call into the ${\ensuremath{\mathsf{cofn}}}$ definition inside ${\syndecl}$.
The recursive call is given as follows:

{\small
\[ ({\decl};{\syndecl}';~{\trampoline}({\subexp}_2'))~=~comp({\decl};{\syndecl};~{\trampoline}({\subexp}_2)) \]
\[ {\syndecl}'~=~\left(\begin{array}{l}~{\ensuremath{\mathsf{@cl}}}~\this.{\ensuremath{\mathsf{clfn}}}~=~id~{\to}~{\subexp}_{1,alt}';\\~{\ensuremath{\mathsf{@co}}}~\this.{\ensuremath{\mathsf{cofn}}}~=~id~{\to}~{\subexp}_{2,alt}';\\~b_{rest}\end{array}\right) \]
}

After the recursive call, $comp$ moves the
transformed ${\subexp}_2'$ into the ${\ensuremath{\mathsf{cofn}}}$ definition,
resulting in ${\subexp}'$ and ${\syndecl}''$
with ${\subexp}_{1,alt}''$ and ${\subexp}_{2,alt}''$.

{\small
\[ ~\phi~\equiv~~~~\big(~\{~{\decl};{\syndecl}'';~{\trampoline}({\subexp}')~\big\}~=~\big\{~comp({\decl};{\syndecl};~{\trampoline}({\subexp}))~\}~\big) \]
\[ {\syndecl}''~=~\left(\begin{array}{l}~{\ensuremath{\mathsf{@cl}}}~\this.{\ensuremath{\mathsf{clfn}}}~=~id~{\to}~{\subexp}_{1,alt}'';\\~{\ensuremath{\mathsf{@co}}}~\this.{\ensuremath{\mathsf{cofn}}}~=~id~{\to}~{\subexp}_{2,alt}'';\\~b_{rest}\end{array}\right) \]
\[ {\subexp}_{1,alt}''~=~\left(\begin{array}{l}~{\ensuremath{\mathsf{if}}}~{\ensuremath{\mathsf{let~}}}(c~{\ensuremath{\mathsf{::}}}~\fv(()~{\to}{\subexp}_1))~=~id\\~{\ensuremath{\mathsf{then}}}~{\subexp}_1\\~{\ensuremath{\mathsf{else}}}~{\subexp}_{2,alt}'\end{array}\right) \]
\[ {\subexp}_{2,alt}''~=~\left(\begin{array}{l}~{\ensuremath{\mathsf{if}}}~{\ensuremath{\mathsf{let~}}}(c~{\ensuremath{\mathsf{::}}}~x~{\ensuremath{\mathsf{::}}}~\fv(x~{\to}{\subexp}_2'))~=~id\\~{\ensuremath{\mathsf{then}}}~{\ensuremath{\mathsf{assert}}}(\this.{\ensuremath{\mathsf{state}}}~{\text{==}}~c~{\text{\&\&}}~\this.{\ensuremath{\mathsf{sender}}}~{\text{==}}~\this.{\ensuremath{\mathsf{who}}});\\~~~~~~~~~\this.{\ensuremath{\mathsf{state}}}~{\ensuremath{\mathsf{:=}}}~0;~{\subexp}_2'\\~{\ensuremath{\mathsf{else}}}~{\subexp}_{2,alt}'\end{array}\right) \]
}

Let ${\trace};{\syntrace}$ be the trace and state produced by initializing ${\decl};{\syndecl}$ with $A$,
and ${\trace};{\syntrace}'$ for initializing ${\decl};{\syndecl}'$, and ${\trace};{\syntrace}''$ for initializing ${\decl};{\syndecl}''$.

{\small
\[ init_A({\decl};{\syndecl})~~~=~{\trace};{\syntrace} \]
\[ init_A({\decl};{\syndecl}')~~=~{\trace};{\syntrace}' \]
\[ init_A({\decl};{\syndecl}'')~=~{\trace};{\syntrace}'' \]
}

According to $\approx$, we only consider terms that reduce to a value.
Therefore, let $\phi_0$ be the judgement that
the term ${\subexp}_0$ closed by $[x{\Mapsto}v]$ in trace ${\trace}$
produces a value $v_0$ and trace $p_1$.

{\small
\[ \phi_0(q_\phi)~=~({\trace};q_\phi;~[x{\Mapsto}v]~{\subexp}_0~~{\to}~{\trace}~p_1;q_\phi;~v_0) \]
}

We define $\phi_1$ based on $\phi$:

{\small
\[ ~~\phi \]
\[ = \]
\[ ~~\big\{~{\decl};{\syndecl}'';~{\trampoline}({\subexp}')~\big\}~=~\big\{~comp({\decl};{\syndecl};~{\trampoline}({\subexp}))~\big\} \]
\[ ~{\to}~generalize~[x{\Mapsto}v]~init_A(...) \]
\[ ~~\big\{~[x{\Mapsto}v]~init_A({\decl};{\syndecl}'';~{\trampoline}({\subexp}'))~\big\}~=~\big\{~[x{\Mapsto}v]~init_A(comp({\decl};{\syndecl};~{\trampoline}({\subexp})))~\big\} \]
\[ ~{\to}~(=~~{\to}~\approx) \]
\[ ~~\big\{~[x{\Mapsto}v]~init_A({\decl};{\syndecl}'';~{\trampoline}({\subexp}'))~\big\}~\approx~\big\{~[x{\Mapsto}v]~init_A(comp({\decl};{\syndecl};~{\trampoline}({\subexp})))~\big\} \]
\[ ~{\to}~IH \]
\[ ~~\big\{~[x{\Mapsto}v]~init_A({\decl};{\syndecl}'';~{\trampoline}({\subexp}'))~\big\}~\approx~\big\{~init_A({\decl};{\syndecl};~{\trampoline}([x{\Mapsto}v]~{\subexp}))~\big\} \]
\[ ~{\to}~\text{def.~}init_A \]
\[ ~~\big\{~{\trace};{\syntrace}'';~{\trampoline}([x{\Mapsto}v]~{\subexp}')~\big\}~\approx~\big\{~{\trace};{\syntrace};~{\trampoline}([x{\Mapsto}v]~{\subexp})~\big\} \]
\[ \equiv \]
\[ ~~\phi_1 \]
}

The lemma holds by the following chain of transitive relations.
We evaluate the compiled program from top to bottom ($~{\to}^*$)
and the original program from bottom to top ($~{\gets}^*$) until configurations converge.

{\small\begin{longtable}{CLL}
   & \vspace{5pt}{\left\{\def\arraystretch{1}\begin{array}{l}\arrayrulecolor{white!80!black}[x{\Mapsto}v]~init_A(comp({\decl};{\syndecl};~{\trampoline}({\subexp})))
\end{array}\right.} &\hspace{-1.5em}\left.\vphantom{\left\{\def\arraystretch{1}\begin{array}{l}  [x|=>v] init_A(comp(d;b; tramp(e))) \end{array}\right\}}\right\}\\
\overset{\text{def.~}{\subexp}}{=}
   & \vspace{5pt}{\left\{\def\arraystretch{1}\begin{array}{l}\arrayrulecolor{white!80!black}[x{\Mapsto}v]~init_A(comp({\decl};{\syndecl};~{\trampoline}({\ensuremath{\mathsf{let~}}}x_3={\downarrow}_s(({\subexp}_0,~()~{\to}{\subexp}_1)));~{\subexp}_2))
\end{array}\right.} &\hspace{-1.5em}\left.\vphantom{\left\{\def\arraystretch{1}\begin{array}{l}  [x|=>v] init_A(comp(d;b; tramp(let x3=↓_s((e0, ()->e1))); e2)) \end{array}\right\}}\right\}\\
\overset{\text{def.~}comp}{=}
   & \vspace{5pt}{\left\{\def\arraystretch{1}\begin{array}{l}\arrayrulecolor{white!80!black}[x{\Mapsto}v]~init_A({\decl};{\syndecl}'';~{\trampoline}(
\\\this.{\ensuremath{\mathsf{who}}}~{\ensuremath{\mathsf{:=}}}~{\subexp}_0;~\this.{\ensuremath{\mathsf{state}}}~{\ensuremath{\mathsf{:=}}}~c;
\\{\ensuremath{\mathsf{More}}}(c~{\ensuremath{\mathsf{::}}}~\fv(()~{\to}{\subexp}_1),~c~{\ensuremath{\mathsf{::}}}~\fv(x~{\to}{\subexp}_2'))))
\end{array}\right.} &\hspace{-1.5em}\left.\vphantom{\left\{\def\arraystretch{1}\begin{array}{l}  [x|=>v] init_A(d;b''; tramp( \\  this.who := e0; this.state := c; \\  More(c :: fv(()->e1), c :: fv(x->e2')))) \end{array}\right\}}\right\}\\
\overset{\text{def.~}init_A}{=}
   & \vspace{5pt}{\left\{\def\arraystretch{1}\begin{array}{l}\arrayrulecolor{white!80!black}{\trace};{\syntrace}'';~[x{\Mapsto}v]~{\trampoline}(
\\\this.{\ensuremath{\mathsf{who}}}~{\ensuremath{\mathsf{:=}}}~{\subexp}_0;~\this.{\ensuremath{\mathsf{state}}}~{\ensuremath{\mathsf{:=}}}~c;
\\{\ensuremath{\mathsf{More}}}(c~{\ensuremath{\mathsf{::}}}~\fv(()~{\to}{\subexp}_1),~c~{\ensuremath{\mathsf{::}}}~\fv(x~{\to}{\subexp}_2'))
\end{array}\right.} &\hspace{-1.5em}\left.\vphantom{\left\{\def\arraystretch{1}\begin{array}{l}  p;q''; [x|=>v] tramp( \\  this.who := e0; this.state := c; \\  More(c :: fv(()->e1), c :: fv(x->e2')) \end{array}\right\}}\right\}\\
\overset{\text{def.~}{\Mapsto}}{=}
   & \vspace{5pt}{\left\{\def\arraystretch{1}\begin{array}{l}\arrayrulecolor{white!80!black}{\trace};{\syntrace}'';~{\trampoline}(
\\\this.{\ensuremath{\mathsf{who}}}~{\ensuremath{\mathsf{:=}}}~[x{\Mapsto}v]~{\subexp}_0;~\this.{\ensuremath{\mathsf{state}}}~{\ensuremath{\mathsf{:=}}}~c;
\\{\ensuremath{\mathsf{More}}}(c~{\ensuremath{\mathsf{::}}}~[x{\Mapsto}v]~\fv(()~{\to}{\subexp}_1),~c~{\ensuremath{\mathsf{::}}}~[x{\Mapsto}v]~\fv(x~{\to}{\subexp}_2')))
\end{array}\right.} &\hspace{-1.5em}\left.\vphantom{\left\{\def\arraystretch{1}\begin{array}{l}  p;q''; tramp( \\  this.who := [x|=>v] e0; this.state := c; \\  More(c :: [x|=>v] fv(()->e1), c :: [x|=>v] fv(x->e2'))) \end{array}\right\}}\right\}\\
\overset{\phi_0({\syntrace}'')}{~{\to}^*}
   & \vspace{5pt}{\left\{\def\arraystretch{1}\begin{array}{l}\arrayrulecolor{white!80!black}{\trace}~p_1;~{\syntrace}'';~{\trampoline}(
\\\this.{\ensuremath{\mathsf{who}}}~{\ensuremath{\mathsf{:=}}}~v_0;~\this.{\ensuremath{\mathsf{state}}}~{\ensuremath{\mathsf{:=}}}~c;
\\{\ensuremath{\mathsf{More}}}(c~{\ensuremath{\mathsf{::}}}~[x{\Mapsto}v]~\fv(()~{\to}{\subexp}_1),~c~{\ensuremath{\mathsf{::}}}~[x{\Mapsto}v]~\fv(x~{\to}{\subexp}_2')))
\\|~\forall~v_0~p_1,~{\ensuremath{\mathsf{if}}}~\phi_0
\end{array}\right.} &\hspace{-1.5em}\left.\vphantom{\left\{\def\arraystretch{1}\begin{array}{l}  p p1; q''; tramp( \\  this.who := v0; this.state := c; \\  More(c :: [x|=>v] fv(()->e1), c :: [x|=>v] fv(x->e2'))) \\  | \forall v0 p1, if \phi_0 \end{array}\right\}}\right\}\\
\overset{\textsc{Rset}\dagger,~\textsc{Rset}\dagger}{~{\to}}
   & \vspace{5pt}{\left\{\def\arraystretch{1}\begin{array}{l}\arrayrulecolor{white!80!black}{\trace}~p_1;~{\syntrace}''~[{\ensuremath{\mathsf{who}}}{\mapsto}v_0,~{\ensuremath{\mathsf{state}}}{\mapsto}c];
\\{\trampoline}({\ensuremath{\mathsf{More}}}(c~{\ensuremath{\mathsf{::}}}~[x{\Mapsto}v]~\fv(()~{\to}{\subexp}_1)),~c~{\ensuremath{\mathsf{::}}}~[x{\Mapsto}v]~\fv(x~{\to}{\subexp}_2'))
\\|~\forall~v_0~p_1,~{\ensuremath{\mathsf{if}}}~\phi_0
\end{array}\right.} &\hspace{-1.5em}\left.\vphantom{\left\{\def\arraystretch{1}\begin{array}{l}  p p1; q'' [who|->v0, state|->c]; \\  tramp(More(c :: [x|=>v] fv(()->e1)), c :: [x|=>v] fv(x->e2')) \\  | \forall v0 p1, if \phi_0 \end{array}\right\}}\right\}\\
\overset{\textsc{Rmore}}{~{\to}}
   & \vspace{5pt}{\left\{\def\arraystretch{1}\begin{array}{l}\arrayrulecolor{white!80!black}{\trace}~p_1;~{\syntrace}''~[{\ensuremath{\mathsf{who}}}{\mapsto}v_0,~{\ensuremath{\mathsf{state}}}{\mapsto}c];
\\tmp~~{\gets}_t~\this.{\ensuremath{\mathsf{clfn}}}(c~{\ensuremath{\mathsf{::}}}~[x{\Mapsto}v]~\fv(()~{\to}{\subexp}_1));
\\{\trampoline}(\this.{\ensuremath{\mathsf{cofn}}}(c~{\ensuremath{\mathsf{::}}}~tmp~{\ensuremath{\mathsf{::}}}~[x{\Mapsto}v]~\fv(x~{\to}{\subexp}_2')))
\\|~\forall~v_0~p_1,~{\ensuremath{\mathsf{if}}}~\phi_0
\end{array}\right.} &\hspace{-1.5em}\left.\vphantom{\left\{\def\arraystretch{1}\begin{array}{l}  p p1; q'' [who|->v0, state|->c]; \\  tmp <-_t this.clfn(c :: [x|=>v] fv(()->e1)); \\  tramp(this.cofn(c :: tmp :: [x|=>v] fv(x->e2'))) \\  | \forall v0 p1, if \phi_0 \end{array}\right\}}\right\}\\
\overset{\textcolor{red}{\textsc{Rbt}}}{~{\to}}
   & \vspace{5pt}{\left\{\def\arraystretch{1}\begin{array}{l}\arrayrulecolor{white!80!black}{\trace}~p_1~{\ensuremath{\mathsf{msg}}}(\textcolor{red}{v_0'},\textcolor{red}{v_2'})~{\ensuremath{\mathsf{wr}}}(0,{\ensuremath{\mathsf{sender}}},\textcolor{red}{v_0'});~{\syntrace}''~[{\ensuremath{\mathsf{who}}}{\mapsto}v_0,~{\ensuremath{\mathsf{state}}}{\mapsto}c];
\\{\ensuremath{\mathsf{let~}}}tmp=\textcolor{red}{v_2'};~{\trampoline}(\this.{\ensuremath{\mathsf{cofn}}}(c~{\ensuremath{\mathsf{::}}}~tmp~{\ensuremath{\mathsf{::}}}~[x{\Mapsto}v]~\fv(x~{\to}{\subexp}_2')))
\\|~\forall~v_0~p_1~\textcolor{red}{v_0'}~\textcolor{red}{v_2'},~{\ensuremath{\mathsf{if}}}~\phi_0
\end{array}\right.} &\hspace{-1.5em}\left.\vphantom{\left\{\def\arraystretch{1}\begin{array}{l}  p p1 msg(\textcolor{red}{v0'},\textcolor{red}{v2'}) wr(0,sender,\textcolor{red}{v0'}); q'' [who|->v0, state|->c]; \\  let tmp=\textcolor{red}{v2'}; tramp(this.cofn(c :: tmp :: [x|=>v] fv(x->e2'))) \\  | \forall v0 p1 \textcolor{red}{v0'} \textcolor{red}{v2'}, if \phi_0 \end{array}\right\}}\right\}\\
\overset{case~\textcolor{red}{v_0'}~=~v_0}{=}
   & \vspace{5pt}{\left\{\def\arraystretch{1}\begin{array}{l}\arrayrulecolor{white!80!black}{\trace}~p_1~{\ensuremath{\mathsf{msg}}}(\textcolor{red}{v_0'},\textcolor{red}{v_2'})~{\ensuremath{\mathsf{wr}}}(0,{\ensuremath{\mathsf{sender}}},\textcolor{red}{v_0'});~{\syntrace}''~[{\ensuremath{\mathsf{who}}}{\mapsto}v_0,~{\ensuremath{\mathsf{state}}}{\mapsto}c];
\\{\ensuremath{\mathsf{let~}}}tmp=\textcolor{red}{v_2'};~{\trampoline}(\this.{\ensuremath{\mathsf{cofn}}}(c~{\ensuremath{\mathsf{::}}}~tmp~{\ensuremath{\mathsf{::}}}~[x{\Mapsto}v]~\fv(x~{\to}{\subexp}_2')))
\\|~\forall~v_0~p_1~\textcolor{red}{v_0'}~\textcolor{red}{v_2'},~{\ensuremath{\mathsf{if}}}~\textcolor{red}{v_0'}~\ne~v_0,~\phi_0~\\\midrule{}~{\trace}~p_1~{\ensuremath{\mathsf{msg}}}(\textcolor{red}{v_0'},\textcolor{red}{v_2'})~{\ensuremath{\mathsf{wr}}}(0,{\ensuremath{\mathsf{sender}}},\textcolor{red}{v_0'});~{\syntrace}''~[{\ensuremath{\mathsf{who}}}{\mapsto}v_0,~{\ensuremath{\mathsf{state}}}{\mapsto}c];
\\{\ensuremath{\mathsf{let~}}}tmp=\textcolor{red}{v_2'};~{\trampoline}(\this.{\ensuremath{\mathsf{cofn}}}(c~{\ensuremath{\mathsf{::}}}~tmp~{\ensuremath{\mathsf{::}}}~[x{\Mapsto}v]~\fv(x~{\to}{\subexp}_2')))
\\|~\forall~v_0~p_1~\textcolor{red}{v_0'}~\textcolor{red}{v_2'},~{\ensuremath{\mathsf{if}}}~\textcolor{red}{v_0'}~=~v_0,~\phi_0
\end{array}\right.} &\hspace{-1.5em}\left.\vphantom{\left\{\def\arraystretch{1}\begin{array}{l}  p p1 msg(\textcolor{red}{v0'},\textcolor{red}{v2'}) wr(0,sender,\textcolor{red}{v0'}); q'' [who|->v0, state|->c]; \\  let tmp=\textcolor{red}{v2'}; tramp(this.cofn(c :: tmp :: [x|=>v] fv(x->e2'))) \\  | \forall v0 p1 \textcolor{red}{v0'} \textcolor{red}{v2'}, if \textcolor{red}{v0'} \ne v0, \phi_0 \} \cup \{ p p1 msg(\textcolor{red}{v0'},\textcolor{red}{v2'}) wr(0,sender,\textcolor{red}{v0'}); q'' [who|->v0, state|->c]; \\  let tmp=\textcolor{red}{v2'}; tramp(this.cofn(c :: tmp :: [x|=>v] fv(x->e2'))) \\  | \forall v0 p1 \textcolor{red}{v0'} \textcolor{red}{v2'}, if \textcolor{red}{v0'} = v0, \phi_0 \end{array}\right\}}\right\}\\
\overset{\substack{\textsc{Rlet},~\textsc{Rget},~\textsc{Rapp},~\textsc{Rt},~\\~\textsc{Rget},~\textsc{Rop},~\textsc{Rget},~\\~\textsc{Rget},~\textsc{Rop},~\textsc{Rop}}}{~{\to}^*}
   & \vspace{5pt}{\left\{\def\arraystretch{1}\begin{array}{l}\arrayrulecolor{white!80!black}{\trace}~p_1~{\ensuremath{\mathsf{msg}}}(\textcolor{red}{v_0'},\textcolor{red}{v_2'})~{\ensuremath{\mathsf{wr}}}(0,{\ensuremath{\mathsf{sender}}},\textcolor{red}{v_0'});~{\syntrace}''~[{\ensuremath{\mathsf{who}}}{\mapsto}v_0,~{\ensuremath{\mathsf{state}}}{\mapsto}c];
\\{\trampoline}({\ensuremath{\mathsf{assert}}}({\ensuremath{\mathsf{false}}});~\this.{\ensuremath{\mathsf{state}}}~{\ensuremath{\mathsf{:=}}}~0;~[x{\Mapsto}\textcolor{red}{v_2'},~x{\Mapsto}v]~{\subexp}_2'))
\\|~\forall~v_0~p_1~\textcolor{red}{v_0'}~\textcolor{red}{v_2'},~{\ensuremath{\mathsf{if}}}~\textcolor{red}{v_0'}~\ne~v_0,~\phi_0~\\\midrule{}{\trace}~p_1~{\ensuremath{\mathsf{msg}}}(\textcolor{red}{v_0'},\textcolor{red}{v_2'})~{\ensuremath{\mathsf{wr}}}(0,{\ensuremath{\mathsf{sender}}},\textcolor{red}{v_0'});~{\syntrace}''~[{\ensuremath{\mathsf{who}}}{\mapsto}v_0,~{\ensuremath{\mathsf{state}}}{\mapsto}c];
\\{\trampoline}({\ensuremath{\mathsf{assert}}}({\ensuremath{\mathsf{true}}});~\this.{\ensuremath{\mathsf{state}}}~{\ensuremath{\mathsf{:=}}}~0;~[x{\Mapsto}\textcolor{red}{v_2'},~x{\Mapsto}v]~{\subexp}_2')
\\|~\forall~v_0~p_1~\textcolor{red}{v_0'}~\textcolor{red}{v_2'},~{\ensuremath{\mathsf{if}}}~\textcolor{red}{v_0'}~=~v_0,~\phi_0
\end{array}\right.} &\hspace{-1.5em}\left.\vphantom{\left\{\def\arraystretch{1}\begin{array}{l}  p p1 msg(\textcolor{red}{v0'},\textcolor{red}{v2'}) wr(0,sender,\textcolor{red}{v0'}); q'' [who|->v0, state|->c]; \\  tramp(assert(false); this.state := 0; [x|=>\textcolor{red}{v2'}, x|=>v] e2')) \\  | \forall v0 p1 \textcolor{red}{v0'} \textcolor{red}{v2'}, if \textcolor{red}{v0'} \ne v0, \phi_0 \} \cup \{p p1 msg(\textcolor{red}{v0'},\textcolor{red}{v2'}) wr(0,sender,\textcolor{red}{v0'}); q'' [who|->v0, state|->c]; \\  tramp(assert(true); this.state := 0; [x|=>\textcolor{red}{v2'}, x|=>v] e2') \\  | \forall v0 p1 \textcolor{red}{v0'} \textcolor{red}{v2'}, if \textcolor{red}{v0'} = v0, \phi_0 \end{array}\right\}}\right\}\\
\overset{\text{def.~}\approx}{\approx}
   & \vspace{5pt}{\left\{\def\arraystretch{1}\begin{array}{l}\arrayrulecolor{white!80!black}{\trace}~p_1~{\ensuremath{\mathsf{msg}}}(\textcolor{red}{v_0'},\textcolor{red}{v_2'})~{\ensuremath{\mathsf{wr}}}(0,{\ensuremath{\mathsf{sender}}},\textcolor{red}{v_0'});~{\syntrace}''~[{\ensuremath{\mathsf{who}}}{\mapsto}v_0,~{\ensuremath{\mathsf{state}}}{\mapsto}c];
\\{\trampoline}({\ensuremath{\mathsf{assert}}}({\ensuremath{\mathsf{true}}});~\this.{\ensuremath{\mathsf{state}}}~{\ensuremath{\mathsf{:=}}}~0;~[x{\Mapsto}\textcolor{red}{v_2'},~x{\Mapsto}v]~{\subexp}_2')
\\|~\forall~v_0~p_1~\textcolor{red}{v_0'}~\textcolor{red}{v_2'},~{\ensuremath{\mathsf{if}}}~\textcolor{red}{v_0'}~=~v_0,~\phi_0
\end{array}\right.} &\hspace{-1.5em}\left.\vphantom{\left\{\def\arraystretch{1}\begin{array}{l}  p p1 msg(\textcolor{red}{v0'},\textcolor{red}{v2'}) wr(0,sender,\textcolor{red}{v0'}); q'' [who|->v0, state|->c]; \\  tramp(assert(true); this.state := 0; [x|=>\textcolor{red}{v2'}, x|=>v] e2') \\  | \forall v0 p1 \textcolor{red}{v0'} \textcolor{red}{v2'}, if \textcolor{red}{v0'} = v0, \phi_0 \end{array}\right\}}\right\}\\
\overset{\textsc{Rlet},~\textsc{Rset}}{~{\to}}
   & \vspace{5pt}{\left\{\def\arraystretch{1}\begin{array}{l}\arrayrulecolor{white!80!black}{\trace}~p_1~{\ensuremath{\mathsf{msg}}}(\textcolor{red}{v_0'},\textcolor{red}{v_2'})~{\ensuremath{\mathsf{wr}}}(0,{\ensuremath{\mathsf{sender}}},\textcolor{red}{v_0'});~{\syntrace}''~[{\ensuremath{\mathsf{who}}}{\mapsto}v_0,~{\ensuremath{\mathsf{state}}}{\mapsto}0];
\\{\trampoline}([x{\Mapsto}\textcolor{red}{v_2'},~x{\Mapsto}v]~{\subexp}_2')
\\|~\forall~v_0~p_1~\textcolor{red}{v_0'}~\textcolor{red}{v_2'},~{\ensuremath{\mathsf{if}}}~\textcolor{red}{v_0'}~=~v_0,~\phi_0
\end{array}\right.} &\hspace{-1.5em}\left.\vphantom{\left\{\def\arraystretch{1}\begin{array}{l}  p p1 msg(\textcolor{red}{v0'},\textcolor{red}{v2'}) wr(0,sender,\textcolor{red}{v0'}); q'' [who|->v0, state|->0]; \\  tramp([x|=>\textcolor{red}{v2'}, x|=>v] e2') \\  | \forall v0 p1 \textcolor{red}{v0'} \textcolor{red}{v2'}, if \textcolor{red}{v0'} = v0, \phi_0 \end{array}\right\}}\right\}\\
\overset{\textcolor{red}{v_0'}~=~v_0}{=}
   & \vspace{5pt}{\left\{\def\arraystretch{1}\begin{array}{l}\arrayrulecolor{white!80!black}{\trace}~p_1~{\ensuremath{\mathsf{msg}}}(v_0,\textcolor{red}{v_2'})~{\ensuremath{\mathsf{wr}}}(0,{\ensuremath{\mathsf{sender}}},v_0);~{\syntrace}''~[{\ensuremath{\mathsf{who}}}{\mapsto}v_0,~{\ensuremath{\mathsf{state}}}{\mapsto}0];
\\{\trampoline}([x{\Mapsto}\textcolor{red}{v_2'},~x{\Mapsto}v]~{\subexp}_2')
\\|~\forall~v_0~p_1~\textcolor{red}{v_2'},~{\ensuremath{\mathsf{if}}}~\phi_0
\end{array}\right.} &\hspace{-1.5em}\left.\vphantom{\left\{\def\arraystretch{1}\begin{array}{l}  p p1 msg(v0,\textcolor{red}{v2'}) wr(0,sender,v0); q'' [who|->v0, state|->0]; \\  tramp([x|=>\textcolor{red}{v2'}, x|=>v] e2') \\  | \forall v0 p1 \textcolor{red}{v2'}, if \phi_0 \end{array}\right\}}\right\}\\
\overset{\phi_1}{\approx}
   & \vspace{5pt}{\left\{\def\arraystretch{1}\begin{array}{l}\arrayrulecolor{white!80!black}{\trace}~p_1~{\ensuremath{\mathsf{msg}}}(v_0,\textcolor{red}{v_2'})~{\ensuremath{\mathsf{wr}}}(0,{\ensuremath{\mathsf{sender}}},v_0);~{\syntrace};
\\{\trampoline}([x{\Mapsto}\textcolor{red}{v_2'},~x{\Mapsto}v]~{\subexp}_2)
\\|~\forall~v_0~p_1~\textcolor{red}{v_2'},~{\ensuremath{\mathsf{if}}}~\phi_0
\end{array}\right.} &\hspace{-1.5em}\left.\vphantom{\left\{\def\arraystretch{1}\begin{array}{l}  p p1 msg(v0,\textcolor{red}{v2'}) wr(0,sender,v0); q; \\  tramp([x|=>\textcolor{red}{v2'}, x|=>v] e2) \\  | \forall v0 p1 \textcolor{red}{v2'}, if \phi_0 \end{array}\right\}}\right\}\\
\overset{\textsc{Rlet},~\textcolor{red}{\textsc{Rbs}}}{~{\gets}}
   & \vspace{5pt}{\left\{\def\arraystretch{1}\begin{array}{l}\arrayrulecolor{white!80!black}{\trace}~p_1;~{\syntrace};~{\trampoline}({\ensuremath{\mathsf{let~}}}x={\downarrow}_s(v_0,~()~{\to}[x{\Mapsto}v]~{\subexp}_1);~[x{\Mapsto}v]~{\subexp}_2)
\\|~\forall~v_0~p_1,~{\ensuremath{\mathsf{if}}}~\phi_0
\end{array}\right.} &\hspace{-1.5em}\left.\vphantom{\left\{\def\arraystretch{1}\begin{array}{l}  p p1; q; tramp(let x=↓_s(v0, ()->[x|=>v] e1); [x|=>v] e2) \\  | \forall v0 p1, if \phi_0 \end{array}\right\}}\right\}\\
\overset{\phi_0({\syntrace})}{~{\gets}^*}
   & \vspace{5pt}{\left\{\def\arraystretch{1}\begin{array}{l}\arrayrulecolor{white!80!black}{\trace};{\syntrace};~{\trampoline}({\ensuremath{\mathsf{let~}}}x={\downarrow}_s([x{\Mapsto}v]~{\subexp}_0,~()~{\to}[x{\Mapsto}v]~{\subexp}_1);~[x{\Mapsto}v]~{\subexp}_2)
\end{array}\right.} &\hspace{-1.5em}\left.\vphantom{\left\{\def\arraystretch{1}\begin{array}{l}  p;q; tramp(let x=↓_s([x|=>v] e0, ()->[x|=>v] e1); [x|=>v] e2) \end{array}\right\}}\right\}\\
\overset{\text{def.~}{\Mapsto}}{=}
   & \vspace{5pt}{\left\{\def\arraystretch{1}\begin{array}{l}\arrayrulecolor{white!80!black}{\trace};{\syntrace};~{\trampoline}([x{\Mapsto}v]~{\ensuremath{\mathsf{let~}}}x={\downarrow}_s({\subexp}_0,~()~{\to}{\subexp}_1);~{\subexp}_2)
\end{array}\right.} &\hspace{-1.5em}\left.\vphantom{\left\{\def\arraystretch{1}\begin{array}{l}  p;q; tramp([x|=>v] let x=↓_s(e0, ()->e1); e2) \end{array}\right\}}\right\}\\
\overset{\text{def.~}{\subexp}}{=}
   & \vspace{5pt}{\left\{\def\arraystretch{1}\begin{array}{l}\arrayrulecolor{white!80!black}{\trace};{\syntrace};~{\trampoline}([x{\Mapsto}v]~{\subexp})
\end{array}\right.} &\hspace{-1.5em}\left.\vphantom{\left\{\def\arraystretch{1}\begin{array}{l}  p;q; tramp([x|=>v] e) \end{array}\right\}}\right\}\\
\overset{\text{def.~}init_A}{=}
   & \vspace{5pt}{\left\{\def\arraystretch{1}\begin{array}{l}\arrayrulecolor{white!80!black}init_A({\syndecl};{\decl};~{\trampoline}([x{\Mapsto}v]~{\subexp}))
\end{array}\right.} &\hspace{-1.5em}\left.\vphantom{\left\{\def\arraystretch{1}\begin{array}{l}  init_A(b;d; tramp([x|=>v] e)) \end{array}\right\}}\right\}\\
\end{longtable}}

\subparagraph{Case} $ ~~~{\subexp}~=~x_0 $.

Let ${\trace};{\syntrace}$ be the trace and state produced by initializing ${\decl};{\syndecl}$ with $A$.

{\small
\[ ~init_A({\decl};{\syndecl})~=~{\trace};{\syntrace} \]
}

The traceset equality holds by definition of $comp$ and $init_A$.

{\small\begin{longtable}{CLL}
   & \vspace{5pt}{\left\{\def\arraystretch{1}\begin{array}{l}\arrayrulecolor{white!80!black}[x{\Mapsto}v]~init_A(comp({\decl};{\syndecl};~{\trampoline}({\subexp})))
\end{array}\right.} &\hspace{-1.5em}\left.\vphantom{\left\{\def\arraystretch{1}\begin{array}{l}  [x|=>v] init_A(comp(d;b; tramp(e))) \end{array}\right\}}\right\}\\
\overset{\text{def.~}{\subexp}}{=}
   & \vspace{5pt}{\left\{\def\arraystretch{1}\begin{array}{l}\arrayrulecolor{white!80!black}[x{\Mapsto}v]~init_A(comp({\decl};{\syndecl};~{\trampoline}(x_0)))
\end{array}\right.} &\hspace{-1.5em}\left.\vphantom{\left\{\def\arraystretch{1}\begin{array}{l}  [x|=>v] init_A(comp(d;b; tramp(x0))) \end{array}\right\}}\right\}\\
\overset{\text{def.~}comp}{=}
   & \vspace{5pt}{\left\{\def\arraystretch{1}\begin{array}{l}\arrayrulecolor{white!80!black}[x{\Mapsto}v]~init_A({\decl};{\syndecl};~comp({\trampoline}(x_0)))
\end{array}\right.} &\hspace{-1.5em}\left.\vphantom{\left\{\def\arraystretch{1}\begin{array}{l}  [x|=>v] init_A(d;b; comp(tramp(x0))) \end{array}\right\}}\right\}\\
\overset{\text{def.~}init_A}{=}
   & \vspace{5pt}{\left\{\def\arraystretch{1}\begin{array}{l}\arrayrulecolor{white!80!black}{\trace};{\syntrace};~[x{\Mapsto}v]~{\trampoline}(x_0)
\end{array}\right.} &\hspace{-1.5em}\left.\vphantom{\left\{\def\arraystretch{1}\begin{array}{l}  p;q; [x|=>v] tramp(x0) \end{array}\right\}}\right\}\\
\overset{\text{def.~}{\Mapsto}}{=}
   & \vspace{5pt}{\left\{\def\arraystretch{1}\begin{array}{l}\arrayrulecolor{white!80!black}{\trace};{\syntrace};~{\trampoline}([x{\Mapsto}v]~x_0)
\end{array}\right.} &\hspace{-1.5em}\left.\vphantom{\left\{\def\arraystretch{1}\begin{array}{l}  p;q; tramp([x|=>v] x0) \end{array}\right\}}\right\}\\
\overset{\text{def.~}{\subexp}}{=}
   & \vspace{5pt}{\left\{\def\arraystretch{1}\begin{array}{l}\arrayrulecolor{white!80!black}{\trace};{\syntrace};~{\trampoline}([x{\Mapsto}v]~{\subexp})
\end{array}\right.} &\hspace{-1.5em}\left.\vphantom{\left\{\def\arraystretch{1}\begin{array}{l}  p;q; tramp([x|=>v] e) \end{array}\right\}}\right\}\\
\overset{\text{def.~}init_A}{=}
   & \vspace{5pt}{\left\{\def\arraystretch{1}\begin{array}{l}\arrayrulecolor{white!80!black}init_A({\decl};{\syndecl};~{\trampoline}([x{\Mapsto}v]~{\subexp}))
\end{array}\right.} &\hspace{-1.5em}\left.\vphantom{\left\{\def\arraystretch{1}\begin{array}{l}  init_A(d;b; tramp([x|=>v] e)) \end{array}\right\}}\right\}\\
\end{longtable}}

\end{proof}

%% file: notes/lemma_dtlp_preserves.txt.tex
\bigskip{}
\begin{globallemma}[comp' preserves traces]
 \emph{$comp'$ is defined on programs.}
\emph{To evaluate a program, it is initialized with a set of clients $A$.}
\emph{$comp'$ preserves the traceset of (closed) programs $P$ for any set of clients $A$.}

\emph{For all definitions ${\syndecl}$, definitions ${\decl}$, terms ${\subexp}_0$,}
\[ \{~init_A(comp'({\decl};{\syndecl};~{\trampoline}({\subexp}_0)))~\} ~\approx~ ... ~\approx~ \{~init_A({\decl};{\syndecl}';~{\trampoline}({\subexp}_0))~\} \]
\end{globallemma}\begin{proof} By induction over term structure.
\subparagraph{Case} $ ~~P~=~({\decl};{\syndecl};~{\subexp}_0) $.

Intuitively, $comp'$ prepends the definitions ${\syndecl}$ with initial definitions for ${\ensuremath{\mathsf{clfn}}}$ and ${\ensuremath{\mathsf{cofn}}}$
that only contain ${\ensuremath{\mathsf{assert}}}({\ensuremath{\mathsf{false}}})$, such that $comp$ can be applied.

{\small
\[ {\syndecl}'~=~\left(\begin{array}{l}~{\ensuremath{\mathsf{@cl}}}~\this.{\ensuremath{\mathsf{clfn}}}~=~id~{\to}~{\ensuremath{\mathsf{assert}}}({\ensuremath{\mathsf{false}}});\\~{\ensuremath{\mathsf{@co}}}~\this.{\ensuremath{\mathsf{cofn}}}~=~id~{\to}~{\ensuremath{\mathsf{assert}}}({\ensuremath{\mathsf{false}}});\\~{\syndecl}\end{array}\right) \]
}

The lemma holds by definition of $comp'$, and Lemma \ref{lemma:comp}.

{\small\begin{longtable}{CLL}
   & \vspace{5pt}{\left\{\def\arraystretch{1}\begin{array}{l}\arrayrulecolor{white!80!black}init_A(comp'({\decl};{\syndecl};~{\trampoline}({\subexp}_0)))
\end{array}\right.} &\hspace{-1.5em}\left.\vphantom{\left\{\def\arraystretch{1}\begin{array}{l}  init_A(comp'(d;b; tramp(e0))) \end{array}\right\}}\right\}\\
\overset{\text{def.~}comp'}{=}
   & \vspace{5pt}{\left\{\def\arraystretch{1}\begin{array}{l}\arrayrulecolor{white!80!black}init_A(comp({\decl};{\syndecl}';~{\trampoline}({\subexp}_0))
\end{array}\right.} &\hspace{-1.5em}\left.\vphantom{\left\{\def\arraystretch{1}\begin{array}{l}  init_A(comp(d;b'; tramp(e0)) \end{array}\right\}}\right\}\\
\overset{\text{Lemma~}\ref{lemma:comp}}{\approx}
   & \vspace{5pt}{\left\{\def\arraystretch{1}\begin{array}{l}\arrayrulecolor{white!80!black}init_A({\decl};{\syndecl}';~{\trampoline}({\subexp}_0))
\end{array}\right.} &\hspace{-1.5em}\left.\vphantom{\left\{\def\arraystretch{1}\begin{array}{l}  init_A(d;b'; tramp(e0)) \end{array}\right\}}\right\}\\
\end{longtable}}

\end{proof}